\documentclass[aps,prl,10pt,twocolumn,superscriptaddress, longbibliography]{revtex4-2}

\usepackage{xpatch}
\makeatletter
\patchcmd{\@ssect@ltx}
    {\addcontentsline{toc}{#1}{\protect\numberline{}#8}}
    {}
    {}
    {}
\makeatother

\usepackage{amsmath,amssymb,amsfonts}
\usepackage{comment}
\usepackage{graphicx}
\usepackage[pdfusetitle,
 bookmarks=true,bookmarksnumbered=false,bookmarksopen=false,
 breaklinks=false,pdfborder={0 0 0},pdfborderstyle={},backref=false,colorlinks=true]
 {hyperref}
\hypersetup{
 linkcolor=magenta, urlcolor=blue, citecolor=blue}
\usepackage{booktabs}
\usepackage{braket}
\usepackage{bbm}
\usepackage{soul}
\usepackage{xcolor}
\usepackage{mathtools}
\usepackage{amsthm}
\usepackage{cleveref}

\usepackage{tikz}
\usetikzlibrary{positioning,calc}
\usetikzlibrary{arrows.meta}

\newtheorem{theorem}{Theorem}
\newtheorem{proposition}{Proposition}
\newtheorem{result}{Result}
\newtheorem{corollary}{Corollary}
\newtheorem{lemma}[theorem]{Lemma}
\newtheorem*{remark}{Remark}
\newtheorem{definition}{Definition}

\newcommand{\tr}{\mathrm{Tr}}
\newcommand{\Tr}{\mathrm{Tr}}
\newcommand{\id}{\mathrm{id}}

\newcommand{\calN}{\mathcal{N}}
\newcommand{\calP}{\mathcal{P}}
\newcommand{\calA}{\mathcal{A}}
\newcommand{\calD}{\mathcal{D}}
\newcommand{\calE}{\mathcal{E}}
\newcommand{\calX}{\mathcal{X}}

\newcommand{\mmw}[1]{{\color{purple} MMW: #1}}

\newcommand{\mpc}[1]{{\color{blue} MP: #1}}

\renewcommand{\thesection}{\Roman{section}}
\renewcommand{\thesubsection}{\Alph{subsection}}

\allowdisplaybreaks

\begin{document}

\title{Onset of superactivation of quantum capacity}

\author{Marco Parentin}
\thanks{These authors contributed equally.}
\affiliation{Politecnico di Torino, 10129 Torino, Italy}

\author{Bjarne Bergh}
\thanks{These authors contributed equally.}
\affiliation{Department of Applied Mathematics and Theoretical Physics, University of Cambridge, Cambridge CB3 0WA, United Kingdom}

\author{Nilanjana Datta}
\affiliation{Department of Applied Mathematics and Theoretical Physics, University of Cambridge, Cambridge CB3 0WA, United Kingdom}

\author{Mark~M. Wilde}
\affiliation{School of Electrical and Computer Engineering, Cornell University, Ithaca, New York 14850, USA}

\date{\today}

\begin{abstract}
Superactivation of quantum capacity is the phenomenon whereby two quantum channels, each with zero quantum capacity, 
can exhibit a strictly positive capacity when used in tandem.
In this work, we explore superactivation in the previously unexplored non-asymptotic regime of finitely many channel uses. We give a definition of finite-blocklength superactivation and propose numerical methods that can certify it. Then, focusing on the 50\% erasure and positive-partial-transpose channels considered in the original work on superactivation, we show that as few as 17 uses of the joint channel already enable qubit transmission with a fidelity unattainable by any number of uses of either channel alone, demonstrating a strong finite-blocklength form of superactivation and opening the door to experimental demonstration.
\end{abstract}

\maketitle
\flushbottom

\textit{Introduction.}---
Characterizing the fundamental limits of information transmission over noisy quantum channels is a central problem in quantum Shannon theory. If the information to be transmitted is quantum, then the relevant quantity for a noisy quantum channel $\calN$ is its quantum capacity, $Q(\calN)$, which quantifies the maximum rate at which quantum information can be reliably transmitted through it~\cite{shor1995scheme}. Despite its importance, the behavior of the quantum capacity remains poorly understood, as it is generally given by a \textit{regularised formula} involving the limit $n \to \infty$ over the number ($n$) of uses of the channel and is therefore intractable in general. Specifically, by the quantum capacity theorem~\cite{Schumacher1996,schumacher1996quantum,Lloyd1997,BarnumNielsenSchumacher1998,BarnumKnillNielsen2000, Shor2002, Devetak2005}:
\begin{equation}
\label{eq:LSWtheorem}
        Q(\mathcal{N}) =\lim_{n\to \infty} \frac{1}{n} I_c(\mathcal{N}^{\otimes n}),
\end{equation}
where $I_c(\mathcal{N})$ is the coherent information of $\mathcal{N}$ \cite{schumacher1996quantum} and $\mathcal{N}^{\otimes n}$ represents $n$ parallel uses of $\mathcal{N}$. In general, $I_c(\mathcal{N}^{\otimes n}) \neq n I_c(\mathcal{N})$, meaning that the limit in~\eqref{eq:LSWtheorem} is unavoidable. This is in contrast to the case of a classical channel, for which the capacity is given by a single-letter expression.

Such non-additivity is a common feature 
of most communication settings that employ quantum channels~\cite{Holevo2002,Schumacher1997, Schumacher1996, Lloyd1997, Shor2002, Devetak2005, cai2004quantum}, with the sole exception of entanglement-assisted information transmission~\cite{bennett1999entanglement, bennett2002entanglement}. Non-additivity implies that quantum correlations across channel inputs can provide an advantage for coding~\cite{hastings2009superadditivity, smith2008structured, Li2009private, divincenzo1998quantum, cubitt2015unbounded,Leditzky2023,Leditzky2023Platypus,BhaleraLeditzky2025}.

The most extreme manifestation of this phenomenon is the \textit{superactivation} of quantum capacity~\cite{Smith2008}: two quantum channels, each individually incapable of reliably transmitting any quantum information, can acquire a strictly positive quantum capacity when used in tandem. Smith and Yard proved the existence of such a pair: a positive-partial-transpose (PPT) channel $\mathcal{P}$ and a 50\% quantum erasure channel $\mathcal{A}$ both have zero quantum capacity; i.e.~$Q(\mathcal{P}) = Q(\mathcal{A}) = 0$, and yet $Q(\calP \otimes \calA) > 0$ for a certain choice of $\mathcal{P}$ and $\mathcal{A}$~\cite{Smith2008}. 
This startling discovery showed that the usefulness of a channel for quantum communication can depend not only on its intrinsic properties but also on the context in which it is used---a feature with no classical counterpart. It also implies that quantum capacity is a strongly non-additive quantity; i.e.~there is no hope of finding a characterization different from~\eqref{eq:LSWtheorem} that is additive. Subsequent works
have since expanded on the theme~\cite{Cubitt2011, Cubitt2012, Smith2011, Brandao2012, Chen2011, Wilde2012, Strelchuk2012, Shirokov2015, Lim2019,wu2025}, and superactivation is now recognized as one of the defining features of quantum Shannon theory~\cite{Wilde_book}.

However, all prior investigations of superactivation have been confined to the \textit{asymptotic} regime, in which an unlimited number of channel uses is assumed. One could therefore question the practical relevance of this effect, given the existing state of the literature. 

In this paper, we address this question by studying superactivation in the \textit{non-asymptotic} regime~\cite{ renner2006securityquantumkeydistribution,datta2009smooth}, where the number of channel uses is finite. Rather than the 
highest achievable rate of quantum communication, we focus on the smallest transmission error, as measured by the channel fidelity $F_c(\calN^{\otimes n}, d)$, which quantifies the maximum success probability that is achievable for transmitting a maximally entangled state of Schmidt rank $d$
with $n$ parallel uses of a channel $\calN$.  We establish upper bounds on the channel fidelity of the individual channels $\mathcal{P}^{\otimes m}$ and $\mathcal{A}^{\otimes m}$, together with a lower bound on the channel fidelity of $n$ uses of the joint channel $(\mathcal{P} \otimes \mathcal{A})^{\otimes n}$. We then define non-asymptotic superactivation through the threshold number of channel uses, $n$, at which the latter exceeds both of the former. 
A central goal is to make this threshold as small as possible, so as to establish whether superactivation could ever be observed in practice. Focusing on the most natural case of qubit transmission ($d = 2$) and combining an analytical channel reduction with an efficient numerical algorithm for lower-bounding the channel fidelity, we show that remarkably few channel uses---just 17---suffice to demonstrate non-asymptotic superactivation. As a byproduct, our approach yields an explicit coding protocol for the joint channel, opening the door to an experimental demonstration of superactivation.


\textit{Entanglement transmission.}---Given a quantum channel $\calN_{A \to B}$, the goal of (one-shot) entanglement transmission is to transfer entanglement with a reference system, in the form of a maximally entangled state of Schmidt rank $d$, 
through the channel $\mathcal{N}$ with the help of an encoding channel $\mathcal{E}$ and a decoding channel $\mathcal{D}$~\cite{Buscemi_2010}. Specifically, the aim is to maximize the fidelity of the output state with a target maximally entangled state $\Phi^d$; so the performance metric is the \textit{channel fidelity}, defined as:
\begin{equation}\label{eq:channel_fidelity}
    F_c(\calN, d) \coloneqq \max_{\calE, \calD} F\!\left(\Phi^d_{RB'},\, \omega_{RB'}\right),
\end{equation}
where $F(\rho, \sigma) \coloneqq \left\|\sqrt{\rho}\sqrt{\sigma}\right\|^2_1$ is the fidelity~\cite{Uhlmann1976} and 
\begin{equation}
\omega_{RB'} \coloneqq (\id_R \otimes (\calD_{B \to B'} \circ \calN_{A \to B} \circ \calE_{A' \to A}))\, (\Phi^d_{RA'}).
\end{equation}

The minimum achievable error is $\varepsilon^*_Q(\calN, d) \coloneqq 1 - F_c(\calN, d)$; in general, this will be larger than zero, and so this scheme is also referred to as an approximate quantum error correction protocol~\cite{schumacher2002approximate}. Suppose we are given $n$ identical instances of $\mathcal{N}$: by using correlated inputs and joint decoding over the parallel concatenation $\mathcal{N}^{\otimes n}$, the minimum achievable error can only decrease; i.e., $\varepsilon^*_Q(\calN^{\otimes n}, d) \leq \varepsilon^*_Q(\calN^{\otimes m}, d)$ for all $n \geq m$, and the inequality can be strict.

The $n$-shot (entanglement transmission) quantum capacity is defined as the maximum achievable rate over $\mathcal{N}^{\otimes n}$ with error at most $\varepsilon$:
\begin{equation}
\label{eq:n-shotquantumcapacity}
    Q^{n, \varepsilon}(\mathcal{N}) \coloneqq \sup_{d\in \mathbb{N}} \left\{\frac{\log d}{n}: \varepsilon^*_Q(\calN^{\otimes n}, d) \leq \varepsilon \right\},
\end{equation}
where here and throughout, $\log$ denotes the binary logarithm (i.e., $\log \equiv 
\log_2$).
The asymptotic quantum capacity is then defined in terms of $Q^{n, \varepsilon}(\mathcal{N})$ as
\begin{equation}
Q(\calN) \coloneqq \lim_{\varepsilon \to 0} \liminf_{n\to\infty} Q^{n,\varepsilon}(\calN).
\end{equation}


\textit{Prior work.}---Let us recall the original superactivation result by Smith and Yard~\cite{Smith2008}, so as to introduce the two channels of interest. In their pioneering work, they considered the following two channels:
\begin{enumerate}
   \item The $50\%$ quantum erasure channel $\calA$~\cite{grassl1997codes}, whose action on an input state consists of sending it unaltered, with probability $1/2$, to the output, and replacing it by an erasure flag orthogonal to all input states, with probability $1/2$.
   
   This channel is symmetric in the sense of~\cite{Smith_2008}, i.e., self-complementary; implying that $Q(\calA) = 0$ by the no-cloning theorem~\cite{Park1970Transition,wootters1982single}.
    \item A positive-partial-transpose (PPT) channel $\mathcal{P}$ is one whose output is a PPT state~\cite{horodecki1996necessary, horodecki1997separability} when acting on 
    one share of a bipartite state. Since these states are undistillable~\cite{horodecki1997separability}, $Q(\calP) = 0$ for every PPT channel $\mathcal{P}$.
\end{enumerate}
The key contribution of~\cite{Smith2008} was to establish a lower bound on the quantum capacity of the joint channel $\calP \otimes \calA$, having the form $Q(\calP \otimes \calA) \geq \frac{1}{2} I_p(\mathcal{P})$, where $I_p(\mathcal{P})$ is the private information~\cite{Devetak2005, cai2004quantum} of the channel $\mathcal{P}$ and quantifies its ability to send private classical messages reliably. The existence of PPT channels with positive private information, referred to as private Horodecki channels~\cite{Horodecki2005SecureKey, horodecki2008low}, enabled them to prove superactivation. They also exhibited an explicit pair of channels with input dimension $d_A = 4$ satisfying $Q(\calP \otimes \calA) > 0.01$ (see also~\cite[Section~8.3.2]{watrous2011tqi}). While useful for proving the possibility of the superactivation effect, the analysis of~\cite{Smith2008} does not clarify the structure of a quantum code for $\calP \otimes \calA$ and does not appear to be helpful for understanding the superactivation effect in the non-asymptotic regime.


A natural first approach for gaining a better understanding of superactivation in the non-asymptotic regime is to fix an allowed error tolerance $\varepsilon$ and compare the value of~\eqref{eq:n-shotquantumcapacity} for the joint channel $\mathcal{P}\otimes \mathcal{A}$ to the values of~\eqref{eq:n-shotquantumcapacity} for the two channels taken individually, 
knowing that we should recover the result of~\cite{Smith2008} in the limits as $n \to \infty $ and then as $\varepsilon \to 0$. Since evaluating $Q^{n,\varepsilon}(\mathcal{N})$ for a general channel $\mathcal{N}$ is computationally difficult, we instead rely on bounding it. Specifically, we evaluate a second-order lower bound on $Q^{n,\varepsilon}(\calP \otimes \calA)$ from~\cite{Tomamichel2016} and compare it with upper bounds on $Q^{n,\varepsilon}(\calP)$  and $Q^{n,\varepsilon}(\calA)$, from \cite[Lemma~1]{Tomamichel2016} and \cite[Corollary~2]{Kaur2021}.
Then we determine the value of $n$ for which the lower bound exceeds both upper bounds. Figure~\ref{fig:normal-approx} depicts various scenarios in which superactivation is predicted to take hold for finite~$n$.

An advantage of the second-order lower bound in~\eqref{eq:2nd-order-lower-bnd} is that it is easy to evaluate and can be used to predict when the superactivation effect occurs
for a given error $\varepsilon$ and blocklength $n$. A downside of this bound is that 
it gives no information about how to code over the joint channel in an experimental setting. Furthermore, it predicts that the number of channel uses required to observe superactivation is on the order of several thousands for an error $\varepsilon$ on the order of $0.1$, which might make it too difficult to observe in practice. In the Methods Section and SM, we explore an alternative approach based on error exponents for quantum communication \cite{berta2026tight}, but we find similarly that this approach predicts that several thousands of channel uses are required to demonstrate superactivation.

The preceding concerns motivate us to consider an alternative numerical approach. Rather than asking at what rate information can be transmitted, we can alternatively focus on transmission quality and define superactivation in terms of the channel fidelity~\eqref{eq:channel_fidelity}. Given an input dimension $d \in \mathbb{N}$, we say that the channels~$\mathcal{P}$ and $\mathcal{A}$ exhibit \emph{$n$-shot superactivation} for some threshold error $\varepsilon \in (0,1)$ if there exists an $n\in \mathbb{N}$ such that
\begin{align}
\max\left\{F_c(\mathcal{P}^{\otimes m},d),\,
F_c(\mathcal{A}^{\otimes m},d)\right\}
&\le 1-\varepsilon \quad \forall\, m\in\mathbb{N},\notag \\
F_c((\mathcal{P}\!\otimes\!\mathcal{A})^{\otimes n},d)
&>1-\varepsilon .\label{eq:n-shot-superact}
\end{align}
In words, there exists a coding protocol for the joint channel achieving a  fidelity that is unattainable by the individual channels, regardless of the number of times they are used. If the channels $\mathcal{P}$ and $\mathcal{A}$ exhibit $n$-shot superactivation we call the \emph{superactivation threshold} $n^*$ the smallest $n$ for which $n$-shot superactivation occurs.

 
\textit{Results.}---In preparation for our results on superactivation in the non-asymptotic regime,
we first establish an explicit pre- and post- processing of the joint channel $\mathcal{P}\otimes \mathcal{A}$ which yields a much simpler effective channel. These pre- and post-processing scheme is inspired by the argument in \cite{Smith2008}, and we show that this effective channel achieves the coherent information lower bound that was shown for $\mathcal{P}\otimes \mathcal{A}$ in \cite{Smith_2008}. We then use this effective channel to obtain lower bounds on achievable channel fidelity of the channel pair $\mathcal{P}\otimes \mathcal{A}$.

This is helpful for both analysis and numerical optimization because this joint channel has input and output dimensions $d_A = 16$ and $ d_B = 20$, respectively, which is too complex for both purposes. A consequence of the following proposition is that we can analyze a much simpler effective channel~$\widetilde{\calN}$, with $d_A = 2$ and $d_B = 4$, rather than the more complex joint channel $\mathcal{P}\otimes \mathcal{A}$.  
\begin{proposition}
\label{prop:equivalent_channel}
There exists pre- and postprocessing operations $(\widetilde{\calE}, \widetilde{\calD})$ for $\calP \otimes \calA$ such that the effective channel 
\begin{equation}
\widetilde{\calN} \coloneqq \widetilde{\calD} \circ (\calP \otimes \calA) \circ \widetilde{\calE}     
\end{equation}
is a qubit-input channel of the following form:
\begin{equation}\label{eq:equivalent_channel}
    \widetilde{\calN}(\rho) = \tfrac{1}{2}\,\rho \otimes \ket{0}\!\bra{0}_Z + \tfrac{1}{2}\,(\calX^p \circ \overline{\Delta})(\rho) \otimes \ket{1}\!\bra{1}_Z,
\end{equation}
where $\calX^p(\rho) \coloneqq (1-p)\,\rho + p\, \sigma_X \rho \sigma_X$ is a bit-flip channel 
with $p = 1/(1+\sqrt{2})$, $\overline{\Delta}(\rho) \coloneqq \tfrac{1}{2}\rho + \tfrac{1}{2} \sigma_Z\rho \sigma_Z$ is the completely dephasing channel, and $Z$ is a classical flag register. Furthermore, $\sigma_X$ and $\sigma_Z$ are the Pauli $X$ and $Z$ matrices, respectively.
\end{proposition}

The channel $\widetilde{\calN}$ is a qubit-input, (qubit $\otimes$ flag)-output channel: with probability $1/2$, the flag register $Z$ contains the value $0$ and the input qubit state $\rho$ passes through unaltered; with probability $1/2$, the flag register $Z$ contains the value $1$ and the input qubit state $\rho$ suffers dephasing followed by bit-flip noise. We prove Proposition~\ref{prop:equivalent_channel} in the supplementary material (SM) by explicit construction of the correction procedure. The intuition behind \eqref{eq:equivalent_channel} is that the code $(\calE^p, \calD^p)$ allows us to recover an identity channel in the cases when the erasure does not occur, while it leads to the noisy Pauli channel $\calX^p \circ \overline{\Delta}$ in the case of an erasure. As mentioned above, this reduction is helpful because it replaces the original high-dimensional joint channel with a much simpler qubit-input channel whose classical-quantum output structure is amenable to analysis and efficient numerical optimization.

\begin{figure*}[ht]
    \centering
    \includegraphics[width=\textwidth]{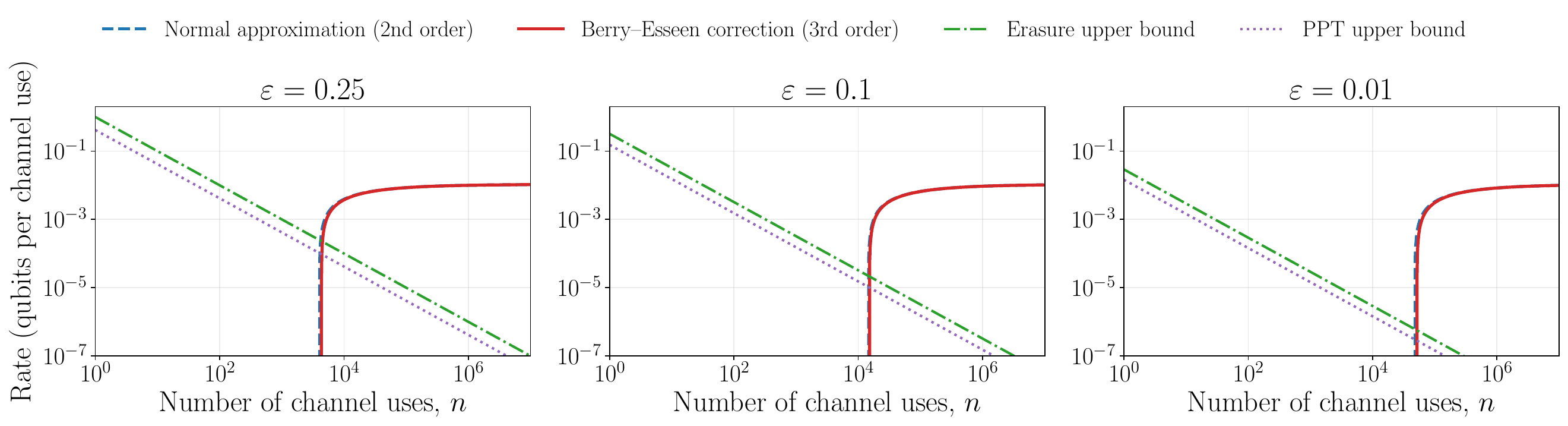}
    \caption{Comparison between lower bounds on the $n$-shot quantum capacity
    $Q^{n,\varepsilon}(\mathcal{P}\otimes \mathcal{A})$ in \eqref{eq:n-shotquantumcapacity}, obtained via the normal approximation \eqref{eq:2nd-order-lower-bnd} and the third-order correction from Berry-Esseen Theorem \eqref{eq:3nd-order-lower-bnd} and the upper bounds on $Q^{n,\varepsilon}(\mathcal{P})$ and $Q^{n,\varepsilon}(\mathcal{A})$  in \eqref{eq:upperbounds_two_channels}, as a function of the number of channel uses, $n$. Both axes use logarithmic scales. The lower bound exceeds both upper bounds at $n \sim 10^3$--$10^4$, certifying superactivation in the finite blocklength regime; however, the slow convergence motivates direct optimization methods at smaller $n$.}
    \label{fig:normal-approx}
\end{figure*}

We now state our primary result, showing that superactivation emerges in the $n$-shot regime for $\calA$ and $\calP$, in the sense of \eqref{eq:n-shot-superact}. In particular, focusing on the simplest and most relevant case of qubit transmission ($d = 2$), we show that very few channel uses are sufficient to achieve a fidelity over the joint channel that is unattainable by any number of uses of the two channels individually, 
leading to the first explicit protocol to test superactivation in practice.

\begin{result}\label{thm:main}
Let $\calP$ be the private Horodecki channel and~$\calA$ the $50\%$ erasure channel considered in
~\cite{Smith2008}. Then, the channel fidelities for qubit transmission satisfy the following bounds:
\begin{equation}
\label{eq:n-shot-superact-theor}
\begin{aligned}
\max\left\{F_c(\mathcal{P}^{\otimes m},2),\,
F_c(\mathcal{A}^{\otimes m},2)\right\}
&\le \frac{3}{4} \quad \forall\, m\in\mathbb{N},\\
F_c((\mathcal{P}\!\otimes\!\mathcal{A})^{\otimes 17},2)
&>\frac{3}{4} .
\end{aligned}
\end{equation}
In particular, $\calP$ and $\calA$ exhibit $n$-shot superactivation in the sense of~\eqref{eq:n-shot-superact} with $d = 2$, $\varepsilon = 1/4$, and superactivation threshold $n^* \leq 17$.
\end{result}

We argue for \eqref{eq:n-shot-superact-theor} as follows. 
For the PPT channel $\mathcal{P}$, the output state $\omega_{RB'}$ has positive partial transpose, so its fidelity with respect to the maximally entangled state (MES) is well-known to be bounded by $F(\omega_{RB'}, \Phi_{RB'}^d) \leq 1/d$~\cite{Horodecki1999}. This holds for any number of copies of $\mathcal{P}$ and implies that $F_c(\mathcal{P}^{\otimes m},2) \leq 1/2$ for $d = 2$. 

On the other hand, the 50\% erasure channel is known to be a two-extendible channel~\cite{kaur2019extendibility} (also known in the literature as `antidegradable' channel~\cite{Caruso2006,leditzky2018approaches}); i.e. its output state is always two-extendible~\cite{werner1989application, doherty2002distinguishing, doherty2004complete}, which means that there exists a state $\sigma_{RB_1 B_2}$ such that $B_1 \simeq B_2$ and $\tr_{B_2} \sigma_{RB_1 B_2} = \tr_{B_1}\sigma_{RB_1 B_2} = (\id_R \otimes \calA_{A \to B})(\rho_{RA})$ for every input state $\rho_{RA}$. Moreover, two-extendibility continues to hold for the tensor-power channel $\mathcal{A}^{\otimes m}$~\cite{doherty2002distinguishing} and is preserved by the local decoding operation~\cite{Nowakowski_2009}, so that the output state of the protocol, $\omega_{RB'}$, is two-extendible. The fidelity of a $k$-extendible state with the MES is well-known to be upper bounded by $F(\omega_{RB'}, \Phi_{RB'}^d) \leq (k+d-1)/(dk)$~\cite{johnson2013compatible}, which  yields the bound  $F_c(\mathcal{A}^{\otimes m},2) \leq 3/4$ for $k = d = 2$.

We now prove that it is possible to exceed these uniform 
upper bounds by using a coding strategy over the joint channel. For that purpose, we first use the channel~$\widetilde{\calN}$ 
obtained in Proposition~\ref{prop:equivalent_channel}. Since $\widetilde{\calN}$ is obtained by pre- and post-processing a single use of $\calP \otimes \calA$, it follows that
\begin{equation}\label{eq:fidelity_reduction}
    F_c((\calP \otimes \calA)^{\otimes n}, 2) \geq F_c(\widetilde{\calN}^{\otimes n}, 2) \quad \forall\, n \in \mathbb{N},
\end{equation}
by the data-processing inequality. It suffices to show that $F_c(\widetilde{\calN}^{\otimes n}, 2) > 3/4$ for some $n \in \mathbb{N}$. We obtain this via explicit encoder--decoder pairs constructed by the symmetric seesaw optimization described below. For $n = 17$, we provide an encoding / decoding protocol $(\calE^*_n, \calD ^*_n)$ for $\widetilde{\calN}^{\otimes n}$ that achieves an entanglement fidelity equal to $0.75013 > 3/4$. Since the bound is achieved by an explicit constructive protocol, the result is established.
\medskip

 Note that the two fidelity bounds provide a strong performance limitation for both channels 
 individually, showing that the success probability of transmitting an entangled qubit is always upper bounded by $3/4$; in particular, it is bounded away from $1$ (and indeed $Q(\calP) =Q(\calA) =0$). The superactivation threshold $n^* \leq 17$ identifies a concrete, finite number of channel uses beyond which the joint channel achieves a quality of quantum communication that is provably impossible with either channel individually, regardless of the number of uses. The following section describes the numerical method that we use to show this.

\begin{figure*}[ht]
    \centering
    \includegraphics[width=\textwidth]{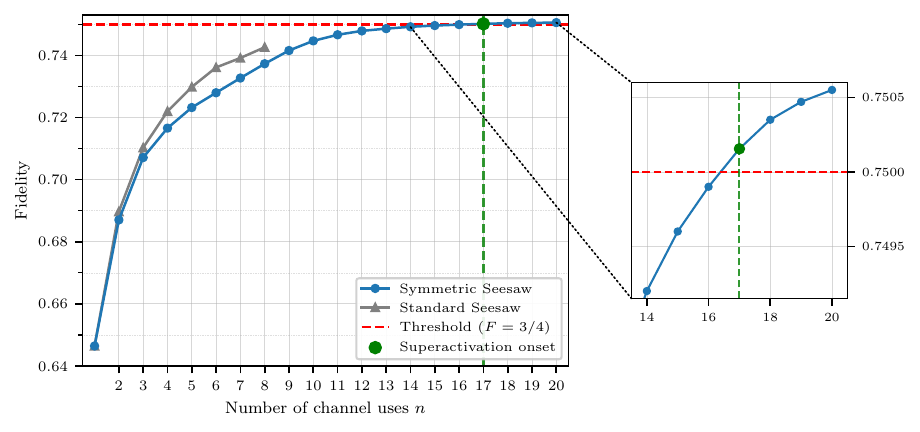}
    \caption{Lower bounds on the channel fidelity 
    $F_c(\widetilde{\calN}^{\otimes n}, 2)$, obtained via our new symmetric seesaw method,  
    as a function of the number of channel uses $n$. The gray curve shows the results of the standard seesaw method from \cite{Reimpell2005} applied to the same channel $\widetilde{\calN}^{\otimes n}$. The latter method's computational complexity is exponential in $n$ and hence infeasible for larger $n$. The dashed red line indicates the threshold $F = 3/4$, which cannot be surpassed 
    by either $\calP$ or $\calA$ individually for any number of uses 
    (Result~\ref{thm:main}). 
    The inset shows a magnified view of the region $n=14$–$20$, 
    where the crossing at $n=17$ becomes visible, 
    demonstrating non-asymptotic superactivation.}
    \label{fig:main}
\end{figure*}

\textit{Efficient Lower bounds on Channel Fidelity.}---Computing $F_c(\widetilde{\calN}^{\otimes n}, 2)$ for the effective
channel~\eqref{eq:equivalent_channel} is challenging for two reasons: the bilinear nature of the optimization and its exponential scaling with $n$. 
To address the first problem, by exploiting the linearity of the objective function in \eqref{eq:channel_fidelity}, one can obtain SDP-computable upper~\cite{Berta2022,Holdsworth_2023, CTV2023,kossmann2025approximatequantumerrorcorrection} and lower bounds~\cite{Reimpell2005,Fletcher2007,Taghavi_2010,johnson2017qvectoralgorithmdevicetailoredquantum}. In this work, since we are interested in finding an achievable coding strategy -- and thus a lower bound --  we use \textit{seesaw optimization techniques}, described in detail for quantum error correction first in~\cite{reimpell2008quantum} and revisited in the SM. In a nutshell, instead of finding the best encoder and decoder simultaneously in \eqref{eq:channel_fidelity}, we alternately find the optimal decoder and then the optimal encoder, improving the fidelity at each iteration up to saturation. The single optimizations at every step can be phrased as SDPs using the Choi representations~\cite{CHOI1975285} of the encoding map $\mathcal{E}_n$  and the decoding map $\mathcal{D}_n$ on $\widetilde{\mathcal{N}}^{\otimes n}$,
and they can be solved efficiently with standard solvers or by using an algorithm introduced in~\cite{Reimpell2005}, called the \textit{power iteration method}. 
To address the second issue---the exponential scaling with $n$---we use a new numerical method, which we call the \textit{symmetric seesaw method}~\cite{Bergh2026}. The idea is to restrict encoders and decoders to be permutation-invariant and exploit the representation theory of the symmetric group~$\mathfrak{S}_n$ to reduce the problem size, bringing the complexity down from exponential to polynomial in $n$~\cite{vallentin2009symmetry}. While this restriction yields a suboptimal solution in general, it allows us to reach significantly larger values of $n$, which is precisely what is needed to beat the threshold of $F=3/4$. 
This method is laid out in detail in \cite{Bergh2026}; we also include a brief summary in the SM.

The results of the numerical analysis, for $n = 1, 2, \ldots, 20$, are reported in Fig.~\ref{fig:main}. The channel fidelity grows steadily with $n$, crossing the threshold $F = 3/4$ at $n = 17$.

In order to reach the values of $n$ in Fig.~\ref{fig:main} using the symmetric seesaw method, we exploit the classical-quantum structure of~\eqref{eq:equivalent_channel}: specifically, for $n$ channel uses, the decoder 
first measures all $n$ classical flags, obtaining a binary string $z \in \{0,1\}^n$ with Hamming weight $k = |z|$ occurring with probability $\binom{n}{k}/2^n$. Conditioned on $k$, the effective qubit channel is $(\calX^p \circ \overline{\Delta})^{\otimes k} \otimes \id^{\otimes(n-k)}$ (up to a relabeling of positions), which has a reduced symmetry $\mathfrak{S}_k \times \mathfrak{S}_{n-k}$ rather than the full $\mathfrak{S}_n$. 
This splits the decoder SDP into $n+1$ independent, smaller subproblems, yielding a significant computational advantage. The details of this method are explained in the SM.


Since the values plotted are \textit{lower} bounds obtained from explicit encoding/decoding protocols, the true channel fidelity can only be higher, and so the actual superactivation threshold (i.e., the smallest $n$ for which superactivation occurs in the sense of~\eqref{eq:n-shot-superact}) satisfies $n^* \leq 17$. We emphasize that this threshold is fundamentally different from asymptotic capacity statements: it identifies a concrete, finite number of channel uses for which a constructive protocol provably outperforms what either channel alone could ever achieve. The symmetric seesaw method produces explicit encoder--decoder pairs achieving these fidelities, which are also permutation-invariant, and hence, 
in principle, easier to implement, providing a protocol that could be implemented experimentally. 
The protocol requires coherent manipulation of $n = 17$ qubits---a resource that, while demanding, is within the capabilities of current quantum hardware platforms, as demonstrated in neutral-atom systems controlling over 50 qubits~\cite{Bernien2017,Manetsch2025} and superconducting and trapped-ion platforms exceeding 20 qubits~\cite{Kam2024,Moses2023}.


\textit{Conclusion and outlook.}---This work provides the first systematic investigation of the superactivation phenomenon in the non-asymptotic regime, revealing structural features that are not apparent in the asymptotic limit. We have shown that, for the  channels considered in~\cite{Smith2008}, a finite and surprisingly small number of channel uses suffices to demonstrate superactivation---in the sense that the joint channel achieves a transmission fidelity unattainable by any number of uses of either individual channel. This transforms superactivation from a purely asymptotic property into a concrete, experimentally testable phenomenon.

Several natural questions emerge from our work. First, one may ask whether the superactivation threshold $n^*$ may be further 
lowered, e.g., by improving the numerical methods to obtain better lower bounds. Second, it would be valuable to perform a similar non-asymptotic analysis for other channel pairs known to exhibit superactivation in the asymptotic regime, such as bosonic Gaussian channels~\cite{Smith2011} or the channels presented in~\cite{Brandao2012}. Third, we suspect that the computational framework presented here, based on permutation-invariant codes, may find independent applications in finite-blocklength quantum coding beyond the superactivation context, for the study  of other channels such as the depolarizing channel.

\medskip
\textbf{Acknowledgments}.
The authors thank Karol Horodecki, Robert Salzmann, Satvik Singh, and Sergii Strelchuk for helpful discussions.
B.B.\ and N.d.\ are supported by the Engineering and Physical Sciences Research Council [Grant Ref: EP/Y028732/1]. MMW acknowledges support from the National Science Foundation under grant no.~2329662
and from the Cornell School of Electrical and Computer Engineering.

\let\oldaddcontentsline\addcontentsline
\renewcommand{\addcontentsline}[3]{}
\bibliography{references}
\let\addcontentsline\oldaddcontentsline

\section*{Data Availability}

The data and code used to generate the numerical results in this work, including the encoding and decoding channels achieving fidelity above $0.75$ for $n=17$, are available in the accompanying folder \texttt{anc}, included in our arXiv post. These resources allow independent verification of the results shown in Figure~\ref{fig:main}. 

The \texttt{anc} folder is organized as follows:
\begin{description}
    \item[\texttt{superactivation\_operators\_n17.npz}]. A NumPy archive containing the optimal symmetric encoder~$\mathcal{E}$, the decoders $(\mathcal{D}_k)_{k=0}^{17}$ (see \eqref{eq:outcomes_modified_seesaw}), and the intermediate operator blocks $\mathcal{M}_k = \mathcal{N}_k \circ \mathcal{E}$ (the concatenation of the encoder with the channels, as seen by the decoder), where $\mathcal{N}_k$ is defined in \eqref{eq:channel_k},  for $k \in \{0, \dots, 17\}$. All operators are stored in the Schur--Weyl block basis \eqref{eq:Shur-Weyl-decomposition}, as block-matrices of the form in \eqref{eq:generic_perm_inv_input}. This file also includes the fidelity, providing a self-contained dataset for verification.
        
    \item[\texttt{verify\_results.py}]. A standalone Python script to validate the results. It checks the physical validity of all operators (ensuring they correspond to valid quantum channels) and computes the entanglement fidelity directly from the stored blocks using formula \eqref{eq:SDP_k_block} to confirm non-asymptotic superactivation.
\end{description}

These files provide a complete, reproducible dataset for confirming the numerical demonstration of non-asymptotic superactivation presented in this work, together with the encoding and decoding maps that achieve it.  

\section*{Methods}

In this section, we provide more details of upper bounds and lower bounds on the $n$-shot capacity $Q^{n,\varepsilon}$ and the minimum achievable error $\varepsilon^*_Q$. We appeal to general formulas from the quantum information theory literature, in particular \cite{Tomamichel2016,Kaur2021,berta2026tight}, and we evaluate them for the channels $\mathcal{P}$, $\mathcal{A}$, and $\mathcal{P}\otimes \mathcal{A}$. Full details of our derivations are available in the supplementary material.

Let us begin by establishing upper bounds on the $n$-shot capacity $Q^{n,\varepsilon}$. Upper bounds on $Q^{n,\varepsilon}(\calP)$  and $Q^{n,\varepsilon}(\calA)$, from \cite[Lemma~1]{Tomamichel2016} and \cite[Corollary~2]{Kaur2021}, respectively, imply that
\begin{align}
\label{eq:upperbounds_two_channels}
    Q^{n,\varepsilon}(\calP) & \leq -\frac{1}{n}\log(1-\varepsilon), \\
    Q^{n,\varepsilon}(\calA) & \leq -\frac{1}{n}\log(1-2\varepsilon),
\end{align}
the former holding for all $\varepsilon \in [0,1)$ and the latter for all $\varepsilon \in [0,1/2)$. The first bound comes from the fact that $\calP$ is a PPT channel, and the second bound comes from the fact that $\calA$ is a two-extendible channel.

Let us now turn to establishing a lower bound on $Q^{n, \varepsilon}(\mathcal{P}\otimes \mathcal{A})$.
With the simplification from Proposition~\ref{prop:equivalent_channel} in place, we can evaluate the coherent information~\cite{schumacher1996quantum} and coherent information variance~\cite{Tomamichel2016} when sending one share of a Bell state through the effective channel $\widetilde{\mathcal{N}}$, leading to the  following lower bound  on the $n$-shot quantum capacity $Q^{n, \varepsilon}(\mathcal{P}\otimes \mathcal{A})$ of the joint channel $\mathcal{P}\otimes \mathcal{A}$:
\begin{equation}
    Q^{n, \varepsilon}(\mathcal{P}\otimes \mathcal{A}) \geq \iota + \sqrt{\frac{\nu}{n}} \Phi^{-1}(\varepsilon) + O\!\left(\frac{\log n}{n} \right),
    \label{eq:2nd-order-lower-bnd}
\end{equation}
where the coherent information $\iota$ and coherent information variance $\nu$ are given by
\begin{equation}
    \iota  \equiv \frac{1}{2} (1  - h(p)),\quad 
    \nu  \equiv \frac{1}{2}v(p) + \frac{1}{4}(1+h(p))^2, \label{eq:dispersion}
\end{equation}
with
\begin{equation}
p = 1/(1+\sqrt{2}),    
\end{equation}
and the binary entropy $h(p)$ and binary entropy variance~$v(p)$ are defined as
\begin{align}
h(p)& \coloneqq  -p\log p - \bar{p} \log \bar{p},\\
v(p) & \coloneqq p (-\log p -h(p))^2 
 + \bar{p}(-\log \bar{p} - h(p))^2,
\end{align}
where $\bar{p}\equiv 1-p$.
Furthermore,
\begin{equation}
\Phi^{-1}(\varepsilon )\coloneqq \sup\{y \in \mathbb{R}: \Phi(y) \leq \varepsilon  \}    
\end{equation}
is the inverse cumulative Gaussian distribution function, such that $\Phi^{-1}(\varepsilon )\leq 0$ for $\varepsilon \in (0,1/2]$ and $\Phi^{-1}(\varepsilon )> 0$ for $\varepsilon\in(1/2,1)$. 

The first-order term $\iota$ matches that from \cite{Smith2008} as an asymptotic rate of quantum communication over the joint channel $\mathcal{P}\otimes \mathcal{A}$, while the second-order term $\sqrt{\frac{\nu}{n}} \Phi^{-1}(\varepsilon)$ features the coherent information variance~$\nu$ and is novel in the context of superactivation. The second-order term gives a sense of how far below the asymptotic rate $\iota$ we must operate in order to meet the error constraint~$\varepsilon$ at finite blocklength $n$. The  coherent information variance $\nu$ quantifies
how much the channel output deviates from behaving like a constant-rate (deterministic) quantum communication channel. The induced channel $\widetilde{\mathcal{N}}$ is an equal probability, flagged mixture of an ideal channel and one that destroys quantum information. As such, at its core is a Bernoulli process with variance $1/4$, so we expect the coherent information variance $\nu$ to take a value strictly greater than $1/4$, as observed in~\eqref{eq:dispersion}. This variance needs to be overcome in the non-asymptotic regime in order for superactivation to occur, and it can be accomplished with sufficiently many channel uses.
We provide a detailed analysis in the Supplementary Material (SM) justifying~\eqref{eq:2nd-order-lower-bnd}, confirming that a strictly positive capacity of the joint channel is indeed achievable at finite blocklength. 

In order to establish an upper bound on the minimum achievable error $\varepsilon^*_Q((\mathcal{P}\otimes \mathcal{A})^{\otimes n},2)$, we can appeal to error exponent bounds for quantum communication, the most recent of which is available from \cite[Theorem~14]{berta2026tight}. Achievable error exponent bounds are results of the form $ \varepsilon_Q^* (\mathcal{N}^{\otimes n}, d) \leq  f(\mathcal{N}, n, d),$ where the function $f$ typically features an \textit{exponential} dependence on $n$ of the form
\begin{equation}
2^{-n\left(\frac{1}{n}g(\mathcal{N}^{\otimes n}) -r\right)},    
\end{equation}
where $g$ is a one-shot entropic quantity that depends on the channel and $r = \frac{\log d}{n}$ is a fixed rate. In \cite[Theorem~14]{berta2026tight}, the authors established the following $n$-shot bound, valid for every $s \in (0, 1)$ (equivalently, $\alpha = 1/(1+s) \in (1/2, 1)$):
\begin{equation}
\label{eq:error_bound_END_MATTER}
     \varepsilon^*_{Q}(\mathcal{N}^{\otimes n}, d)\leq 
     2\sqrt{\frac{s^s(1-s)^{1-s}}{s}} \cdot 2^{
     -\frac{s}{2}\left(I^c_{1/(1+s)}(\mathcal{N}^{\otimes n}) - \log d\right)},
\end{equation}
where
\begin{equation}
    I_\alpha^c(\mathcal{N}) \coloneqq \max_{\psi_{RA}} \inf_{\sigma_{B}}
    D_{\alpha}(\omega_{RB}\|\mathbbm{1}_{R}\otimes\omega_{B})
\end{equation}
is the \textit{Petz--R\'enyi coherent information of $\mathcal{N}$}, with
\begin{equation}
 \omega_{RB} \equiv (\id_R \otimes \mathcal{N}_{A \to B})(\psi_{RA})
\end{equation}
and
\begin{equation}
D_\alpha(\rho \|\sigma) \coloneqq \frac{1}{\alpha -1}
\log\!\left(\Tr[\rho^\alpha \sigma^{1-\alpha}]\right)
\end{equation}
(for $\mathrm{supp}(\rho) \subseteq \mathrm{supp}(\sigma)$) the
$\alpha$-\textit{Petz--R\'enyi relative entropy}~\cite{petz1985quasi}.
Evaluating this bound requires establishing a lower bound on the
Petz--R\'enyi coherent information of the joint channel
$\mathcal{P}\otimes \mathcal{A}$, which we prove is given by
\begin{multline}
\iota_\alpha \coloneqq \\ \frac{\alpha}{\alpha-1}
\log\!\left(\frac{1}{2}\left(\frac{1}{\left(1+\sqrt{2}\right)^{\alpha}}
+\left(\frac{2}{1+\sqrt{2}}\right)^{\alpha}\right)^{\frac{1}{\alpha}}
+2^{-\frac{1}{\alpha}}\right).
\end{multline}
Then, by appealing to~\eqref{eq:error_bound_END_MATTER}, we find the
following upper bound on the minimum achievable error
$\varepsilon^*_Q((\mathcal{P}\otimes \mathcal{A})^{\otimes n}, 2)$:
\begin{equation}
\label{eq:error_exponent_bound_superactivation_special_case_END_MATTER}
     \varepsilon^*_{Q}((\mathcal{P}\otimes \mathcal{A})^{\otimes n}, 2)
     \leq 2\sqrt{\frac{s^s(1-s)^{1-s}}{s}} \cdot 2^{
     -\frac{ns}{2}\left(\, \iota_{\frac{1}{1+s}} - \frac{1}{n}\right)},
\end{equation}
which one can optimize over $s \in (0,1)$ for each $n$. This bound shows that the error probability for qubit transmission indeed vanishes as $n \to \infty$, as we expect from the asymptotic result (the rate $\iota_\alpha$ for all $\alpha < 1$ is achievable in the limit $n \to \infty$), but the values of $n$ required to surpass the error threshold $1/4$, i.e.~the uniform lower bound for qubit transmission over an arbitrary number of copies of the $50\%$ erasure channel, are on the order of $n \approx 10^5$. This occurs because the exponent multiplying $n$ in
\eqref{eq:error_exponent_bound_superactivation_special_case_END_MATTER}
is extremely small (of order $10^{-5}$) and the prefactor
$\sqrt{s^s(1-s)^{1-s}/s}$ is non-negligible at the optimum, so the
bound trivializes unless $n$ becomes very large. This further
strengthens our main result of Theorem~\ref{thm:main}, which shows
that far fewer channel uses are required to surpass this bound. For a
more detailed discussion of error bounds, we refer to
Figure~\ref{fig:fidelity_bounds} and the surrounding section in the SM.

\clearpage
\onecolumngrid

\large

\begin{center}
\textbf{\large Supplemental Material}
\end{center}

\setcounter{secnumdepth}{3}

\setcounter{section}{0}
\setcounter{equation}{0}
\setcounter{figure}{0}
\setcounter{table}{0}

\renewcommand{\thesection}{\Roman{section}}
\renewcommand{\thesubsection}{\Alph{subsection}}

\renewcommand{\thesection}{S\arabic{section}}
\renewcommand{\thesubsection}{S\arabic{section}.\arabic{subsection}}
\renewcommand{\theequation}{\thesection.\arabic{equation}}
\renewcommand{\thefigure}{S\arabic{figure}}
\renewcommand{\thetable}{S\arabic{table}}

\makeatletter
\@addtoreset{equation}{section}
\makeatother

\tableofcontents

\section{Basic Notation}

\subsection{Quantum Channels: Choi representation, Concatenations}

Let $\mathcal{H}$ be a finite-dimensional complex Hilbert space of dimension $d_{\mathcal{H}}$, and let $\{\ket{i}_\mathcal{H}\}_{i = 1}^{d_\mathcal{H}}$ denote the  standard basis. We will also use the notation $[d_\mathcal{H}] \coloneqq \{1, \ldots, d_{\mathcal{H}}\}$. We denote by $\mathcal{L}(\mathcal{H})$ the set of linear operators acting on $\mathcal{H}$, by $\mathcal{P}(\mathcal{H})$ the set of positive semidefinite operators acting on $\mathcal{H}$, and by $\mathcal{D}(\mathcal{H}) \coloneqq \{\rho \in \mathcal{P}(\mathcal{H}): \Tr(\rho) = 1\}$ the set of quantum states, i.e.~density operators acting on $\mathcal{H}$. Pure states are density operators of unit rank, i.e.~of the form $\rho = \psi \equiv  |\psi\rangle \!\langle{\psi}|$ for some unit vector $\ket{\psi} \in \mathcal{H}$. Given two density operators $\rho$ and $ \sigma$, we define the fidelity as $F(\rho, \sigma) \coloneqq \left(\Tr[\sqrt{\sqrt{\sigma} \rho \sqrt{\sigma}}]\right)^2$~\cite{Uhlmann1976}. The fidelity is a measure of closeness between two states, and it simplifies to the inner product $F(\rho, \psi) = \bra{\psi}\rho \ket{\psi}$ when one of the two states is pure.

We denote by $\mathbbm{1}_{\mathcal{H}}$ the identity operator acting on $\mathcal{H}$ and by $\pi_{\mathcal{H}} \coloneqq \frac{\mathbbm{1}_{\mathcal{H}}}{d_{\mathcal{H}}}$ the maximally mixed state acting on $\mathcal{H}$. For two Hermitian operators $\rho, \sigma \in \mathcal{L} (\mathcal{H})$, we write $\sigma \geq \rho $ if $\sigma -\rho \in \mathcal{P}(\mathcal{H})$.  We will label different quantum systems by capital Roman letters, and in a communication setting, we denote by a letter $A$ the Hilbert spaces of systems that are in the possession of the sender (say, Alice) and by a letter $B$, the Hilbert spaces of the systems that are in the possession of the receiver (say, Bob). We write a bipartite operator as $X_{RA} \in \mathcal{P}(R \otimes A)$, where the subscripts explicitly indicate the two local subsystems. A quantum state $\rho_{AB}$ is separable if it can be written as a convex combination of product states on $A$ and $B$:
    \begin{equation}
        \rho_{AB} = \sum_{x} p_X(x) \rho^x_A \otimes \sigma^x_B,
    \end{equation}
    for some  probability distribution $p_X(x)$ and some sets of states $\{\rho^x_A\}_{x \in \mathcal{X}}$ and $\{\sigma^x_B\}_{x \in \mathcal{X}}$. A state that is not separable is called entangled. We say that $\rho_{AB}$ is a \textit{positive partial transpose} (PPT) state if $\rho_{AB}^{T_B} \geq 0$ \cite{Peres1996PPT,horodecki1996necessary}, where the superscript $T_B$ denotes a partial transpose (i.e.~transposition on $B$ alone). All separable states are PPT, but there exist PPT entangled states.

For a bipartite state $\rho_{RA}$, the marginal system of the subsystem $R$ is obtained by taking the partial trace $\rho_A = \Tr_R[ \rho_{RA} ] \coloneqq \sum_{i}(\mathbbm{1}_A \otimes \bra{i}_R) \rho_{AR}(\mathbbm{1}_A \otimes \ket{i}_R)$. We write $R \simeq A$ if $R$ is isomorphic to $A$. Given the standard basis $\{\ket{i}\}_{i = 1}^{d}$, we denote by $\Phi^d_{RA} \equiv  \ket{\Phi}\!\bra{\Phi}_{RA}$ the maximally entangled state (MES) of dimension $d$, defined as:\begin{equation}
    \ket{\Phi}_{RA} \coloneqq \frac{1}{\sqrt{d}} \sum_{i = 1}^d \ket{i}_R \ket{i}_A.
\end{equation}
We will also use the unnormalized maximally entangled state $\Gamma^d_{RA} \coloneqq d\Phi^d_{RA}$. Note that, given a bipartite state $\rho_{RA}$, the fidelity $F(\rho, \Phi^d) = \bra{\Phi^d}\rho\ket{\Phi^d}$ denotes the probability of passing an entanglement test of being equal to the maximally entangled state. Every PPT state $\rho$ has the property that $F(\rho, \Phi^d) \leq \frac{1}{d}$. This follows from~\cite[Lemma~2]{Rains1999Bound}:
\begin{equation}
\label{eq:PPT_singlet_fidelity}
    F(\rho, \Phi^d) = \Tr[\Phi^d \rho] = \frac{1}{d}\Tr[\mathbb{F} \rho^{T_B}] \leq \frac{1}{d} \left\|\rho^{T_B}\right\|_1 = \frac{1}{d},
\end{equation}
where we used the self-adjointness of $T_B$, the standard property $d \ket{\Phi^d}\!\bra{\Phi^d}^{T_B} = \mathbb{F}$, where $\mathbb{F}$ is the (unitary) flip operator, the variational characterization of the trace norm, and in the last step, the PPT property of $\rho$.

Given $n$ copies of the same Hilbert space $\mathcal{H}$, we use the notation $\mathcal{H}^{\otimes n} \coloneqq \mathcal{H} \otimes \mathcal{H} \otimes \cdots \otimes \mathcal{H}$, and we denote by $\{\ket{\underline{i}}_{\mathcal{H}^{\otimes n}}\}_{\underline{i} \in[d_{\mathcal{H}}]^{\times n}}$, i.e.~$\underline{i} \coloneqq  (i_1, i_2, \ldots, i_n)$ with $i_k \in [d_\mathcal{H}]$ for all $k = 1, \ldots, n$. 
\

We denote by $\operatorname{CP}(A\to B)$ the set of linear completely positive maps from $\mathcal{L}(A)$ to $\mathcal{L}(B)$. A quantum channel $\mathcal{N}_{A \to B}$ is a completely positive and trace-preserving map from $\mathcal{L}(A)$ to $\mathcal{L}(B)$, and we write $\mathcal{N}_{A \to B}\in \operatorname{CPTP}(A \to B)$.

Quantum channels describe the most general physical operation one can perform on quantum systems, both as a local operation (e.g., the encoding or decoding maps, respectively $\mathcal{E}$ and $\mathcal{D}$) or as a communication medium (denoted as~$\mathcal{N}$).  The ideal communication channel is represented by the identity channel, denoted as $\operatorname{id}_{A\to B}$. We will also denote the set of completely positive and unital maps as $\operatorname{CPU}(A \to B)$. Given a quantum channel $\mathcal{N}_{A \to B}\colon \mathcal{L}(\mathcal{A}) \to \mathcal{L}(B)$, we define its adjoint channel (with respect to the Hilbert--Schmidt inner product) as the map $\mathcal{N}^*_{B \to A}\colon \mathcal{L}(\mathcal{H}_B)\to \mathcal{L}(\mathcal{H}_A)$, such that $\Tr[Y \mathcal{N}(X)] = \Tr[\mathcal{N}^*(Y) X]$ for all $X \in  \mathcal{L}(\mathcal{H}_A)$ and $Y \in  \mathcal{L}(\mathcal{H}_B)$. One has the correspondence $  \mathcal{N}_{A \to B} \in \operatorname{CPTP}(A \to B) \iff \mathcal{N}^*_{B \to A} \in \operatorname{CPU}(B \to A)$, i.e.~unitality and trace preservation are dual notions. Choosing to use CPU maps instead of CPTP maps to describe quantum channels corresponds to selecting to work in the Heisenberg picture of quantum mechanics (linear transformations of observables) rather than the Schr\"odinger picture (linear transformations of states)~\cite{Wolf2012Channels}.

Among the several known characterizations of CPTP maps, in this work we will extensively make use of the Choi representation~\cite{CHOI1975285}. In particular, every quantum channel $\mathcal{N} \in \operatorname{CPTP}(A\to B)$ is uniquely described by its Choi state, defined, for $R \simeq A$, as $ \Phi^{\mathcal{N}}_{RB} \coloneqq (\operatorname{id}_R\otimes \mathcal{N}_{A\to B})(\Phi^{d_A}_{RA})$. 
The Choi matrix is the corresponding unnormalized operator $ \Gamma^{\mathcal{N}}_{RB} \coloneqq d_A  \Phi^{\mathcal{N}}_{RB}  =  (\operatorname{id}_R\otimes \mathcal{N}_{A\to B})(\Gamma^{d_A}_{RA})$. Notice in particular that $\Phi^{\operatorname{id}}_{RB} = \Phi^d_{RB}$ and $\Gamma^{\operatorname{id}}_{RB} = \Gamma^d_{RB}$ for $d_A = d_B = d$. By the Choi--Jamiolkowski isomorphism, the Choi state $\Phi^{\mathcal{N}}_{RB}$ contains all the information about the superoperator~$\mathcal{N}$. Specifically, its action on an arbitrary operator $X_{A} \in \mathcal{L}(\mathcal{H}_A)$ can be obtained as:\begin{equation}
\label{eq:choi_jamiolkovski_isomorphism}
    \mathcal{N}_{A \to B}(X_{A}) = \Tr_A[(X_{A}^{T} \otimes \mathbbm{1}_{B})\Gamma^{\mathcal{N}}_{AB}].
\end{equation} 
In the following, we will extensively use the following facts: \begin{equation}
\begin{aligned}
    & \Gamma^{\mathcal{N}}_{RB} \in \mathcal{P}(\mathcal{H}_{R} \otimes \mathcal{H}_{B}) \\
    &\Tr_{B}  \Gamma^{\mathcal{N}}_{RB} = \mathbbm{1}_R
\end{aligned}
\iff
\mathcal{N}_{A\to B} \in \mathrm{CPTP}(A \to B),
\end{equation}
\begin{equation}
\begin{aligned}
    & \Gamma^{\mathcal{N}}_{RB}  \in \mathcal{P}(\mathcal{H}_{R} \otimes \mathcal{H}_{B}) \\
    &\Tr_{R}  \Gamma^{\mathcal{N}}_{RB} = \mathbbm{1}_B
\end{aligned}
\iff
\mathcal{N}_{A\to B} \in \mathrm{CPU}(A \to B).
\end{equation}
This implies that the constraints of being a quantum channel in the Schr\"odinger or Heisenberg pictures are suited well for semidefinite optimization.

In communication settings, we will often consider serial concatenations, parallel concatenations, and adjoints of quantum channels in the Choi representation.

Given $\mathcal{N}_{A \to B}$ and $\mathcal{D}_{B \to C}$, their serial concatenation is denoted by $\mathcal{M}_{A \to C} = \mathcal{D}_{B \to C}\circ \mathcal{N}_{A \to B}$, and its Choi representation is as follows:
\begin{equation}
\label{eq:concatenation_choi}
   \Gamma^{\mathcal{M}}_{RC}  = \Tr_B[(\Gamma^{\mathcal{N}}_{RB})^{T_B} \cdot \Gamma^{\mathcal{D}}_{BC}],
\end{equation}
where we denoted by $T_B$ the partial transpose map over $B$, acting on $X_{AB}\in \mathcal{L}(AB)$ as $X_{AB}^{T_B} = \sum_{j,j'=1}^{d_B} (\mathbbm{1}_A \otimes \ket{j}_{B}\bra{j'}_B) X_{AB} (\mathbbm{1}_A \otimes \ket{j}_{B}\bra{j'}_B)$, and identities are clear from the context.

Given $\mathcal{M}_{A_1 \to B_1}$ and $\mathcal{N}_{A_2 \to B_2}$, the tensor-product channel $\mathcal{O}_{A_1A_2 \to B_1B_2} =\mathcal{M}_{A_1 \to B_1}\otimes \mathcal{N}_{A_2 \to B_2} $ has the Choi representation 
\begin{equation}
  \Gamma^{\mathcal{O}}_{A_1A_2B_1B_2} =\Gamma^{\mathcal{M}}_{A_1B_1}\otimes \Gamma^{\mathcal{N}}_{A_2B_2},
\end{equation}
where we left implicit a swap $B_1 \leftrightarrow A_2$ after taking the Kronecker product. Of particular interest is the case $\mathcal{N}^{\otimes n}_{A^n \to B^n}$, i.e.~the $n$-fold tensor product of $\mathcal{N}_{A \to B}$, where one has:\begin{equation}
\label{eq:choi_tensor_product}
            \Gamma^{\mathcal{N}^{\otimes n}}_{A^n B^n} =(\Gamma^{\mathcal{N}}_{AB})^{\otimes n}.
        \end{equation}
        Finally, if $\mathcal{N}^*_{B \to A}$ denotes the adjoint map of $\mathcal{N}_{A \to B}$, its Choi representation has the form:\begin{equation}
        \label{eq:adjoint_choi}
            \Gamma^{\mathcal{N}^*}_{BA} = (\Gamma^{\mathcal{N}}_{AB})^{T},
        \end{equation}
        where again a swap $A \leftrightarrow B$ is hidden after taking the transpose of the matrix (in the canonical basis of $\Phi$).
        
\subsection{Semidefinite Programming and Symmetry}

A semidefinite program (SDP) is a convex optimization problem in which one maximizes a linear objective function over the intersection of the cone of positive semidefinite matrices with an affine subspace~\cite{vandenberghe1996semidefinite}. More precisely, given Hermitian $C, D $ and a Hermiticity preserving map $\Psi$, an SDP takes the form:
\begin{equation}\label{SDP_definition}
\begin{aligned}
\textbf{Primal:} \quad & \max_{X \geq 0} \; \Tr[C X] && \textbf{Dual:} \quad & \min_{Y \geq 0} \; \Tr[D Y] \\
& \text{s.t.} \quad \Psi(X) \leq D && & \text{s.t.} \quad \Psi^*(Y) \geq C.
\end{aligned}
\end{equation}
The primal and dual optimal values satisfy $\alpha \leq \beta$ (weak duality), with equality under mild regularity conditions (strong duality, guaranteed, e.g., by Slater's condition). SDPs are efficiently solvable in polynomial time via interior-point methods; see~\cite{skrzypczyk2023semidefinite} for a review in the context of quantum information. All SDPs in this work satisfy strong duality.


In many cases of practical interest, an SDP is invariant under a group action, and one can restrict the optimization to the invariant subspace without loss of optimality~\cite{vallentin2009symmetry}. In this work, the relevant symmetry is the symmetric group~$\mathfrak{S}_n$, acting on $n$ copies of a Hilbert space $\mathcal{H}$ by permuting tensor factors:
\begin{equation}
    P(\pi) \cdot (\ket{v_1} \otimes \cdots \otimes \ket{v_n}) \coloneqq \ket{v_{\pi^{-1}(1)}} \otimes \cdots \otimes \ket{v_{\pi^{-1}(n)}} \qquad \forall\, \pi \in \mathfrak{S}_n,
\end{equation}
where $P(\pi)$ denotes the corresponding unitary (permutation matrix). A state $\rho \in \mathcal{D}(\mathcal{H}^{\otimes n})$ is \textit{permutation-invariant} if $P(\pi)\,\rho\,P(\pi)^\dagger = \rho$ for all $\pi \in \mathfrak{S}_n$. For every operator $\rho$, its symmetrization $\bar{\rho} \coloneqq \frac{1}{n!} \sum_{\pi \in \mathfrak{S}_n} P(\pi)\,\rho\,P(\pi)^\dagger$ is permutation-invariant by construction. The subspace of permutation-invariant operators is as follows:
\begin{equation}
    \mathrm{End}^{\mathfrak{S}_n}(\mathcal{H}^{\otimes n}) \coloneqq \bigl\{\rho \in \mathcal{L}(\mathcal{H}^{\otimes n}) : P(\pi)\,\rho\,P(\pi)^\dagger = \rho \;\; \forall\, \pi \in \mathfrak{S}_n \bigr\}.
\end{equation}
If an SDP of the form~\eqref{SDP_definition} is invariant under permutations (i.e., $P(\pi)\,X\,P(\pi)^\dagger$ is feasible with the same objective value whenever $X$ is), then an optimal solution exists in $\mathrm{End}^{\mathfrak{S}_n}(\mathcal{H}^{\otimes n})$ by convexity. We will often consider in this work multipartite operators $X_{RS^n} \in \mathcal{L}(R \otimes S^{\otimes n})$ for some system $S$ and a fixed reference $R$. A bipartite state $\rho_{RS} \in \mathcal{D}(R \otimes S)$ is called \textit{$n$-extendible} with respect to $S$~\cite{doherty2002distinguishing, doherty2004complete} if there exists a permutation-invariant extension $\sigma_{RS^n} \in \mathrm{End}^{\mathfrak{S}_n}(R \otimes S^{\otimes n})$ with $\sigma_{RS^n} \geq 0$ and $\mathrm{Tr}_{S^{n-1}}[\sigma_{RS^n}] = \rho_{RS}$. A quantum channel $\mathcal{N}$ is is called \textit{$n$-extendible}~\cite{pankowski2013entanglementdistillationextendiblemaps, kaur2019extendibility} if its Choi state $\Phi^{\mathcal{N}}_{RB}$ is $n$-extendible with respect to $B$, or equivalently if its output $\omega_{RB} \coloneqq(\mathrm{id}_R \otimes \mathcal{N}_{A \to B})(\rho_{RA})$ is $n$-extendible for every input state $\rho_{RA}$.

A quantum channel $\mathcal{N}_n \in \mathrm{CPTP}(A^n \to B^n)$ is \textit{permutation-covariant} if:
\begin{equation}
    \mathcal{N}_n \bigl(P_{A^n}(\pi)\,\rho\,P_{A^n}(\pi)^\dagger\bigr) = P_{B^n}(\pi)\,\mathcal{N}_n(\rho)\,P_{B^n}(\pi)^\dagger \qquad \forall\, \pi \in \mathfrak{S}_n.
\end{equation}
The simplest example is an $n$-fold tensor product $\mathcal{N}^{\otimes n}$, whose Choi matrix $(\Gamma^{\mathcal{N}})^{\otimes n}$ is permutation-invariant under simultaneous local permutations: $\Gamma^{\mathcal{N}^{\otimes n}}_{A^n B^n} \in \mathrm{End}^{\mathfrak{S}_n}(A^{\otimes n} \otimes B^{\otimes n})$.

\subsection{\texorpdfstring{Entanglement Transmission: Channel Fidelity and $n$-Shot Quantum Capacity}{Entanglement Transmission: Channel Fidelity and n-Shot Quantum Capacity}}

We now formally introduce the information-processing task of entanglement transmission~\cite{Buscemi_2010}, where the goal is to find a code, consisting of an encoding and a decoding channel, that enables a state entangled with a reference system to be reliably transmitted through a communication channel $\mathcal{N}$. We remark that this is a strong requirement, because reliable entanglement transmission implies reliable transmission, on average, of all unentangled input states~\cite{BarnumNielsenSchumacher1998}. Given a channel $\mathcal{N}$ and a parallel repetition of it, an $n$-shot entanglement transmission code is defined by a triplet $\{d, \mathcal{E}, \mathcal{D}\}$, where $d$ is the Schmidt rank of a maximally entangled state $\Phi^d_{RA'}$ that we want to transmit over the $n$-fold tensor product channel $\mathcal{N}^{\otimes n}$, while $\mathcal{E}_{A' \to A^n}$ and $\mathcal{D}_{B^n\to B'}$ are quantum channels called the encoding and the decoding operations, respectively. The complete scheme is shown in Figure~\ref{fig:quantumcommunicationquantumchannel}.

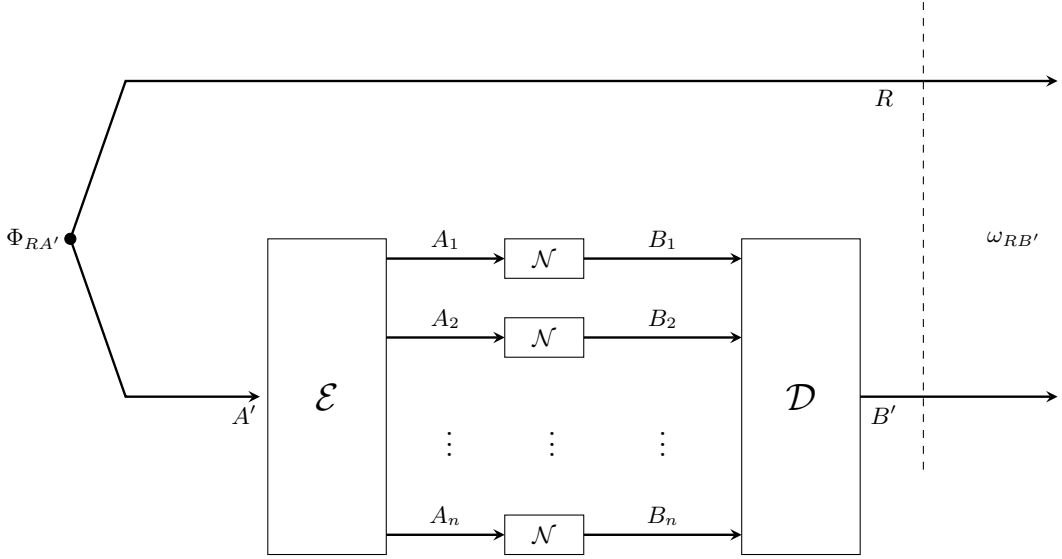
\begin{figure}[!ht]
\centering
\begin{tikzpicture}[>=stealth, scale=1.05]

\tikzset{
    thickarrow/.style={->, line width=1.5pt},
    qarrow/.style={->, line width=0.9pt},
    wire/.style={line width=0.9pt}
}

\draw (1,1) rectangle (2.5,5);
\node at (1.75,3) {\Large $\mathcal{E}$};

\draw (7,1) rectangle (8.5,5);
\node at (7.75,3) {\Large $\mathcal{D}$};

\foreach \y in {4.75,3.75,1.25} {
    \draw (4,\y-0.25) rectangle (5,\y+0.25);
    \node at (4.5,\y) {$\mathcal{N}$};
}
\node at (3.3,2.5) {\large \vdots};
\node at (4.6,2.5) {\large\vdots};
\node at (6,2.5) {\large\vdots};

\foreach \y/\lbl in {4.75/$A_1$,3.75/$A_2$,1.25/$A_n$} {
    \draw[qarrow] (2.5,\y) -- (4,\y) node[midway, above] {\lbl};
}

\foreach \y/\lbl in {4.75/$B_1$,3.75/$B_2$,1.25/$B_n$} {
    \draw[qarrow] (5,\y) -- (7,\y) node[midway, above] {\lbl};
}

\filldraw[black] (-1.5,5) circle (2pt);
\node[left] at (-1.5,5) {$\Phi_{RA'}$};

\draw[qarrow] (-1.5,5) -- (-0.8,7) -- (11,7);
\draw[qarrow] (-1.5,5) -- (-0.8,3) -- (0.9,3);
\node[below] at (8.8,7) {$R$};
\node[below] at (0.7,3) {$A'$};

\draw[qarrow] (8.5,3) -- (11,3);
\node[below] at (8.8,3) {$B'$};

\node[right] at (10,5) {$\omega_{RB'}$};
\draw[dashed] (9.3,8) -- (9.3,2);
\end{tikzpicture}
\caption{Entanglement transmission protocol over $n$ parallel uses of $\mathcal{N}$. The systems $R$, $A'$, and $B'$ are isomorphic and of dimension $d$. The encoder $\mathcal{E}_{A' \to A^n}$ encodes the part $A'$ of the maximally entangled state $\Phi^d_{RA'}$ into the channel input systems. Later, the decoder $\mathcal{D}_{B^n \to B'}$ recovers the state from the channel output systems. The performance of the code is measured using the fidelity $F(\omega_{RB'}, \Phi_{RB'})$.}
\label{fig:quantumcommunicationquantumchannel}
\end{figure}
The state at the end of the protocol is denoted as:\begin{equation}
\label{eq:output_state_protocol_et}
    \omega_{RB'} \equiv (\operatorname{id}_R\otimes (\mathcal{D}_{B^n\to B'} \circ \mathcal{N}^{\otimes n}_{A^n \to B^n}  \circ \mathcal{E}_{A'\to A^n}))(\Phi_{RA'}).
\end{equation}
We say that a triple $(n, d, \varepsilon )$ is an achievable pair for $n$-shot entanglement transmission over the channel $\mathcal{N}$ if there exists an $n$-shot entanglement transmission code with a rate of $\frac{\log d}{n}$, measured in qubits per channel use, satisfying:\begin{equation}
    F(\omega_{RB'}, \Phi_{RB'}) \geq 1-\varepsilon ,
\end{equation}
where $\varepsilon  \in (0,1)$ is the tolerated error and $F(\omega_{RB'}, \Phi_{RB'})$ can be seen as the probability of the output state passing an entanglement test.

In essence, for a given $n \in \mathbb{N}$, the goal of Alice and Bob is to devise a protocol $(\mathcal{E}_n, \mathcal{D}_n)$ that allows for making the overall channel $\mathcal{D}_{B^n \to B'} \circ \mathcal{N}_{A^n \to B^n} \circ \mathcal{E}_{A' \to A^n}$ as close as possible to $\operatorname{id}_{A' \to B'}$. This leads to the definition of channel fidelity.

\begin{definition}
    Given $n,d \in \mathbb{N}$ and a quantum channel $\mathcal{N}\in \operatorname{CPTP}(A \to B)$, with $A,B$ finite dimensional Hilbert spaces, with $d \leq \min\{d_A, d_B\}$, the channel fidelity is defined as the joint optimization problem:\begin{equation}
\label{eq:channel_fidelity_SM}
\begin{aligned}
    F_c(\mathcal{N}^{\otimes n}, d) 
    &\coloneqq \max_{\substack{\mathcal{E} \in \operatorname{CPTP}(A'\to A^n), \\ \mathcal{D} \in \operatorname{CPTP}(B^n\to B')}} 
    F\Bigl(\Phi_{RB'}^d, \bigl(\id_R \otimes (\mathcal{D}_{B^n \to B'} \circ \mathcal{N}_{A^n \to B^n} \circ \mathcal{E}_{A' \to A^n})\bigr)(\Phi^d_{RA'})\Bigr).
\end{aligned}
\end{equation}
\end{definition} 

The channel fidelity can also be interpreted in terms of the Choi representation:
\begin{equation}
\label{eq:reformulation_channel_fidelity}
     F_c(\mathcal{N}^{\otimes n}, d) \coloneqq \max_{\mathcal{E}_n, \mathcal{D}_n} 
    F\Bigl(\Phi^d_{RB'}, \Phi^{\mathcal{D}_n \circ \mathcal{N}^{\otimes n} \circ \mathcal{E}_n}_{RB'}\Bigr).
\end{equation}
The objective function in~\eqref{eq:reformulation_channel_fidelity} is also known in the literature as the entanglement fidelity~\cite{Schumacher1996} of the overall channel, defined as:\begin{equation}
\label{eq:entanglement_fidelity}
    F_e(\mathcal{M}) \coloneqq F\Bigl(\Phi^d, \Phi^{\mathcal{M}}\Bigr) = \bra{\Phi^d}\Phi^{\mathcal{M}}\ket{\Phi^d},
\end{equation}
where $d$ denotes the input and output dimension of $\mathcal{M}$. This quantifies the ability of the channel to preserve entanglement as the maximum probability of its Choi state passing an entanglement test (in line with the observation that $\Phi^{\id} = \Phi^d$). Conversely, the minimum achievable error in $n$-shot entanglement transmission is defined as:
\begin{equation}
\label{eq:minimium_chievable-error}
    \varepsilon ^*_Q(\mathcal{N}^{\otimes n}, d) \coloneqq 1 -F_c(\mathcal{N}^{\otimes n},d).
\end{equation}

\begin{definition} 
    Given $n \in \mathbb{N}$, $\varepsilon  \in (0,1)$ and a quantum channel $\mathcal{N}\in \operatorname{CPTP}(A \to B)$, with $A,B$ finite-dimensional Hilbert spaces, the $n$-shot, $\varepsilon$-error quantum capacity is defined to be the maximum achievable rate over $\mathcal{N}^{\otimes n}$ with error at most~$\varepsilon $:\begin{equation}
\label{eq:n-shotquantumcapacity_SM}
                Q^{n, \varepsilon }(\mathcal{N}) \coloneqq \sup_{d\in \mathbb{N}} \left\{\frac{\log_2d}{n}:\varepsilon ^*_Q(\mathcal{N}^{\otimes n}, d)\leq  \varepsilon  \right\}.
    \end{equation}
    The quantum capacity of $\mathcal{N}$ is defined as:\begin{equation}
    \label{eq:quantum_capacity}
        Q(\mathcal{N}) \coloneqq \lim_{\varepsilon  \to 0} \liminf_{n\to \infty} Q^{n, \varepsilon }(\mathcal{N}).
    \end{equation}
\end{definition}

While in the asymptotic setting the quantum capacity is the only quantity of interest (as $\varepsilon ^*_Q(\mathcal{N}^{\otimes n}, 2^{nR}) \xrightarrow{n \to \infty} 0$ for every achievable rate $R$), in the finite-blocklength setting one can focus both on $ Q^{n, \varepsilon }(\mathcal{N})$ or on $\varepsilon ^*_Q(\mathcal{N}^{\otimes n}, d)$. Choosing one or the other is a matter of deciding to focus on different aspects of information transmitted over $\mathcal{N}^{\otimes n}$, and both perspectives have found interest in the literature.

The quantum capacity theorem~\cite{Schumacher1996,schumacher1996quantum,Lloyd1997,BarnumNielsenSchumacher1998,BarnumKnillNielsen2000, Shor2002, Devetak2005} expresses the asymptotic quantum capacity $Q(\mathcal{N})$ in terms of its coherent information. Specifically, 
\begin{equation}
\label{eq:LSD}
    Q(\mathcal{N}) = \lim_{n \to \infty} \frac{1}{n} I_c(\mathcal{N}^{\otimes n}),
\end{equation}
where the coherent information of $\mathcal{N}$ is defined as:\begin{equation}
\label{eq:coherentInfoChannel}
I_c(\mathcal{N}) \coloneqq  \max_{\psi_{RA}} I_c(R\rangle B)_{\omega}
\end{equation}
and $I_c(R\rangle B)_{\omega} \coloneqq  S(B)_{\omega} - S(RB)_{\omega}$, where $S(B)_{\omega} \coloneqq  - \Tr[\omega_B\log \omega_B]$ is the von Neumann entropy. Also, we denoted $\omega_B \equiv \Tr_R [\omega_{RB}]$, where $\omega_{RB} \equiv (\mathrm{id}_R \otimes \mathcal{N}_{A \to B})(\psi_{RA})$, and the maximization is over pure states $\psi_{RA}$ with $d_R = d_A$. 

The coherent information is a measure of a channel’s ability to preserve quantum correlations. It is superadditive, $I_c(\mathcal{N}^{\otimes n}) \geq n I_c(\mathcal{N})$, and for some channels,  the inequality can be strict \cite{divincenzo1998quantum}. Thus we have in general that $Q(\mathcal{N}) \geq I_c(\mathcal{N}) \equiv Q^{(1)}(\mathcal{N})$, and the coherent information of $\mathcal{N}$ is only an achievable rate for reliable qubit transmission over $\mathcal{N}$ in the asymptotic limit.


\section{Pairs of channels exhibiting superactivation}

\subsection{Erasure and PPT Channels}

The qudit erasure channel of parameter $q$ is denoted as $\mathcal{A}^{q}\colon \mathcal{L}(\mathcal{H}_A) \to  \mathcal{L}(\mathcal{H}_B)$, with $d_{B} = d_A +1$, and it acts on an input by transmitting it unaltered with probability $1-q$ and by replacing it with a state $\ket{e}\!\bra{e}$ orthogonal to all input states with probability $q\in [0,1]$:
    \begin{equation}
    \label{eq:quantumerasurechannel}
       \mathcal{A}^{q}(\rho)\coloneqq (1-q) \rho + q \Tr(\rho)\ket{e}\!\bra{e}.
    \end{equation}
   An erasure can be perfectly detected by a projective measurement, but the original information is irretrievably lost in this case.  In this paper we will be concerned with the $50\%$ four-dimensional input erasure channel, i.e.~the case $q = \frac{1}{2}$, $d_A = 4$, and $d_B = 5$, which we simply denote by $\mathcal{A}$. This channel is symmetric; i.e.~it coincides with a channel complementary to it, as can be seen from the following isometric extension of it:\begin{equation}
    V_{A \to BE}^{\mathcal{A}} \ket{\psi}_A= \sqrt{\frac{1}{2}}\ket{\psi}_B\ket{e}_E + \sqrt{\frac{1}{2}}\ket{e}_B\ket{\psi}_E. 
\end{equation} 
Hence, the channel is two-extendible (also known as anti-degradable), and by the no-cloning theorem~\cite{Park1970Transition,wootters1982single}, its quantum capacity must vanish, $Q(\mathcal{A}) = 0$ \cite{bennett1997}. As we are interested in the non-asymptotic setting, the goal is to characterize how fast $Q^{n,\varepsilon }(\mathcal{A}) \to 0$ as $n\to \infty$  and $\varepsilon  \to 0$. The tightest known upper bound on the (unassisted) $n$-shot quantum capacity of an antidegradable channel (like~$\mathcal{A}$) is given in~\cite[Corollary~2]{Kaur2021} in the context of \textit{two-extendible assisted quantum communication} and establishes the following bound for all $\varepsilon \in [0,1/2)$:
\begin{equation}
\label{eq:outer_bound_erasure}
    Q^{n,\varepsilon }(\mathcal{A}) \leq  \frac{1}{n} \log\!\left(\frac{1}{1-2\varepsilon }\right).
\end{equation}

A PPT channel (also known in the literature as an entanglement-binding or Horodecki channel~\cite{Horodecki1999Binding}) is a channel $\mathcal{P}$ which, upon acting on one share of a bipartite state, outputs a so-called bound-entangled PPT state. The latter is a state from which two distant parties (who share the state) cannot distill Bell states via local quantum operations and classical communication (LOCC). Equivalently, its Choi matrix is PPT, i.e.~$(\Phi^{\mathcal{P}}_{RB})^{T_B} \geq 0$. Such a channel has zero quantum capacity, $Q(\mathcal{P}) = 0$. This can be seen from~\eqref{eq:PPT_singlet_fidelity}, as the output state of $\omega_{RB'}$ of every entanglement transmission protocol over these channels has fidelity with the MES upper bounded by $1/d$.


%
An important class of PPT channels is the set of entanglement breaking (EB) channels, whose Choi matrix is separable, but there also exist PPT channels which are not EB. An important example of those is the class of so-called \textit{private Horodecki channels}, which can be used to send classical private messages~\cite{horodecki2008low, Horodecki2005SecureKey}; i.e.~they are useful for cryptographic tasks. Specifically, this paper will use as a PPT channel the private Horodecki channel having the following Choi state (Eq.~(13) of \cite{horodecki2008low}):\begin{equation}
\Phi^{\mathcal{P}}_{RB} \coloneq \sum_{i=0}^{3}q_{i}\psi_{R'B'}^{i} \otimes\rho_{R''B''}^{(i)},
\label{eq:choi-state-horo}
\end{equation}
where $R \equiv R' R'',$ $B\equiv B' B''$, and we defined:
\begin{align}
|\psi^{0}\rangle  &  \coloneqq  \ket{\Phi^2} = \frac{1}{\sqrt{2}}\left(  |00\rangle+|11\rangle\right)  ,\\
|\psi^{1}\rangle &  \coloneqq \frac{1}{\sqrt{2}}\left(  |00\rangle
-|11\rangle\right)  ,\\
|\psi^{2}\rangle  &  \coloneqq \frac{1}{\sqrt{2}}\left(  |01\rangle+|10\rangle\right)  ,\\
|\psi^{3}\rangle  &  \coloneqq \frac{1}{\sqrt{2}}\left(  |01\rangle-|10\rangle\right)  ,
\end{align}
as the Bell basis states, and
\begin{align}
\rho^{(0)}  &  \coloneqq \frac{1}{2}\left[  \ket{00}\!\bra{00} +\psi^{2}\right]  ,\\
\rho^{(1)}  &  \coloneqq \frac{1}{2}\left[   \ket{11}\!\bra{11}+\psi^{3}\right]  ,\\
\rho^{(2)}  &  \coloneqq \chi^{+},\\
\rho^{(3)}  &  \coloneqq \chi^{-},\\
|\chi^{+}\rangle &  \coloneqq \frac{1}{2}\left(  \sqrt
{2+\sqrt{2}}|00\rangle+\sqrt{2-\sqrt{2}}|11\rangle\right)  ,\\
|\chi^{-}\rangle &  \coloneqq \frac{1}{2}\left(  \sqrt
{2-\sqrt{2}}|00\rangle-\sqrt{2+\sqrt{2}}|11\rangle\right)  ,\\
q_{0}  &  =q_{1}\coloneq \frac{1-p}{2} \coloneq \frac{1}{2} \left(  \frac{\sqrt{2}}{1+\sqrt{2}}\right), \\
q_{2}  &  =q_{3}\coloneq \frac{p}{2}=\frac{1}{2}\left(  \frac{1}{1+\sqrt{2}}\right),
\end{align}
where we defined
\begin{equation}
    p \coloneq \frac{1}{1+\sqrt{2}}.
\end{equation}
One can check that $\Phi^{\mathcal{P}}_{RB}$ is a valid Choi state of a quantum channel, i.e.~$\Phi^{\mathcal{P}}_{RB} \geq 0$ and $\Tr_B \Phi^{\mathcal{P}}_{RB} = \pi_R$. 

One can also verify that it is a PPT state, i.e.~$(\Phi^{\mathcal{P}}_{RB})^{T_B} \geq 0$, and its rank is equal to six; indeed, Smith and Yard \cite{Smith2008} provided the six Kraus operators, so to define the action of $\mathcal{P}$ as:
\begin{align}
\label{eq:ppt_channel}
    \mathcal{P}\colon \mathcal{L}(\mathcal{H}) &\to  \mathcal{L}(\mathcal{H}) \\
    \rho &\mapsto  \sum_{i=1}^6 K_i \rho K_i^\dag,
\end{align}
where $d_\mathcal{H} = 4$ and:
\begin{align}
K_1  &  \coloneqq \sqrt{\frac{1-p}{2}}\mathbbm{1} \otimes|0\rangle
\langle0|,\\
K_2  &  \coloneqq \sqrt{\frac{1-p}{4}}\mathbbm{1} \otimes \sigma_X,\\
K_3  &  \coloneqq \sqrt{\frac{1-p}{2}}\sigma_Z \otimes|1\rangle
\langle1|,\\
K_4  &  \coloneqq \sqrt{\frac{1-p}{4}}\sigma_Z\otimes \sigma_Y,\\
K_5  &  \coloneqq \sqrt{p}\sigma_X\otimes M_{0},\\
K_6  &  \coloneqq \sqrt{p}\sigma_Y\otimes M_1,
\end{align}
where \footnote{There is a sign typo in the coefficient of $|1\rangle\langle1|$ of $M_1$ in the Supporting Material of \cite{Smith2008}.}
\begin{align}
M_{0}  &  \coloneqq \frac{1}{2}\left(  \sqrt{2+\sqrt{2}}|0\rangle
\langle0|+\sqrt{2-\sqrt{2}}|1\rangle\langle1|\right)  ,\\
M_{1}  &  \coloneqq \frac{1}{2}\left(  \sqrt{2-\sqrt{2}}|0\rangle
\langle0|-\sqrt{2+\sqrt{2}}|1\rangle\langle1|\right).
\end{align}


One can verify using~\eqref{eq:ppt_channel} that indeed $(\mathrm{id}_R \otimes \mathcal{P}_{A \to B})(\Phi^4_{RA}) = \Phi^{\mathcal{P}}_{RB}$ defined in~\eqref{eq:choi-state-horo}, so that the above Kraus operators are in direct correspondence with the PPT Choi state in~\eqref{eq:choi-state-horo}.

Again, we are interested in tight upper bounds on $Q^{n,\varepsilon }(\mathcal{P})$. Given that $\omega_{RB'}$ in~\eqref{eq:output_state_protocol_et} is a PPT state, it belongs to the Rains set~\cite{Rains1999Bound, Audenaert2002Asymptotic} and therefore using the hypothesis testing Rains relative entropy bound of \cite[Lemma~1]{Tomamichel2016}, which holds for \textit{PPT-assisted quantum communication} (see also~\cite{Leung2015PPT,tomamichel2017}) we have:\begin{equation}
    \label{eq:outer_bound_PPT}
     Q^{n,\varepsilon }(\mathcal{P}) \leq  \frac{1}{n} \log\!\left(\frac{1}{1-\varepsilon }\right).
\end{equation}

Smith and Yard~\cite{Smith2008} proved that these two channels exhibit asymptotic superactivation of quantum capacity. In other words, the quantum capacity of the joint channel $\mathcal{A} \otimes \mathcal{P}$ is strictly larger than zero while the two single channels are individually useless for quantum communication (see Figure~\ref{fig:superactivation_formal}).

\begin{figure}[!ht]
\centering
\begin{tikzpicture}[>=stealth, scale=1.0]

\tikzset{
    qarrow/.style={->, line width=0.9pt},
    block/.style={draw, minimum width=1.6cm},
    decodeblock/.style={draw, minimum width=1.6cm},
    channel/.style={draw, minimum width=1.2cm, minimum height=0.5cm},
    chanlabel/.style={font=\normalsize}
}

\draw[block] (1,5.1) rectangle (2.6,8.1);
\node at (1.8,6.6) {\Large $\mathcal{E}_1$};

\draw[block] (1,1.9) rectangle (2.6,4.9);
\node at (1.8,3.4) {\Large $\mathcal{E}_2$};

\draw[decodeblock] (4.8,5.1) rectangle (6.4,8.1);
\node at (5.6,6.6) {\Large $\mathcal{D}_1$};

\draw[decodeblock] (4.8,1.9) rectangle (6.4,4.9);
\node at (5.6,3.4) {\Large $\mathcal{D}_2$};

\draw[block] (8.4,1.9) rectangle (10.0,8.1);
\node at (9.2,5.0) {\Large $\mathcal{E}$};

\draw[decodeblock] (12.0,1.9) rectangle (13.6,8.1);
\node at (12.8,5.0) {\Large $\mathcal{D}$};

\draw[channel] (3,7.6) rectangle (4.2,8.0);
\node[chanlabel] at (3.6,7.8) {$\mathcal{A}$};

\draw[channel] (3,7.0) rectangle (4.2,7.4);
\node[chanlabel] at (3.6,7.2) {$\mathcal{A}$};

\node at (3.6,6.3) {\Large$\vdots$};

\draw[channel] (3,5.2) rectangle (4.2,5.6);
\node[chanlabel] at (3.6,5.4) {$\mathcal{A}$};

\draw[channel] (3,4.4) rectangle (4.2,4.8);
\node[chanlabel] at (3.6,4.6) {$\mathcal{P}$};

\draw[channel] (3,3.8) rectangle (4.2,4.2);
\node[chanlabel] at (3.6,4.0) {$\mathcal{P}$};

\node at (3.6,3.1) {\Large$\vdots$};

\draw[channel] (3,2.0) rectangle (4.2,2.4);
\node[chanlabel] at (3.6,2.2) {$\mathcal{P}$};

\draw[channel] (10.4,7.6) rectangle (11.6,8.0);
\node[chanlabel] at (11.0,7.8) {$\mathcal{A}$};

\draw[channel] (10.4,7.0) rectangle (11.6,7.4);
\node[chanlabel] at (11.0,7.2) {$\mathcal{A}$};

\node at (11.0,6.3) {\Large$\vdots$};

\draw[channel] (10.4,5.2) rectangle (11.6,5.6);
\node[chanlabel] at (11.0,5.4) {$\mathcal{A}$};

\draw[channel] (10.4,4.4) rectangle (11.6,4.8);
\node[chanlabel] at (11.0,4.6) {$\mathcal{P}$};

\draw[channel] (10.4,3.8) rectangle (11.6,4.2);
\node[chanlabel] at (11.0,4.0) {$\mathcal{P}$};

\node at (11.0,3.1) {\Large$\vdots$};

\draw[channel] (10.4,2.0) rectangle (11.6,2.4);
\node[chanlabel] at (11.0,2.2) {$\mathcal{P}$};

\foreach \y in {7.8,7.2,5.4} {
    \draw[qarrow] (2.6,\y) -- (3,\y);
    \draw[qarrow] (4.2,\y) -- (4.8,\y);
}
\foreach \y in {4.6,4.0,2.2} {
    \draw[qarrow] (2.6,\y) -- (3,\y);
    \draw[qarrow] (4.2,\y) -- (4.8,\y);
}

\foreach \y in {7.8,7.2,5.4,4.6,4.0,2.2} {
    \draw[qarrow] (10.0,\y) -- (10.4,\y);
    \draw[qarrow] (11.6,\y) -- (12.0,\y);
}

\draw[qarrow] (0.2,6.8) -- (1,6.8);
\draw[qarrow] (0.2,3.1) -- (1,3.1);
\draw[qarrow] (6.4,6.8) -- (7.2,6.8);
\draw[qarrow] (6.4,3.1) -- (7.2,3.1);

\draw[qarrow] (7.6,5.2) -- (8.4,5.2);
\draw[qarrow] (13.6,5.2) -- (14.4,5.2);

\draw[dashed] (7.6,1.6) -- (7.6,8.3);

\end{tikzpicture}
\caption{Setting for superactivation. On the left, Alice and Bob attempt to separately use two zero-capacity channels $\mathcal{A}$ and $\mathcal{P}$ to transmit quantum information. Alice uses separate encoders $\mathcal{E}_1$ and $\mathcal{E}_2$ for each group of channels, and Bob uses separate decoders $\mathcal{D}_1$ and $\mathcal{D}_2$. Any attempt will fail because the quantum capacity of each channel is equal to zero. On the right, the same two channels are used in parallel for the same task. Alice’s encoder $\mathcal{E}$ now has simultaneous access to the inputs of all channels being used, and Bob’s decoding $\mathcal{D}$ is also performed jointly. In the case that the superactivation effect takes hold, noiseless quantum communication is possible in this setting.}
\label{fig:superactivation_formal}
\end{figure}
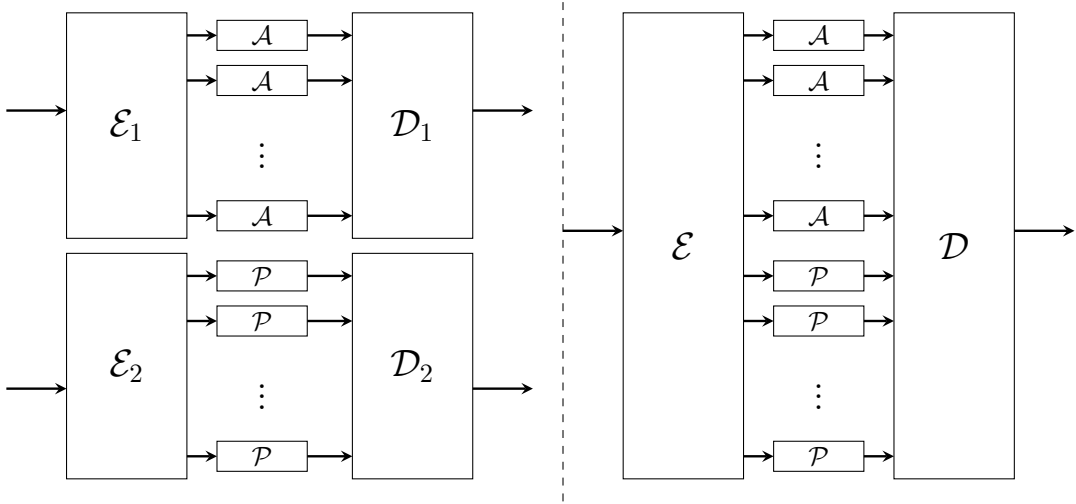
They showed this by establishing a lower bound on the coherent information of the joint channel $\mathcal{A}\otimes \mathcal{P}$ in terms of the private information of $\mathcal{P}$, which is strictly greater than zero. 

\subsection{Second-Order Analysis}
\label{sec:second-order-analysis-background}

One can attempt to obtain a non-asymptotic statement from these asymptotic results by studying second-order expansions of the quantum capacity, as given in~\cite{Tomamichel2016, wilde2017converse}. 

Second-order asymptotics concerns contributions to the rate that are proportional to $\frac{1}{\sqrt{n}}$ when we move from the $n$-shot to the asymptotic regime. In this context, the relevant entropic quantities can be systematically organized using the notion of relative entropy moments.

\subsubsection{Relative Entropy Moments}

For all $k\in\mathbb{N}$, we define the \textit{relative entropy $k$th moment} of a state $\rho$ and a positive semi-definite operator $\sigma$ as
\begin{equation}
M_{k}(\rho\Vert\sigma)\coloneqq\operatorname{Tr}[\rho\left(  \log\rho -\log\sigma\right)  ^{k}].
\end{equation}
Setting $k=1$ recovers the standard Umegaki relative entropy~\cite{umegaki1962}:
\begin{equation}
D(\rho \| \sigma) = M_1(\rho\|\sigma) = \Tr[\rho (\log \rho - \log \sigma)],
\end{equation}
which is well-defined if ${\rm{supp}}\, \rho \subseteq {\rm{supp}}\, \sigma$, and is equal to $+\infty$ else. The \textit{relative entropy variance}~\cite{li2014second,tomamichel2013hierarchy} is defined as
\begin{equation}
V(\rho\Vert\sigma)\coloneqq\operatorname{Tr}[\rho\left( \log\rho- \log\sigma-D(\rho\Vert\sigma)\right)  ^{2}] = M_{2}(\rho\Vert\sigma)-\left[  D(\rho\Vert\sigma)\right]^{2}.
\end{equation}

The following properties of relative entropy moments are generalizations of standard properties of the Umegaki relative entropy and will be used in the second-order analysis below.

It is well known that the relative entropy is invariant with respect to  isometries, namely:
\begin{equation}
    D(\rho \| \sigma) = D(\mathcal{U}(\rho)\Vert\mathcal{U}(\sigma)),
\end{equation}
for every isometric channel $\mathcal{U}(\cdot)\coloneqq U(\cdot)U^{\dag}$, where $U$ is an isometry.

\begin{lemma}
\label{lem:mom-iso-inv}
For every state $\rho$ and positive semidefinite operator $\sigma$ such that ${\rm{supp}}\, \rho \subseteq {\rm{supp}}\, \sigma$,
the relative entropy $k$th moment satisfies
\begin{equation}
M_{k}(\rho\Vert\sigma)=M_{k}(\mathcal{U}(\rho)\Vert\mathcal{U}(\sigma)),
\label{eq:spec-decomp-formula-moment-k}
\end{equation}
for all $k\in\mathbb{N}$ and every isometric channel $\mathcal{U}(\cdot)\coloneqq U(\cdot)U^{\dag}$, where $U$ is an isometry.
\end{lemma}

\begin{proof}
The formula in~\eqref{eq:spec-decomp-formula-moment-k} follows because
\begin{align}
M_{k}(\mathcal{U}(\rho)\Vert\mathcal{U}(\sigma)) &  = \operatorname{Tr}[\mathcal{U}(\rho)\left(  \log\mathcal{U}(\rho
)-\log\mathcal{U}(\sigma)\right)  ^{k}]\\
& = \operatorname{Tr}[\mathcal{U}(\rho)\left(  \mathcal{U}\left(  \log
\rho-\log\sigma\right)  \right)  ^{k}] \\
&  =\operatorname{Tr}\!\left[\mathcal{U}(\rho)\mathcal{U}\!\left(  \left(  \log
\rho-\log\sigma\right)  ^{k}\right)  \right]\\
&  =\operatorname{Tr}\!\left [(\mathcal{U}^\dag \circ\mathcal{U})(\rho)\left(  \left(  \log
\rho-\log\sigma\right)  ^{k}\right)  \right]\\
  &  =\operatorname{Tr}[\rho\left(  \log\rho
-\log\sigma\right)  ^{k}]\\
&  =M_{k}(\rho\Vert\sigma),
\end{align}
where we used the facts that $\mathcal{U}(\log A)=\log\mathcal{U}(A)$ on the range of $U$ for a positive semi-definite operator $A$, and that $\mathcal{U}(\rho)$ is supported on the range of $U$, and $\mathcal{U}(A^{k})=(\mathcal{U}(A))^{k}$.
\end{proof}

It is well known that the relative entropy satisfies:
\begin{equation}
    D(\rho \otimes \tau \Vert \sigma \otimes \tau) = D(\rho\Vert\sigma),
\end{equation}
for every state $\tau \in \mathcal{D}(\mathcal{H})$.

\begin{lemma}
\label{lem:kth-mom-prod-state}For states $\rho$ and $\tau$ and a positive
semi-definite operator $\sigma$, such that  ${\rm{supp}}\, \rho \subseteq {\rm{supp}}\, \sigma$, and for all $k\in\mathbb{N}$, the following holds:
\begin{equation}
M_{k}(\rho\otimes\tau\Vert\sigma\otimes\tau)=M_{k}(\rho\Vert\sigma).
\end{equation}
\end{lemma}

\begin{proof}
Consider that
\begin{align}
  M_{k}(\rho\otimes\tau\Vert\sigma\otimes\tau)  & =\operatorname{Tr}[\left(  \rho\otimes\tau\right)  \left(  \log\!\left(
\rho\otimes\tau\right)  -\log\!\left(  \sigma\otimes\tau\right)  \right)
^{k}]\\
&  =\operatorname{Tr}[\left(  \rho\otimes\tau\right)  \left(  \left[  \log
\rho-\log\sigma\right]  \otimes \mathbbm{1}\right)  ^{k}]\\
&  =\operatorname{Tr}\!\left[  \left(  \rho\otimes\tau\right)  \left(  \left[
\log\rho-\log\sigma\right]  ^{k}\otimes \mathbbm{1}\right)  \right] \\
&  =\operatorname{Tr}[\rho\left(  \log\rho-\log\sigma\right)  ^{k}]
  =M_{k}(\rho\Vert\sigma),
\end{align}
concluding the proof.
\end{proof}

Let $\rho_{XB}$ and $\sigma_{XB}$ be a classical--quantum state and positive
semi-definite operator, respectively, defined as%
\begin{align}
\rho_{XB}  &  \coloneqq\sum_{x}r(x)|x\rangle\!\langle x|_{X}\otimes\rho
_{B}^{x},\label{eq:cq-state-1}\\
\sigma_{XB}  &  \coloneqq\sum_{x}s(x)|x\rangle\!\langle x|_{X}\otimes
\sigma_{B}^{x}. \label{eq:cq-state-2}%
\end{align}
Then it is well known that%
\begin{equation}
D(\rho_{XB}\Vert\sigma_{XB})=D(r\Vert s)+\sum_{x}r(x)D(\rho_{B}^{x}\Vert
\sigma_{B}^{x}).
\end{equation}

\begin{lemma}
\label{lem:cq-moments}For a classical--quantum state $\rho_{XB}$ and positive
semi-definite operator $\sigma_{XB}$ of the form in
\eqref{eq:cq-state-1}--\eqref{eq:cq-state-2}\ and for all $k\in\mathbb{N}$,
the relative entropy $k$th moment can be written as follows:%
\begin{equation}
M_{k}(\rho_{XB}\Vert\sigma_{XB})=\sum_{i=0}^{k}\binom{k}{i}\sum_{x}r(x)\left[
\log_{2}\!\left(  \frac{r(x)}{s(x)}\right)  \right]  ^{i}M_{k-i}(\rho_{B}%
^{x}\Vert\sigma_{B}^{x}),
\end{equation}
where we note that%
\begin{align}
\left.  \binom{k}{i}\sum_{x}r(x)\left[  \log_{2}\!\left(  \frac{r(x)}%
{s(x)}\right)  \right]  ^{i}M_{k-i}(\rho_{B}^{x}\Vert\sigma_{B}^{x}%
)\right\vert _{i=0}  &  =\sum_{x}r(x)M_{k}(\rho_{B}^{x}\Vert\sigma_{B}^{x}),\\
\left.  \binom{k}{i}\sum_{x}r(x)\left[  \log_{2}\!\left(  \frac{r(x)}%
{s(x)}\right)  \right]  ^{i}M_{k-i}(\rho_{B}^{x}\Vert\sigma_{B}^{x}%
)\right\vert _{i=k}  &  =M_{k}(r\Vert s).
\end{align}
In the special case that $r=s$, it follows that
\begin{equation}
M_{k}(\rho_{XB}\Vert\sigma_{XB})=\sum_{x}r(x)M_{k}(\rho_{B}^{x}\Vert\sigma
_{B}^{x}). \label{eq:direct-sim-kth-moment}
\end{equation}

\end{lemma}

\begin{proof}
Consider that%
\begin{align}
  M_{k}(\rho_{XB}\Vert\sigma_{XB})
&  =\operatorname{Tr}\!\left[  \left(  \sum_{x^{\prime}}r(x^{\prime})|x^{\prime
}\rangle\!\langle x^{\prime}|_{X}\otimes\rho_{B}^{x^{\prime}}\right)  \left(
\begin{array}
[c]{c}%
\log_{2}\!\left(  \sum_{x}r(x)|x\rangle\!\langle x|_{X}\otimes\rho_{B}%
^{x}\right) \\
-\log_{2}\!\left(  \sum_{x}s(x)|x\rangle\!\langle x|_{X}\otimes\sigma_{B}%
^{x}\right)
\end{array}
\right)  ^{k}\right] \\
&  =\sum_{x^{\prime}}r(x^{\prime})\operatorname{Tr}\!\left[
\begin{array}
[c]{c}%
\left(  |x^{\prime}\rangle\!\langle x^{\prime}|_{X}\otimes\rho_{B}^{x^{\prime
}}\right)  \times\\
\left(  \sum_{x}|x\rangle\!\langle x|_{X}\otimes\left[  \log_{2}\!\left(
r(x)\rho_{B}^{x}\right)  -\log_{2}\!\left(  s(x)\sigma_{B}^{x}\right)  \right]
\right)  ^{k}%
\end{array}
\right] \\
&  =\sum_{x^{\prime}}r(x^{\prime})\operatorname{Tr}\!\left[
\begin{array}
[c]{c}%
\left(  |x^{\prime}\rangle\!\langle x^{\prime}|_{X}\otimes\rho_{B}^{x^{\prime
}}\right)  \times\\
\left(  \sum_{x}|x\rangle\!\langle x|_{X}\otimes\left[  \log_{2}\!\left(
r(x)\rho_{B}^{x}\right)  -\log_{2}\!\left(  s(x)\sigma_{B}^{x}\right)  \right]
^{k}\right)
\end{array}
\right] \\
&  =\sum_{x}r(x)\operatorname{Tr}\!\left[  \rho_{B}^{x}\left[  \log_{2}\!\left(
r(x)\rho_{B}^{x}\right)  -\log_{2}\!\left(  s(x)\sigma_{B}^{x}\right)  \right]
^{k}\right] \\
&  =\sum_{x}r(x)\operatorname{Tr}\!\left[  \rho_{B}^{x}\left[  \log_{2}\!\left(
\frac{r(x)}{s(x)}\right)  \mathbbm{1}_{B}+\log_{2}\rho_{B}^{x}-\log_{2}\sigma_{B}%
^{x}\right]  ^{k}\right] \label{eq:useful-cq-form-Mk}\\
&  =\sum_{x}r(x)\operatorname{Tr}\!\left[  \rho_{B}^{x}\sum_{i=0}^{k}\binom
{k}{i}\left[  \log_{2}\!\left(  \frac{r(x)}{s(x)}\right)  \right]  ^{i}\left[
\log_{2}\rho_{B}^{x}-\log_{2}\sigma_{B}^{x}\right]  ^{k-i}\right] \\
&  =\sum_{x}r(x)\sum_{i=0}^{k}\binom{k}{i}\left[  \log_{2}\!\left(  \frac
{r(x)}{s(x)}\right)  \right]  ^{i}\operatorname{Tr}\!\left[  \rho_{B}^{x}\left[
\log_{2}\rho_{B}^{x}-\log_{2}\sigma_{B}^{x}\right]  ^{k-i}\right] \\
&  =\sum_{x}r(x)\sum_{i=0}^{k}\binom{k}{i}\left[  \log_{2}\!\left(  \frac
{r(x)}{s(x)}\right)  \right]  ^{i}M_{k-i}(\rho_{B}^{x}\Vert\sigma_{B}^{x}),
\end{align}
concluding the proof.
\end{proof}

\begin{corollary}
\label{cor:variance-cq-state}
For a classical--quantum state $\rho_{XB}$ and positive semi-definite operator
$\sigma_{XB}$ of the form in~\eqref{eq:cq-state-1}--\eqref{eq:cq-state-2}, the
following holds:%
\begin{equation}
V(\rho_{XB}\Vert\sigma_{XB})=\sum_{x}r(x)V(\rho_{B}^{x}\Vert\sigma_{B}%
^{x})+\operatorname{Var}[W],
\end{equation}
where $W$ is a random variable that takes a value $\log_{2}\!\left(  \frac
{r(x)}{s(x)}\right)  +D(\rho_{B}^{x}\Vert\sigma_{B}^{x})$ with probability
$r(x)$. In particular, if $r=s$, then%
\begin{equation}
V(\rho_{XB}\Vert\sigma_{XB})=\sum_{x}r(x)V(\rho_{B}^{x}\Vert\sigma_{B}%
^{x})+\operatorname{Var}[Z],
\end{equation}
where $Z$ is a random variable that takes a value $D(\rho_{B}^{x}\Vert
\sigma_{B}^{x})$ with probability $r(x)$.
\end{corollary}

\begin{proof}
Applying~\eqref{eq:useful-cq-form-Mk} of Lemma~\ref{lem:cq-moments}, consider
that%
\begin{align}
  M_{2}(\rho_{XB}\Vert\sigma_{XB})
&  =\sum_{x}r(x)\operatorname{Tr}\!\left[  \rho_{B}^{x}\left(  \log_{2}\!\left(
\frac{r(x)}{s(x)}\right)  I_{B}+\left(  \log_{2}\rho_{B}^{x}-\log_{2}%
\sigma_{B}^{x}\right)  \right)  ^{2}\right] \\
&  =\sum_{x}r(x)\operatorname{Tr}\!\left[  \rho_{B}^{x}\left(
\begin{array}
[c]{c}%
\left[  \log_{2}\!\left(  \frac{r(x)}{s(x)}\right)  +D(\rho_{B}^{x}\Vert
\sigma_{B}^{x})\right]  I_{B}\\
+\left(  \log_{2}\rho_{B}^{x}-\log_{2}\sigma_{B}^{x}-D(\rho_{B}^{x}\Vert
\sigma_{B}^{x})\right)
\end{array}
\right)  ^{2}\right] \\
&  =\sum_{x}r(x)\operatorname{Tr}\!\left[  \rho_{B}^{x}\left(
\begin{array}
[c]{c}%
\left[  \log_{2}\!\left(  \frac{r(x)}{s(x)}\right)  +D(\rho_{B}^{x}\Vert
\sigma_{B}^{x})\right]  ^{2}I_{B}\\
+2\left[  \log_{2}\!\left(  \frac{r(x)}{s(x)}\right)  +D(\rho_{B}^{x}\Vert
\sigma_{B}^{x})\right]  \times\\
\left(  \log_{2}\rho_{B}^{x}-\log_{2}\sigma_{B}^{x}-D(\rho_{B}^{x}\Vert
\sigma_{B}^{x})\right) \\
+\left(  \log_{2}\rho_{B}^{x}-\log_{2}\sigma_{B}^{x}-D(\rho_{B}^{x}\Vert
\sigma_{B}^{x})\right)  ^{2}%
\end{array}
\right)  \right] \\
&  =\sum_{x}r(x)\left[  \log_{2}\!\left(  \frac{r(x)}{s(x)}\right)  +D(\rho
_{B}^{x}\Vert\sigma_{B}^{x})\right]  ^{2}+\sum_{x}r(x)V(\rho_{B}^{x}%
\Vert\sigma_{B}^{x}).
\end{align}
Recall that%
\begin{equation}
D(\rho_{XB}\Vert\sigma_{XB})=\sum_{x}r(x)\left[  \log_{2}\!\left(  \frac
{r(x)}{s(x)}\right)  +D(\rho_{B}^{x}\Vert\sigma_{B}^{x})\right]  .
\end{equation}
Observing the first term of $M_{2}(\rho_{XB}\Vert\sigma_{XB})$ is the second
moment of a random variable that takes a value $\log_{2}\!\left(  \frac
{r(x)}{s(x)}\right)  +D(\rho_{B}^{x}\Vert\sigma_{B}^{x})$ with probability
$r(x)$ and $D(\rho_{XB}\Vert\sigma_{XB})$ is its expectation, then we find
that
\begin{align}
  V(\rho_{XB}\Vert\sigma_{XB})
&  =M_{2}(\rho_{XB}\Vert\sigma_{XB})-\left[  D(\rho_{XB}\Vert\sigma
_{XB})\right]  ^{2}\\
&  =\sum_{x}r(x)\left[  \log_{2}\!\left(  \frac{r(x)}{s(x)}\right)  +D(\rho
_{B}^{x}\Vert\sigma_{B}^{x})\right]  ^{2}+\sum_{x}r(x)V(\rho_{B}^{x}%
\Vert\sigma_{B}^{x})\nonumber\\
&  \qquad-\left[  D(\rho_{XB}\Vert\sigma_{XB})\right]  ^{2}\\
&  =\sum_{x}r(x)\left[  \log_{2}\!\left(  \frac{r(x)}{s(x)}\right)  +D(\rho
_{B}^{x}\Vert\sigma_{B}^{x})-D(\rho_{XB}\Vert\sigma_{XB})\right]
^{2}\nonumber\\
&  \qquad+\sum_{x}r(x)V(\rho_{B}^{x}\Vert\sigma_{B}^{x}).
\end{align}
This concludes the proof.
\end{proof}

To capture the third-order contribution in finite blocklength expansions, we follow~\cite{Nussbaum_2009, li2014second}. Let $\rho, \sigma \in \mathcal{D}(\mathcal{H})$ have spectral decompositions $\rho = \sum_x p(x)\ket{\psi^x}\!\bra{\psi^x}$ and $\sigma = \sum_y q(y)\ket{\phi^y}\!\bra{\phi^y}$ for some orthonormal bases $\{\ket{\psi^x}\}_x$ and $\{\ket{\phi^y}\}_y$. If $(X,Y)$ denotes the pair of random variables on alphabet $\{(x,y)\}_{x,y=1}^{d_\mathcal{H}}$ with joint distribution $p_{X,Y}(x,y) = p(x)\,|\!\braket{\phi^y|\psi^x}\!|^2$, then \cite[Lemma~3]{li2014second}:\begin{equation}
    D(\rho\|\sigma) = \mathbb{E}_{(X,Y)}\left[\log\left(\frac{p(X)}{q(Y)}\right)\right],
\end{equation}
and\begin{equation}
\label{eq:variance-relative-entropy-spectrum}
    V(\rho\|\sigma) = \mathrm{Var}_{(X,Y)}\left[\log\left(\frac{p(X)}{q(Y)}\right)\right].
\end{equation}

We define the \textit{third absolute moment of the relative entropy} as
\begin{equation}
\label{eq:third_order_term_definition}
    T^3(\rho \| \sigma) \coloneqq \mathbb{E}_{(X,Y)}\!\left|\log\!\left( \frac{p(X)}{q(Y)}\right) - D(\rho \| \sigma)  \right|^3.
\end{equation}
Note that, unlike $D$ and $V$, this quantity cannot be expressed as a polynomial of the moments $\{M_k\}$ due to the absolute value.

\subsubsection{Second-Order Expansion of Coherent Information}

For quantum communication tasks, we define the \textit{coherent information $k$th moment} of a bipartite state $\sigma_{RB}$ as
\begin{equation}
M_k(R\rangle B)_\sigma \coloneqq M_k(\sigma_{RB} \| \mathbbm{1}_R \otimes \sigma_B),
\end{equation}
where $\sigma_B = \Tr_R[\sigma_{RB}]$. Setting $k=1$ recovers the coherent information $I(R\rangle B)_\sigma$, while the \textit{coherent information variance} is
\begin{align}
\label{eq:coherent_entropy_variance}
V(R \rangle B)_\sigma & \coloneqq  V(\sigma_{RB}\|\mathbbm{1}_R \otimes \sigma_B) = M_2(R\rangle B)_\sigma - [I(R\rangle B)_\sigma]^2.
\end{align}
The \textit{third absolute moment of the coherent information} is defined from~\eqref{eq:third_order_term_definition} analogously:
\begin{equation}
\label{eq:third-order-coh-absolute-moment}
T^3(R\rangle B)_\sigma \coloneqq T^3(\sigma_{RB} \| \mathbbm{1}_R \otimes \sigma_B).
\end{equation}
At the channel level, the channel's coherent information variance~\cite{Tomamichel2016} is:
\begin{equation}
    V_c^\varepsilon (\mathcal{N}) \coloneqq \begin{cases}\min_{\psi_{RA} \in \Pi }V(R\rangle B)_{\sigma} & \text{if } \varepsilon  \leq \frac{1}{2},    \\
    \max_{\psi_{RA}\in \Pi }V(R\rangle B)_{\sigma} & \text{if } \varepsilon  \geq \frac{1}{2},
    \end{cases}
\end{equation}
where the set $\Pi \subseteq \mathcal{D}(\mathcal{H}_{RA})$ contains all states achieving the maximum in the coherent information~\eqref{eq:coherentInfoChannel}. Then, one has:
\begin{equation}
\label{eq:lowerBound_analytic_n_shot}
  Q^{\varepsilon,n}(\mathcal{N}) \geq  I_c(\mathcal{N}) + \sqrt{\frac{V_c^\varepsilon (\mathcal{N})}{n}} \Phi^{-1}(\varepsilon) + \mathcal{O}\!\left(\frac{\log n}{n}\right),
\end{equation}
where $\Phi^{-1}(\varepsilon)\coloneqq \sup\{y \in \mathbb{R}: \Phi(y) \leq \varepsilon \}$  is the inverse cumulative Gaussian distribution function (recall that  $\Phi(x) \coloneqq \frac{1}{\sqrt{2\pi}} \int_{-\infty}^x e^{-t^2/2}\, dt$).

As a preliminary analysis of superactivation, an initial strategy would be to compare~\eqref{eq:lowerBound_analytic_n_shot} with the upper bounds~\eqref{eq:outer_bound_erasure} and~\eqref{eq:outer_bound_PPT}, to understand for a given error threshold $\varepsilon$, what is the value of $n$ for which the former exceeds both of the latter.

To apply~\eqref{eq:lowerBound_analytic_n_shot}, we need to explicitly write the third-order contribution hidden in the term $\mathcal{O}\!\left(\frac{\log n}{n}\right)$.  The channel's third absolute coherent information moment is:
\begin{equation}
    T_c^3(\mathcal{N}) \coloneqq \min_{\psi_{RA} \in \Pi} T^3(R \rangle B)_\sigma.
\end{equation}
Using the Berry--Esseen theorem following the analysis of~\cite{li2014second}, the right-hand side can be rewritten as:
\begin{equation}
\label{eq:Berry_Esseen_bound}
I_c(\mathcal{N}) + \sqrt{\frac{V_c^\varepsilon (\mathcal{N})}{n}} \Phi^{-1}\!\left(\varepsilon - \frac{1}{\sqrt{n}}\frac{C\, T_c^3(\mathcal{N})}{V_c^\varepsilon(\mathcal{N})^{3/2}}\right),
\end{equation}
where $C \in [0.40973,\, 0.4784]$ is a constant~\cite{korolev2012improvement}.

We highlight that, since~\eqref{eq:Berry_Esseen_bound} is derived starting from second-order asymptotics and the Berry--Esseen theorem, it is applicable for sufficiently large $n$ (i.e., $n= O(1/\sqrt{\varepsilon})$), but it can become inapplicable for small values of $n$. More precisely, the bound requires
\begin{equation}
\varepsilon  - \frac{1}{\sqrt{n}}
\frac{C T^3}{V_c^\varepsilon \!\left(\mathcal{N}\right)^{3/2}} \ge 0,
\end{equation}
so that the argument of the inverse cumulative distribution function is non-negative. We should also point out that~\eqref{eq:lowerBound_analytic_n_shot}  does not give an explicit encoding/decoding protocol for the joint channel.
In our case, the channel $\mathcal{N}$ in~\eqref{eq:lowerBound_analytic_n_shot} is the tensor product $\mathcal{P}\otimes \mathcal{A}$, and in the coherent information~\eqref{eq:coherentInfoChannel}, we maximize over pure states $\psi_{RA}$ with $d_R = d_A = 16$; similarly for the second and third-order coherent information moments ($V^\epsilon_c(\mathcal{P}\otimes \mathcal{A})$ and $T_c(\mathcal{P}\otimes \mathcal{A})$). Given the large dimensions involved in the underlying optimization problems, we need to come up with a simplifying protocol that reduces the complexity of the evaluation of the entropic quantities.


\section{Protocol for Effective Channel}
\label{sec:effective_protocol}

\begin{figure}[h!]
\centering
\begin{tikzpicture}[>=stealth, scale=1.3]
\tikzset{
    qarrow/.style={->, line width=1pt},
    block/.style={draw, minimum width=1.2cm, minimum height=4.0cm},
    procblock/.style={draw, minimum width=1.0cm, minimum height=4.0cm},
    channel/.style={draw, minimum width=1.4cm, minimum height=1.6cm},
    chanlabel/.style={font=\large},
    dimlabel/.style={font=\small, above=8pt}
}
\def\xshift{-1}

\draw[draw=gray!70, thick, dashed, rounded corners=10pt, fill=gray!8]
    ({\xshift+1.7},2.6) rectangle ({\xshift+7.0},7.4);
\node[text=gray!70!black] at ({\xshift+4.35},5.0) {\Large$\widetilde{\mathcal{N}}$};

\draw[block] (\xshift,3.0) rectangle ({\xshift+1},7.0);
\node at ({\xshift+0.5},5.0) {\LARGE $\mathcal{E}$};

\draw[procblock] ({\xshift+2.0},3.0) rectangle ({\xshift+3.0},7.0);
\node at ({\xshift+2.5},5.0) {\Large $\widetilde{\mathcal{E}}$};

\draw[channel] ({\xshift+3.8},5.7) rectangle ({\xshift+5.0},7.0);
\node[chanlabel] at ({\xshift+4.4},6.35) {\Large $\mathcal{P}$};

\draw[channel] ({\xshift+3.8},3.0) rectangle ({\xshift+5.0},4.3);
\node[chanlabel] at ({\xshift+4.4},3.65) {\Large $\mathcal{A}$};

\draw[procblock] ({\xshift+5.8},3.0) rectangle ({\xshift+6.8},7.0);
\node at ({\xshift+6.3},5.0) {\Large $\widetilde{\mathcal{D}}$};

\draw[block] ({\xshift+7.8},3.0) rectangle ({\xshift+8.8},7.0);
\node at ({\xshift+8.3},5.0) {\LARGE $\mathcal{D}$};


\draw[qarrow] ({\xshift-0.8},5.0) -- (\xshift,5.0);

\draw[qarrow] ({\xshift+1},5.0) -- ({\xshift+2.0},5.0)
    node[dimlabel, midway] {$d_{A_1'}\!=\!2$};

\draw[qarrow] ({\xshift+3.0},6.35) -- ({\xshift+3.8},6.35)
    node[dimlabel, midway] {$d_{A_1}\!=\!4$};

\draw[qarrow] ({\xshift+3.0},3.65) -- ({\xshift+3.8},3.65)
    node[dimlabel, midway] {$d_{A_2}\!=\!4$};

\draw[qarrow] ({\xshift+5.0},6.35) -- ({\xshift+5.8},6.35)
    node[dimlabel, midway] {$d_{B_1}\!=\!4$};

\draw[qarrow] ({\xshift+5.0},3.65) -- ({\xshift+5.8},3.65)
    node[dimlabel, midway] {$d_{B_2}\!=\!5$};

\draw[qarrow] ({\xshift+6.8},6.35) -- ({\xshift+7.8},6.35)
    node[dimlabel, midway] {$d_{B_1'}\!=\!2$};

\draw[->, line width=2pt] ({\xshift+6.8},3.65) -- ({\xshift+7.8},3.65)
    node[dimlabel, midway] {$d_Z\!=\!2$};

\draw[qarrow] ({\xshift+8.8},5.0) -- ({\xshift+9.6},5.0);

\end{tikzpicture}
\caption{Protocol for joint channel.}
\label{fig:Protocol_joint}
\end{figure}

The idea of the underlying protocol for our two channels of interest is shown in Fig.~\ref{fig:Protocol_joint}.
Before showing the proof of Proposition~\ref{prop:equivalent_channel}, we recall some basic notions from the task of secret key distillation~\cite{DevetakWinter2005Distillation,Horodecki2005SecureKey,Horodecki2008PrivateStates} to define the notion of twisting and untwisting unitary.

\subsection{Secret-Key Distillation}
\subsubsection{Private States}

Secret key distillation is a fundamental cryptographic task in which two communication parties share a bipartite state $\rho_{A_1 B_1}$, whose purification is $\psi_{A_1 B_1 E}$ (the system $E$ is supposed to be under the possession of an eavesdropper Eve), and their goal is to perform  LOCC in order to transform it to a state that approximates an ideal secret key of dimension $K$, which is of the form $\gamma_{A_1 B_1 E} \approx \overline{\Phi^K}_{A_1 B_1} \otimes \sigma_E$, where
\begin{equation}\label{eq:key-state-example}
    \overline{\Phi^K}_{A_1 B_1} \coloneqq \frac{1}{K}\sum_{i=0}^{d-1} \ket{i}\!\bra{i}_{A_1} \otimes \ket{i}\!\bra{i}_{B_1}
\end{equation}
is a maximally classically correlated state and $\sigma_E$ is some state of Eve's system.

An equivalent notion of private state, which motivates our construction, is given by the following~\cite{Horodecki2005SecureKey,Horodecki2008PrivateStates}. Let  $A_1 = A_1' A_1''$ and $B_1 = B_1' B_1''$, where the systems $A_1'$ and $B_1'$ are called the 
\textit{key systems}, and the systems $A_1''$ and $B_1''$ are called the \textit{shield systems}. A state $\gamma_{A_1'B_1'A_1''B_1''}$ is called a \textit{bipartite private state} of dimension $K$, containing $\log K$ bits of secrecy, if after purifying $\gamma_{A_1'B_1'A_1''B_1''}$  to a state $\psi_{A_1'B_1'A_1''B_1''E}$ and tracing over the shield systems $A_1''B_1''$,  the resulting state $\psi_{A_1'B_1'E}$  is a secret key state of dimension $K$, as in~\eqref{eq:key-state-example}.

This definition of private state is useful because it can be characterized in the following form (see, e.g., \cite[Theorem~15.5]{Khatri2020Principles}):
\begin{equation}
\label{eq:characterization_private_states}
\gamma_{A_1'B_1'A_1''B_1''} =U_{A_1'B_1'A_1''B_1''}\big(\Phi^K_{A_1'B_1'}\otimes\sigma_{A_1''B_1''}\big)U^\dagger_{A_1'B_1'A_1''B_1''},
\end{equation}
where $\Phi^K_{A_1'B_1'}$ is an MES of dimension $K$ and $\sigma_{A_1''B_1''}$ is an arbitrary state of the shield systems.
The unitary $U_{A_1'B_1'A_1''B_1''}$ is called a  \textit{global twisting unitary} and has the form
\begin{equation}
\label{eq:twisting_unitary}
U_{A_1'B_1'A_1''B_1''} =\sum_{i,j=0}^{K-1}\ket{i}\!\bra{i}_{A_1'}\otimes\ket{j}\!\bra{j}_{B_1'}\otimes U_{A_1''B_1''}^{ij},
\end{equation}
where $U_{A_1''B_1''}^{ij}$ is unitary  for all $i,j \in \{0,\ldots,K-1\}$. Notice that $d_{A_1'} = d_{B_1'} = K$, but we make no assumptions about the dimensions $d_{A_1''} =  d_{B_1''} \eqqcolon d_s$ of the shield systems. A quantum channel $\mathcal{N}_{A_1 \to B_1}$ whose Choi state $\Phi^{\mathcal{N}}_{A_1 B_1}$ is a private state in the sense of~\eqref{eq:characterization_private_states} is called a \textit{private channel}. If such a state is also a PPT state, then it is called a private Horodecki channel. These channels output PPT states that can be used to distill a shared key between the parties holding $A_1$ and $B_1$. More formally, in analogy with the notion of entanglement test, defined by the projections $\{\Phi^d_{A_1 B_1}, \mathbbm{1}_{A_1 B_1} - \Phi^d_{A_1 B_1}\}$, we can define a \textit{privacy test} by the projections $\{\Pi_{A_1'B_1'A_1''B_1''}, 
\mathbbm{1}_{A_1'B_1'A_1''B_1''} - \Pi_{A_1'B_1'A_1''B_1''}\}$ (see, e.g.,~\cite{wilde2017converse}), 
where
\begin{equation}
\Pi_{A_1'B_1'A_1''B_1''} \coloneqq U_{A_1'B_1'A_1''B_1''} \big( \Phi^K_{A_1'B_1'} \otimes \mathbbm{1}_{A_1''B_1''} \big) U^\dagger_{A_1'B_1'A_1''B_1''}.
\end{equation}

Operationally, this test corresponds to  \textit{untwisting} the private state and projecting onto a maximally  entangled state on the key systems $A_1'B_1'$.  While all PPT states fail to pass an entanglement test with high probability (see~\eqref{eq:PPT_fidelity}), there exist some of them that are private in the sense specified above.

\subsubsection{Mixtures of Private States}

Let $\{\gamma^\ell_{A_1'B_1'A_1''B_1''}\}_{\ell=0}^{K-1}$ be a set of $K$ 
private states, each given by
\begin{equation}
\gamma^\ell_{A_1'B_1'A_1''B_1''} = U^\ell_{A_1'B_1'A_1''B_1''} \big((\Phi^K_{A_1'B_1'})^\ell \otimes \sigma^\ell_{A_1''B_1''}\big)  (U^\ell_{A_1'B_1'A_1''B_1''})^\dagger,
\end{equation}
where
\begin{align}
(\Phi^K_{A_1'B_1'})^\ell &\coloneqq  X_{B_1'}^{\ell}\, \Phi^K_{A_1'B_1'}\, X_{B_1'}^{-\ell}, \\
X_{B_1'}^{\ell} &\coloneqq \sum_{x=0}^{K-1} \ket{x \oplus \ell}\!\bra{x}_{B_1'}, \\
U^\ell_{A_1'B_1'A_1''B_1''} &\coloneqq \sum_{i,j=0}^{K-1} 
\ket{i}\!\bra{i}_{A_1'} \otimes 
\ket{j \oplus \ell}\!\bra{j \oplus \ell}_{B_1'} \otimes 
U^{ij,\ell}_{A_1''B_1''}.
\end{align}
The operator $X_{B_1'}^{\ell}$ is the generalized shift operator, with 
$X_{B_1'}^{-\ell} = (X_{B_1'}^{\ell})^\dagger$. The set 
$\{(\Phi^K_{A_1'B_1'})^\ell\}_{\ell=0}^{K-1}$ is an orthogonal set of states, 
i.e., $(\Phi^K)^\ell (\Phi^K)^{\ell'} = \delta_{\ell,\ell'} (\Phi^K)^\ell$. 
Expanding explicitly:
\begin{equation}
\label{eq:gamma_ell_explicit}
\gamma^\ell_{A_1'B_1'A_1''B_1''} = \frac{1}{K} \sum_{i,j=0}^{K-1} 
\ket{i}\!\bra{j}_{A_1'} \otimes 
X_{B_1'}^{\ell} \ket{i}\!\bra{j}_{B_1'} X_{B_1'}^{-\ell} \otimes 
U^{i,\ell}_{A_1''B_1''}\, \sigma^\ell_{A_1''B_1''}\, 
(U^{j,\ell}_{A_1''B_1''})^\dagger.
\end{equation}

Now let $\omega_{A_1'B_1'A_1''B_1''}$ denote the following state:
\begin{equation}
\label{eq:choi_state_generic_PPT_channel}
\omega_{A_1'B_1'A_1''B_1''} \coloneqq \sum_{\ell=0}^{K-1} p_\ell\, 
\gamma^\ell_{A_1'B_1'A_1''B_1''},
\end{equation}
where $\{p_\ell\}_\ell$ is a probability distribution. Observe that
\begin{equation}
\omega_{A_1'B_1'A_1''B_1''} = V_{A_1'B_1'A_1''B_1''} \left( 
\sum_{\ell=0}^{K-1} p_\ell\, (\Phi^K_{A_1'B_1'})^\ell \otimes 
\sigma^\ell_{A_1''B_1''} \right) (V_{A_1'B_1'A_1''B_1''})^\dagger,
\end{equation}
where
\begin{equation}
\label{eq:twisting-unitaries-shift}
V_{A_1'B_1'A_1''B_1''} \coloneqq \sum_{\ell=0}^{K-1} \sum_{i=0}^{K-1}  \ket{i}\!\bra{i}_{A_1'} \otimes \ket{i \oplus \ell}\!\bra{i \oplus \ell}_{B_1'} \otimes 
U^{i,\ell}_{A_1''B_1''}.
\end{equation}

\begin{lemma}
\label{lem:key-attacked}
Let $\omega_{A_1'B_1'A_1''B_1''} = \sum_{\ell=0}^{K-1} p_\ell\, 
\gamma^\ell_{A_1'B_1'A_1''B_1''}$ be a mixture of shifted private states 
as in~\eqref{eq:choi_state_generic_PPT_channel}. If one performs the completely dephasing channel $\overline{\Delta}(\cdot) \coloneqq \sum_{i=0}^{K-1} \ket{i}\!\bra{i}(\cdot)\ket{i}\!\bra{i}$ on the key 
system $A_1'$ and traces out the shield systems $A_1''B_1''$, the 
resulting state on the key systems $A_1'B_1'$ is:
\begin{equation}
\label{eq:key-attacked}
    (\overline{\Delta}_{A_1'} \otimes \operatorname{Tr}_{A_1''B_1''})
    (\omega_{A_1'B_1'A_1''B_1''}) 
    = \sum_{\ell=0}^{K-1} p_\ell\, \overline{\Phi^K}_{A_1'B_1'}^{\ell},
\end{equation}
where $\overline{\Phi^K}$ is a maximally correlated state of dimension $K$ and:
\begin{align}
    \overline{\Phi^K}_{A_1'B_1'}^{\ell} &\coloneqq 
    X_{B_1'}^{\ell}\, \overline{\Phi^K}_{A_1'B_1'}\, X_{B_1'}^{-\ell}
    = \frac{1}{K} \sum_{i=0}^{K-1} \ket{i}\!\bra{i}_{A_1'} \otimes 
    \ket{i \oplus \ell}\!\bra{i \oplus \ell}_{B_1'}.
\end{align}
\end{lemma}

\begin{proof}
Using the explicit expansion~\eqref{eq:gamma_ell_explicit} of each 
$\gamma^\ell$, we compute:
\begin{align}
    \overline{\Delta}_{A_1'}(\omega_{A_1'B_1'A_1''B_1''})
    &= \overline{\Delta}_{A_1'}\!\left(\sum_{\ell=0}^{K-1} p_\ell\, 
    \gamma^\ell_{A_1'B_1'A_1''B_1''}\right) \nonumber\\
    &= \sum_{\ell=0}^{K-1} p_\ell \frac{1}{K} \sum_{i,j=0}^{K-1} 
    \overline{\Delta}_{A_1'}(\ket{i}\!\bra{j}_{A_1'}) \otimes 
    X_{B_1'}^{\ell}\ket{i}\!\bra{j}_{B_1'}X_{B_1'}^{-\ell} \otimes 
    U^{i,\ell}_{A_1''B_1''}\, \sigma^\ell_{A_1''B_1''}\, 
    (U^{j,\ell}_{A_1''B_1''})^\dagger \nonumber\\
    &= \sum_{\ell=0}^{K-1} p_\ell \frac{1}{K} \sum_{i=0}^{K-1} 
    \ket{i}\!\bra{i}_{A_1'} \otimes 
    X_{B_1'}^{\ell}\ket{i}\!\bra{i}_{B_1'}X_{B_1'}^{-\ell} \otimes 
    U^{i,\ell}_{A_1''B_1''}\, \sigma^\ell_{A_1''B_1''}\, 
    (U^{i,\ell}_{A_1''B_1''})^\dagger,
\end{align}
where in the last step we used 
$\overline{\Delta}(\ket{i}\!\bra{j}) = \delta_{ij}\ket{i}\!\bra{i}$. In the literature, this state is called a 
``key-attacked state'' (see, e.g.,~\cite{Christandl_2017}). Tracing over the shield systems $A_1''B_1''$ and using $\operatorname{Tr}[U^{i,\ell}\, \sigma^\ell\, 
(U^{i,\ell})^\dagger] = \operatorname{Tr}[\sigma^\ell] = 1$, we find that
\begin{align}
    (\overline{\Delta}_{A_1'} \otimes \operatorname{Tr}_{A_1''B_1''})
    (\omega_{A_1'B_1'A_1''B_1''})
    &= \sum_{\ell=0}^{K-1} p_\ell \frac{1}{K} \sum_{i=0}^{K-1} 
    \ket{i}\!\bra{i}_{A_1'} \otimes 
    X_{B_1'}^{\ell}\ket{i}\!\bra{i}_{B_1'}X_{B_1'}^{-\ell} \nonumber\\
    &= \sum_{\ell=0}^{K-1} p_\ell\, 
    \overline{\Phi^K}_{A_1'B_1'}^{\ell},
\end{align}
concluding the proof.
\end{proof}

In the case of interest considered in the main text, both the key and shield systems are qubit systems,
i.e.\ $d_{A_1'} = d_{A_1''} = d_{B_1'} = d_{B_1''} = 2$. Setting $\ell = 0$
in~\eqref{eq:gamma_ell_explicit}, a generic (unshifted) private state takes
the following explicit form:
\begin{multline}
\gamma_{A_1'B_1'A_1''B_1''} = \frac{1}{2}\ket{00}\!\bra{00}_{A_1'B_1'}
\otimes U^{0}\sigma\, (U^{0})^\dagger 
+ \frac{1}{2}\ket{00}\!\bra{11}_{A_1'B_1'}\otimes U^{0}\sigma\, 
(U^{1})^\dagger \\
+ \frac{1}{2}\ket{11}\!\bra{00}_{A_1'B_1'}\otimes U^{1}\sigma\, 
(U^{0})^\dagger 
+ \frac{1}{2}\ket{11}\!\bra{11}_{A_1'B_1'}\otimes U^{1}\sigma\, 
(U^{1})^\dagger,
\end{multline}
where we wrote $U^{i} \equiv U^{i,0}_{A_1''B_1''}$ for brevity. Thus, a 
private state of the above form is described by two unitaries $U^0, U^1$ on
the shield systems and an arbitrary shield state $\sigma_{A_1''B_1''}$.

We now show that the Choi state $\Phi^{\mathcal{P}}_{A_1 B_1}$ of the Horodecki
channel in~\eqref{eq:choi-state-horo} is a mixture of two shifted private
states.
\begin{lemma}
\label{lem:choi-as-mixture}
The Choi state $\Phi^{\mathcal{P}}_{RB}$ of the Horodecki channel 
in~\eqref{eq:choi-state-horo} is a mixture of two shifted private states:
\begin{equation}
\label{eq:choi-as-mixture}
\Phi^{\mathcal{P}}_{RB} = (1-p)\,\gamma^0_{R'B'R''B''} + p\,
\gamma^1_{R'B'R''B''},
\end{equation}
with $p = 1/(1+\sqrt{2})$, where $\gamma^0$ and $\gamma^1$ are private 
states of the form~\eqref{eq:gamma_ell_explicit} with parameters:
\begin{align}
\label{eq:sector0-params}
&\textit{Sector } \ell = 0: &
\sigma^0 &= \pi_4, &
U^{0,0} &= \mathbbm{1}_4, &
U^{1,0} &= \begin{pmatrix}
    1 & 0 & 0 & 0\\
    0 & 0 & 1 & 0\\
    0 & 1 & 0 & 0\\
    0 & 0 & 0 & -1
\end{pmatrix}, \\[6pt]
\label{eq:sector1-params}
&\textit{Sector } \ell = 1: &
\sigma^1 &= \overline{\Phi^2}, &
U^{0,1} &= \mathbbm{1}_4, &
U^{1,1} &= \begin{pmatrix}
    \frac{1}{\sqrt{2}} & 0 & 0 & \frac{1}{\sqrt{2}}\\[2pt]
    0 & 0 & 1 & 0\\
    0 & 1 & 0 & 0\\[2pt]
    \frac{1}{\sqrt{2}} & 0 & 0 & -\frac{1}{\sqrt{2}}
\end{pmatrix},
\end{align}
where $\pi_4 \coloneq \mathbbm{1}_4/4$ is the maximally mixed state, $\overline{\Phi^2} \coloneqq 
\frac{1}{2}(\ket{00}\!\bra{00} + \ket{11}\!\bra{11})$ is the maximally correlated qubit state.
\end{lemma}

\begin{proof}
The four Bell states in the Choi state~\eqref{eq:choi-state-horo} split into two groups according to their key-system support. We identify the two groups separately.

\medskip\noindent\textit{Sector $\ell = 0$.} 
Since $\psi^0$ and  $\psi^1$ have support on the even-parity subspace $\{\ket{00}, \ket{11}\}_{A_1'A_2'}$, the $\ell = 0$ contribution is
\begin{equation}
\label{eq:sector0-claim}
(1-p)\,\gamma^0 = q_0\,\psi^0 \otimes \rho^{(0)} 
+ q_1\,\psi^1 \otimes \rho^{(1)},
\end{equation}
with $q_0 = q_1 = (1-p)/2$. 
From~\eqref{eq:gamma_ell_explicit} with $K = 2$ and $\ell = 0$, the 
private state $\gamma^0$ has the expansion
\begin{equation}
\gamma^0_{A_1' A_2' A_1'' A_2''} = \frac{1}{2}\sum_{i,j \in \{0,1\}} 
\ket{i}\!\bra{j}_{A_1'} \otimes \ket{i}\!\bra{j}_{A_2'} \otimes 
(U^{i,0}\, \sigma^0\, (U^{j,0})^\dagger)_{A_1'' A_2''}.
\end{equation}
Comparing the two sides 
of~\eqref{eq:sector0-claim}, we obtain the constraints:
\begin{align}
\label{eq:sector0-diag-constraint}
U^{i,0}\, \sigma^0\, (U^{i,0})^\dagger &= \frac{1}{2}(\rho^{(0)} + \rho^{(1)}) 
\qquad \forall\, i \in \{0,1\},\\
\label{eq:sector0-antidiag-constraint}
U^{0,0}\,\sigma^0\,(U^{1,0})^\dagger 
&= \frac{1}{2}\bigl(\rho^{(0)} - \rho^{(1)}\bigr).
\end{align}
A direct computation gives:\begin{align}
    \rho^{(0)} + \rho^{(1)} &= \frac{1}{2}\bigl(
\ket{00}\!\bra{00}
+ \ket{01}\!\bra{01} + \ket{10}\!\bra{10}+ \ket{11}\!\bra{11} \bigr) = \frac{1}{2}\mathbbm{1}_4,\\
\rho^{(0)} - \rho^{(1)} &= \frac{1}{2}\bigl(
\ket{00}\!\bra{00}
+ \ket{01}\!\bra{10} + \ket{10}\!\bra{01}  - \ket{11}\!\bra{11} \bigr).
\end{align}Substituting into~\eqref{eq:sector0-antidiag-constraint} and \eqref{eq:sector0-diag-constraint}, we obtain a solution:\begin{equation}
\label{eq:solution_l0}
\sigma^0 = \pi_4 \qquad U^{0,0} = \mathbbm{1}_4\qquad U^{1,0} = \bigl(
\ket{00}\!\bra{00} - \ket{11}\!\bra{11} 
+ \ket{01}\!\bra{10} + \ket{10}\!\bra{01}\bigr).
\end{equation}.

\medskip\noindent\textit{Sector $\ell = 1$.} Since  $\psi^2$ and $\psi^3$ have support on the odd-parity subspace $\{\ket{01}, \ket{10}\}_{A_1'A_2'}$, the $\ell = 1$ contribution is
\begin{equation}
\label{eq:sector1-claim}
p\,\gamma^1 = q_2\,\psi^2 \otimes \rho^{(2)} 
+ q_3\,\psi^3 \otimes \rho^{(3)},
\end{equation}
with $q_2 = q_3 = p/2$. 
From~\eqref{eq:gamma_ell_explicit} with $K = 2$ and $\ell = 1$:
\begin{equation}
\gamma^1_{A_1' A_2' A_1'' A_2''} = \frac{1}{2}\sum_{i,j \in \{0,1\}} 
\ket{i}\!\bra{j}_{A_1'} \otimes 
(\sigma_X)_{A_2'}\ket{i}\!\bra{j}_{A_2'}(\sigma_X)^\dagger_{A_2'} 
\otimes 
(U^{i,1}\, \sigma^1\, (U^{j,1})^\dagger)_{A_1'' A_2''}.
\end{equation}
Comparing the two sides 
of~\eqref{eq:sector1-claim}, we obtain the constraints:
\begin{align}
\label{eq:sector1-diag-constraint}
U^{i,1}\, \sigma^1\, (U^{i,1})^\dagger &= \frac{1}{2}(\rho^{(2)} + \rho^{(3)}) 
\qquad \forall\, i \in \{0,1\},\\
\label{eq:sector1-antidiag-constraint}
\,U^{0,1}  \sigma^0\, (U^{1,1})^\dagger 
&= \frac{1}{2} \bigl(\rho^{(2)} - \rho^{(3)}\bigr).
\end{align}
A direct computation gives:\begin{align}
 \rho^{(2)} + \rho^{(3)} &= 
\ket{00}\!\bra{00} + \ket{11}\!\bra{11} = \overline{\Gamma^2,}\\
\rho^{(2)} - \rho^{(3)} &= \frac{\sqrt{2}}{2}\bigl(\ket{00}\!\bra{00} + \ket{00}\!\bra{11} + \ket{11}\!\bra{00} - \ket{11}\!\bra{11}\bigr).
\end{align} 
Substituting into~\eqref{eq:sector1-diag-constraint} and~\eqref{eq:sector1-antidiag-constraint}, we obtain a solution:\begin{equation}
\label{eq:solution_l1}
\sigma^1 = \overline{\Phi^2} \qquad U^{0,1} = \mathbbm{1}_4\qquad U^{1,1} = \frac{1}{\sqrt{2}}\bigl(
\ket{00}\!\bra{00} - \ket{11}\!\bra{11} 
+ \ket{00}\!\bra{11} + \ket{11}\!\bra{00}\bigr) + \ket{01}\!\bra{10} + \ket{10}\!\bra{01}.
\end{equation}

Note that the solutions in \eqref{eq:solution_l0} and \eqref{eq:solution_l1} are not unique, and any choice of states $\{\sigma_{\ell}\}$ and unitaries $\{U^{i, \ell}\}_{i, \ell}$ satisfying \cref{eq:sector0-diag-constraint,eq:sector0-antidiag-constraint,eq:sector1-diag-constraint,eq:sector1-antidiag-constraint} equally works.
\end{proof}

\subsection{Proof of Effective Channel}

The proof of Proposition~\ref{prop:equivalent_channel} is constructive and is based on providing explicit encoding and decoding channels. It takes inspiration from the original construction of~\cite{Smith2008, Oppenheim_2008} and can be understood as a correction procedure for the channel $\mathcal{P}$ in~\eqref{eq:ppt_channel}. We first describe the action of the channels $\widetilde{\mathcal{E}}$ and $\widetilde{\mathcal{D}}$ and then prove that they give rise to the channel in~\eqref{eq:equivalent_channel}, when suitably combined with the erasure channel $\mathcal{A}$ and the PPT channel $\mathcal{P}$. 

When describing the encoding/decoding protocol in what follows, we label the input of the PPT channel as $A_1 \equiv A_1' A_1''$ and that of the erasure channel as $A_2 \equiv A_2' A_2''$; similarly, their outputs are denoted as $B_1$ and $B_2$, respectively. For reasons that will become clear in the following, the systems $A_1'$ and $A_2'$ are called \textit{key systems} and their dimension equals $d_{A_1'}= d_{A_2'} = K$, while the systems $A_1''$ and $A_2''$ are called \textit{shield systems}, and their dimension equals $d_{A_1''}= d_{A_2''} = d_s$.
\

To fix the notation, we denote the input state as:\begin{equation}
    \rho_{A_1'} \equiv \sum_{i,j=1}^{d_{A_1'}} \rho_{ij} \ket{i}\!\bra{j}_{A_1'},
\end{equation} and the states at various stages of the protocol are:\begin{align}
 \label{eq:definition_zeta}
    \zeta_{A_1 A_2} &\equiv \widetilde{\mathcal{E}}_{A_1' \to A_1'A_1''A_2'A_2''}(\rho_{A_1'}),\\
     \label{eq:definition_omega}
    \omega_{B_1 B_2} &\equiv (\mathcal{P}_{A_1\to B_1}\otimes \mathcal{A}_{A_2\to B_2})(\zeta_{A_1 A_2}),\\
    \label{eq:definition_xi}
    \xi_{B_1' Z} &\equiv \widetilde{\mathcal{D}}_{B_1 B_2 \to B_1' Z}(\omega_{B_1 B_2}).
\end{align}This notation will be used consistently throughout this section.

In the case of interest for superactivation, the key and shield systems are two-dimensional. In the following, however, we will describe the protocol for the most general case of a PPT channel having a Choi state as in~\eqref{eq:choi_state_generic_PPT_channel} and a quantum erasure channel as in~\eqref{eq:quantumerasurechannel} of input dimension $d_s \cdot K$ (and output dimension $d_s \cdot K +1$), and generic erasure probability $q \in[0,1]$. We will later specialize it to~\eqref{eq:choi-as-mixture} to prove Proposition~\ref{prop:equivalent_channel}, which corresponds to the choice:\begin{equation}
    K = d_S = 2, \qquad p = 1/(1+\sqrt{2}), \qquad q =1/2.
\end{equation}

\begin{theorem}
\label{thm:effective_channel_general}
Let $\mathcal{P} \in \mathrm{CPTP}(A_1 \to B_1)$ be a channel with $d_{A_1} = d_{B_1} = K d_s$ whose Choi state admits a decomposition into $K$ shifted private states as in~\eqref{eq:choi_state_generic_PPT_channel}:
\begin{equation}
    \Phi^{\mathcal{P}}_{A_1 B_1} = \sum_{\ell=0}^{K-1} p_\ell\, \gamma^\ell_{A_1'B_1'A_1''B_1''},
\end{equation}
and let $\mathcal{A}^{q} \in \mathrm{CPTP}(A_2 \to B_2)$ be a quantum erasure channel of input dimension $K d_s$. Then, there exists an encoding / decoding protocol $(\widetilde{\mathcal{E}},\widetilde{\mathcal{D}})$ such that $\widetilde{\mathcal{N}}\coloneq \widetilde{\mathcal{D}} \circ (\mathcal{P} \otimes \mathcal{A}^q) \circ \widetilde{\mathcal{E}}$ is the flagged channel
\begin{equation}
\label{eq:effective_channel_general}
    \widetilde{\mathcal{N}}(\rho) = (1-q)\,\rho \otimes \ket{0}\!\bra{0}_Z + q\, (\mathcal{X}^{\overline{p}} \circ \overline{\Delta})(\rho) \otimes \ket{1}\!\bra{1}_Z,
\end{equation}
with $\overline{\Delta}(\cdot) \coloneqq \sum_{i=0}^{K-1} \ket{i}\!\bra{i}(\cdot)\ket{i}\!\bra{i}$ the completely dephasing channel, $\mathcal{X}^{\overline{p}}(\cdot) \coloneqq \sum_{\ell=0}^{K-1} p_\ell\, X^\ell(\cdot)X^{-\ell}$ the random displacement channel.
\end{theorem}
\begin{proof}
    The next sections describe the action of the encoder $\widetilde{\mathcal{E}}$ and decoder $\widetilde{\mathcal{D}}$, using the notation introduced above.
\subsection{Encoding Procedure}
The encoding procedure can be understood as a generalization of the construction of~\cite{Smith2008}.
The encoder takes as $K$-dimensional input the key system $A_1'$ and provides as output a $K^2\cdot d_s^2$-dimensional system $A_1A_2$ that is then sent over the joint channel $\mathcal{P}_{A_1 \to B_1} \otimes \mathcal{A}^q_{A_2 \to B_2}$, according to the scheme in Figure~\ref{fig:encoder_protocol}:\begin{enumerate}
    \item Alice prepares the states $\Phi^{d_s}_{A_{1}'' A_2''}$ and $\ket{0}\!\bra{0}_{A_2'}$. This operation can be seen as the application of an appending channel:\begin{equation}
         \left(\operatorname{id}_{A_1'}\otimes \mathcal{R}^{\Phi^{d_s} \otimes \ket{0}\!\bra{0}}_{\mathbb{C}\to A_1''A_{2}'' A_2'}\right)(\rho_{A_1'}) \coloneqq  \rho_{A_1'} \otimes \Phi^{d_s}_{A_1'' A_2''} \otimes \ket{0}\!\bra{0}_{A_2'}.
    \end{equation} 
 \item Alice's input state $A_1'$ is coherently copied into $A_2'$ by applying a generalized controlled-NOT gate $\mathcal{U}_{A_1' A_2'}^{\operatorname{CNOT}}(\cdot) = U^{\operatorname{CNOT}}_{A_1'A_2'} (\cdot) U^{\operatorname{CNOT}}_{A_1'A_2'}$, defined by the unitary operator
\begin{equation}
\label{eq:gen_CNOT}
    U^{\operatorname{CNOT}}_{A_1'A_2'} \coloneqq \sum_{\ell=0}^{K-1} \ket{\ell}\!\bra{\ell}_{A_1'} \otimes X_{A_2'}^{\ell},
\end{equation}
where $X^\ell_{A_2'}$ is the generalized shift operator.  
    \item Alice sends $A_1 = A_1'A_1''$ over the PPT channel and $A_2 = A_2' A_2''$ over the erasure channel. 
\end{enumerate}

\begin{figure}[ht]
\centering
\begin{tikzpicture}[
    every node/.style={font=\large},
    block/.style={draw, line width=0.7pt, minimum width=2.2cm, minimum height=1.1cm, align=center},
    arrow/.style={line width=0.7pt, -{Stealth[scale=1.0]}},
    scale=0.88, transform shape
]

\draw[dashed, rounded corners=10pt, thick, draw=gray!70, fill=gray!8]
  (-1.6,-1.1) rectangle (5.5,5.0);
\node[above right, font=\large\bfseries, text=gray!70!black] at (2.5,5) {$\widetilde{\mathcal{E}}$};

\filldraw[black] (0.5,1.2) circle (1.3pt);
\node[left] at (0.5,1.2) {$\Phi^{d_s}_{A_{1}'' A_2''}$};
\draw[arrow] (0.5,1.2) -- (1.5,2.8) -- (6,2.8);
\draw[arrow] (0.5,1.2) -- (1.5,-0.4) -- (6,-0.4);
\node[above] at (3,2.8) {$A_1''$};
\node[above] at (3,-0.5) {$A_2''$};

\filldraw[black] (2.5,1) circle (1.3pt);
\node[left] at (2.5,1) {$\ket{0}_{A_{2}'}$};
\draw[arrow] (2.5,1) -- (6, 1);
\node[above] at (3,1) {$A_2'$};

\draw[arrow] (-3,4.2) -- (6,4.2);
\node[above] at (-2,4.2) {$A_1'$};

\filldraw[black] (4,4.2) circle (1.3pt);

\draw[line width=0.7pt] (4,4.2) -- (4,1.2);

\draw[fill=white] (4,1) circle (6pt);
\draw[line width=0.7pt] (4,0.8) -- (4,1.3);
\draw[line width=0.7pt] (3.7,1) -- (4.3,1);

\draw[block] (6,2.2) rectangle (8.2,4.8);
\node at (7.1,3.5) {\Large $\mathcal{P}$};

\draw[block] (6,-0.8) rectangle (8.2,1.6);
\node at (7.1,0.4) {\Large $\mathcal{A}$};

\end{tikzpicture}
\caption{Encoder of the protocol for the joint channel.}
\label{fig:encoder_protocol}
\end{figure}
These operations can be seen as the application of the following channel:
\begin{equation}
\label{eq:encoder_protocol}
    \widetilde{\mathcal{E}}_{A_1' \to A_1'A_1''A_2'A_2''}(\cdot) \coloneqq \mathcal{U}^{\operatorname{CNOT}}_{A_1' A_2'}\circ (\operatorname{id}_{A_1'}\otimes \mathcal{R}^{\Phi^{d_s} \otimes \ket{0}\!\bra{0}}_{\mathbb{C}\to A_1''A_{2}'' A_2'})(\cdot),
\end{equation}
so that the output state of the encoder $\zeta_{A_1 A_2}$ becomes: 
\begin{equation}
\label{eq:output_state_encoder}
    \zeta_{A_1 A_2} \equiv \sum_{i,j =0}^{K-1} \rho_{ij}\, \ket{i}\!\bra{j}_{A_1'} \otimes \ket{i}\!\bra{j}_{A_2'} \otimes \Phi^{d_S}_{A_1''A_2''},
\end{equation}
where the CNOT has coherently copied the values in $A_1'$ into $A_2'$.
    
\subsection{Decoding Procedure}

The decoding procedure is at the heart of the correction protocol. Since $d_{B_1} =K\cdot d_s$ and $d_{B_2} = K\cdot d_s +1$, the overall output of the two channels is $(K^2 \cdot d_s^2 +1)$-dimensional. Bob receives the output of the joint channel, denoted as $\omega_{B_1B_2}$, and first performs an \textit{erasure measurement}, that is, a projective measurement of the form $\Pi_{B_2} = \{\Pi_{A_2}, \ket{e}\!\bra{e}_{B_2}\}$ where $\ket{e}$ is the erasure symbol and $\Pi_{A_2}$ denotes the projector onto the input Hilbert space.
The action of the decoder is then described in Figure~\ref{fig:decoder_complete}. Note that $\{\mathcal{D}^z\}_{z \in \{0,1\}}$ is a conditional quantum channel, i.e.~a collection of (two) CPTP maps with a classical input and a quantum input, and the output also contains the classical information about the erasure measurement in the system~$Z$.

Thus, the output state decomposes as: \begin{equation}
\label{eq:omega_full}
    \omega_{B_1 B_2}
    = (1-q)\,(\mathcal{P}\otimes\operatorname{id}_{A_2})(\zeta_{A_1A_2})
    + q\,\mathcal{P}(\zeta_{A_1})\otimes\ket{e}\!\bra{e}_{B_2},
\end{equation}
where $\zeta_{A_1} \equiv \mathrm{Tr}_{A_2}[\zeta_{A_1A_2}]$. The erasure measurement acts on the subsystem $B_2$ alone, and its two outcomes determine which decoder is applied to the remaining systems. The two conditional states are:
\begin{equation}
\label{eq:omega_conditional}
    \omega^{(0)}_{B_1 A_2}
    \equiv (\mathcal{P}\otimes\operatorname{id}_{A_2})(\zeta_{A_1A_2}),
    \qquad
    \omega^{(1)}_{B_1}
    \equiv \mathcal{P}(\zeta_{A_1}),
\end{equation}
so that the post-measurement state, including the classical flag $Z$, is:
\begin{equation}
    \omega_{B_1 A_2 Z}
    = (1-q)\,\omega^{(0)}_{B_1 A_2}\otimes\ket{0}\!\bra{0}_Z
    + q\,\omega^{(1)}_{B_1}\otimes\ket{1}\!\bra{1}_Z.
\end{equation}
Therefore, the final state of the protocol also splits as:\begin{equation}
\label{eq:output_state_protocol}
    \xi_{B_1' Z}
    = (1-q)\,\xi^{(0)}_{B_1'}\otimes\ket{0}\!\bra{0}_Z
    + q\,\xi^{(1)}_{B_1'}\otimes\ket{1}\!\bra{1}_Z,
\end{equation}
where\begin{equation}
\label{eq:xi_conditional}
    \xi^{(0)}_{B_1'}
    \equiv (\mathcal{D}^0_{B_1 A_2 \to B_1'})(\omega^{(0)}_{B_1 A_2}),
    \qquad
    \xi^{(1)}_{B_1'}
    \equiv (\mathcal{D}^1_{B_1 \to B_1'})(\omega^{(1)}_{B_1}),
\end{equation}
consistent with the definition in~\eqref{eq:definition_xi}.

Since the decoder $\widetilde{\mathcal{D}}$ conditions on $Z$, we treat each case separately.
\begin{figure}[!ht]
\centering
\begin{tikzpicture}[>=stealth, scale=0.95]
\tikzset{
    qarrow/.style={->, line width=0.9pt},
    classarrow/.style={->, line width=1.5pt},
    block/.style={draw, minimum width=1.8cm, minimum height=2.8cm},
    piblock/.style={draw, minimum width=1.8cm, minimum height=2.8cm},
    decodeblock/.style={draw, minimum width=2.0cm, minimum height=4.5cm},
    lab/.style={font=\large}
}

\draw[dashed, rounded corners=8pt, thick, draw=gray!70, fill=gray!8]
  (3.8,-0.7) rectangle (9,5.8);
\node[above right, font=\normalsize\bfseries, text=gray!70!black] at (6.8,5.8) {$\widetilde{\mathcal{D}}$};

\draw[block] (0,3.2) rectangle (1.8,5.6);
\node at (0.9,4.4) {\large $\mathcal{P}$};

\draw[block] (0,0.2) rectangle (1.8,2.6);
\node at (0.9,1.4) {\large $\mathcal{A}$};

\draw[qarrow] (1.8,4.4) -- (6.3,4.4);
\node[lab, above] at (2.8,4.4) {$B_1$};

\draw[qarrow] (1.8,1.4) -- (4.0,1.4);
\node[lab, above] at (2.8,1.4) {$B_2$};

\draw[piblock] (4.0,0.2) rectangle (5.8,2.6);
\node at (4.9,1.4) {\large $\Pi_{B_2}$};

\draw[qarrow] (5.8,1.4) -- (6.3,1.4);

\draw[line width=1.5pt] (4.9,0.2) -- (4.9,-0.3) -- (7.55,-0.3);
\node[lab, above] at (6.2,-0.3) {$Z$};

\draw[classarrow] (7.55,-0.3) -- (7.55,0.4);

\draw[classarrow] (7.55,-0.3) -- (10.0,-0.3);
\node[lab, above] at (9.4,-0.3) {$Z$};

\filldraw[black] (7.55,-0.3) circle (2pt);

\draw[decodeblock] (6.3,0.4) rectangle (8.8,5.6);
\node at (7.55,3.0) {\large $\mathcal{D}^z$};

\draw[qarrow] (8.8,4.4) -- (10.0,4.4);
\node[lab, above] at (9.4,4.4) {$B_1'$};

\end{tikzpicture}
\caption{Decoder of the protocol for the joint channel: a global view.}
\label{fig:decoder_complete}
\end{figure}

\paragraph{Case 1: No Erasure:}

In the case of no erasure, Bob uses the decoder $\mathcal{D}^0$, which can be seen in Figure~\ref{fig:decoder_noerasure}. It consists of the following steps:\begin{enumerate}
    \item Bob performs a global \textit{untwisting} unitary channel $\mathcal{V}^\dag_{B_1'A_2'B_1''A_2''}$ on the output systems. Recalling~\eqref{eq:twisting-unitaries-shift}, the corresponding unitary operator has the form
    \begin{equation}
    \label{eq:untwisting_unitary}
    V^\dag_{B_1'A_2'B_1''A_2''} \coloneqq \sum_{\ell=0}^{K-1} \sum_{i=0}^{K-1}  \ket{i \oplus \ell}\!\bra{i \oplus \ell}_{B_1'} \otimes \ket{i}\!\bra{i}_{A_2''}\otimes  
\left(U^{i,\ell}_{A_2''B_1''}\right)^\dag,
    \end{equation}
    where each $\left(U^{i,\ell}_{A_2''B_1''}\right)^\dag$ is a unitary on the shield systems $A_2''B_1''$ that inverts the corresponding $U^{i,\ell}_{A_1''B_1''}$.
    
    \item Bob traces over the shield systems $B_1''A_2''$. These first two operations correspond to \textit{privacy squeezing}, as defined in~\cite{horodecki2008low}.
    
    \item Bob performs a generalized parity measurement on systems $B_1'A_2'$, i.e.\ a projective measurement with $K$ outcomes given by:
    \begin{equation}
    \label{eq:parity_measurement}
        \Pi^{\mathcal{M}} = \left\{
        (\mathbbm{1}_{B_1'} \otimes X_{A_2'}^{-\ell})\,
        \overline{\Gamma^K}_{B_1' A_2'}\,
        (\mathbbm{1}_{B_1'} \otimes X_{A_2'}^{\ell})
        \right\}_{\ell \in \{0,\ldots,K-1\}},
    \end{equation}
    where 
    \begin{equation}
        \overline{\Gamma^K}_{B_1' A_2'} \coloneqq \sum_{i=0}^{K-1} \ket{ii}\!\bra{ii}_{B_1'A_2'}
    \end{equation}
    is the projector onto the maximally correlated subspace and $X_{A_2'}^{\ell}$ is the generalized shift operator. 
    
    \item Bob corrects the shift error. If the outcome of the measurement is $\ell$, he applies the correction 
    $\mathcal{X}^{-\ell}_{B_1'}(\cdot) \coloneqq X_{B_1'}^{-\ell}(\cdot)X_{B_1'}^{\ell}$ to the key system $B_1'$. For $\ell = 0$, this is the identity.
    
    \item Bob performs a generalized inverse CNOT on $B_1'A_2'$, defined by the unitary operator
    \begin{equation}
    \label{eq:gen_inv_CNOT}
        (U^{\operatorname{CNOT}}_{B_1'A_2'})^{-1} \coloneqq 
        \sum_{i=0}^{K-1} \ket{i}\!\bra{i}_{B_1'} \otimes X_{A_2'}^{-i}.
    \end{equation}
    \item Bob discards the system $A_2'$.
\end{enumerate}

Similarly to the encoding procedure, all these operations can be seen as a sequential concatenation of channels:
\begin{equation}
\label{eq:decoder_D0}
    \mathcal{D}^0(\cdot) \coloneqq 
    \Tr_{A_2'} \circ 
    (\mathcal{U}^{\operatorname{CNOT}}_{B_1'A_2'})^{-1} \circ 
    (\mathcal{X}^{-\ell}_{B_1'L \to B_1'} \otimes \mathrm{id}_{A_2'}) \circ 
    \mathcal{M}_{B_1'A_2' \to B_1'A_2'L} \circ 
    \Tr_{B_1''A_2''} \circ 
   \mathcal{V}^\dag_{B_1'A_2'B_1''A_2''}(\cdot),
\end{equation}
where $\mathcal{M}$ is a quantum instrument~\cite{DaviesLewis1970Instruments} implementing the parity  measurement~\eqref{eq:parity_measurement} with classical output $\ell$  stored in the register $L$, $\mathcal{X}^{-\ell}_{B_1'L \to B_1'}$ is a conditional quantum channel applying the shift correction $X^{-\ell}_{B_1'}$ controlled by the measurement outcome in $L$.

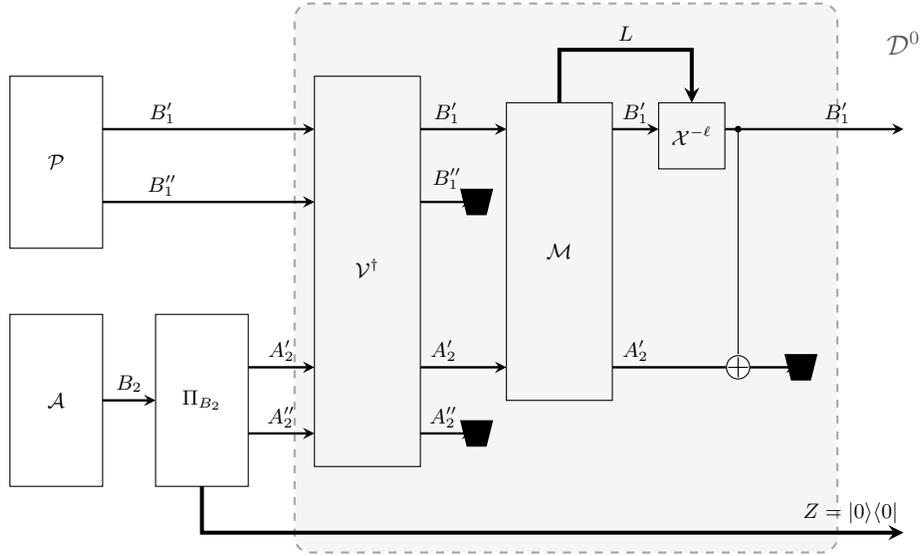
\begin{figure}[!ht]
\centering
\begin{tikzpicture}[>=stealth, scale=0.88, transform shape]
\tikzset{
    qarrow/.style={->, line width=0.9pt},
    classarrow/.style={->, line width=1.5pt},
    block/.style={draw, minimum width=1.4cm, minimum height=3.2cm},
    piblock/.style={draw, minimum width=1.4cm, minimum height=3.2cm},
    decodeblock/.style={draw, minimum width=1.6cm, minimum height=5.2cm},
    dblock/.style={draw=gray!70, dashed, rounded corners=6pt, thick}
}

\fill[gray!8, rounded corners=6pt] (4.3,-0.8) rectangle (12.5,7.5);
\draw[dblock] (4.3,-0.8) rectangle (12.5,7.5);
\node[text=gray!70!black] at (13.5,6.9) {\large $\mathcal{D}^0$};

\draw[block] (0,3.8) rectangle (1.4,6.4);
\node at (0.7,5.1) {$\mathcal{P}$};

\draw[block] (0,0.2) rectangle (1.4,2.8);
\node at (0.7,1.5) {$\mathcal{A}$};

\draw[qarrow] (1.4,5.6) -- (4.6,5.6);
\node at (2.3,5.85) {$B_1'$};

\draw[qarrow] (1.4,4.5) -- (4.6,4.5);
\node at (2.3,4.75) {$B_1''$};

\draw[qarrow] (1.4,1.5) -- (2.2,1.5);
\node at (1.8,1.75) {$B_2$};

\draw[piblock] (2.2,0.2) rectangle (3.6,2.8);
\node at (2.9,1.5) {$\Pi_{B_2}$};

\draw[qarrow] (3.6,2.0) -- (4.6,2.0);
\node at (4.1,2.25) {$A_2'$};

\draw[qarrow] (3.6,1.0) -- (4.6,1.0);
\node at (4.1,1.25) {$A_2''$};

\draw[classarrow] (2.9,0.2) -- (2.9,-0.5) -- (13.5,-0.5);
\node at (12.7,-0.2) {$Z = \ket{0}\!\bra{0}$};

\draw[decodeblock] (4.6,0.5) rectangle (6.2,6.4);
\node at (5.4,3.5) {$\mathcal{V}^\dag$};

\draw[qarrow] (6.2,5.6) -- (7.5,5.6);
\node at (6.6,5.85) {$B_1'$};

\draw[qarrow] (6.2,4.5) -- (7.0,4.5);
\node at (6.6,4.85) {$B_1''$};

\draw[qarrow] (6.2,2.0) -- (7.5,2.0);
\node at (6.55,2.25) {$A_2'$};

\draw[qarrow] (6.2,1.0) -- (7,1.0);
\node at (6.55,1.25) {$A_2''$};

\fill[black] (6.8,4.7) -- (7.3,4.7) -- (7.2,4.3) -- (6.9,4.3) -- cycle;
\fill[black] (6.8,1.2) -- (7.3,1.2) -- (7.2,0.8) -- (6.9,0.8) -- cycle;

\draw[decodeblock] (7.5,1.5) rectangle (9.1,6.0);
\node at (8.3,3.75) {$\mathcal{M}$};

\draw[qarrow] (9.1,5.6) -- (9.8,5.6);
\node at (9.45,5.85) {$B_1'$};

\draw[decodeblock] (9.8,5.0) rectangle (10.8,6.0);
\node at (10.3,5.5) {$\mathcal{X}^{-\ell}$};

\draw[classarrow] (8.3,6.0) -- (8.3,6.8) -- (10.3,6.8) -- (10.3,6.0);
\node at (9.3,7.05) {$L$};

\draw[qarrow] (9.1,2.0) -- (11.9,2.0);
\node at (9.45,2.25) {$A_2'$};

\filldraw[black] (11.0,5.6) circle (1.2pt);
\draw (11.0,5.6) -- (11.0,2.2);
\draw[fill=white] (11.0,2.0) circle (5pt);
\draw (11.0,1.85) -- (11.0,2.15);
\draw (10.85,2.0) -- (11.15,2.0);

\fill[black] (11.7,2.2) -- (12.2,2.2) -- (12.1,1.8) -- (11.8,1.8) -- cycle;

\draw[qarrow] (10.8,5.6) -- (13.5, 5.6);
\node at (12.5,5.85) {$B_1'$};

\end{tikzpicture}
\caption{Scheme of decoder $\mathcal{D}^0$ corresponding to the case that an erasure is not detected.}
\label{fig:decoder_noerasure}
\end{figure}

\paragraph{Case 2: Erasure:} 

In the case of an erasure, Bob uses the decoder $\mathcal{D}^1$, which can be seen in Figure~\ref{fig:decoder_erasure}. It consists of the following steps:
\begin{enumerate}
    \item Bob discards the output system of the erasure channel $B_2$, since the original input state has been lost to the environment.
    \item Bob discards the other shield qubit $B_1''$.
\end{enumerate}

This decoder corresponds to the trace-out channel:
\begin{equation}
    \mathcal{D}^1 (\cdot)\coloneqq \Tr_{B_1''A_2'A_2''} (\cdot).
\end{equation}
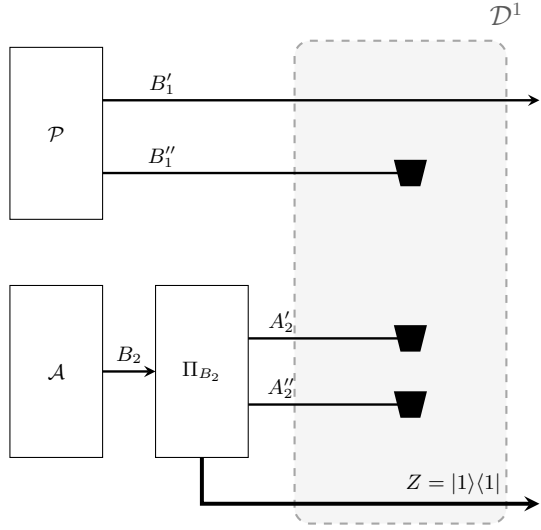
\begin{figure}[!ht]
\centering
\begin{tikzpicture}[>=stealth, scale=0.88, transform shape]
\tikzset{
    qarrow/.style={->, line width=0.9pt},
    classarrow/.style={->, line width=1.5pt},
    block/.style={draw, minimum width=1.4cm, minimum height=3.2cm},
    piblock/.style={draw, minimum width=1.4cm, minimum height=3.2cm},
    decodeblock/.style={draw, minimum width=1.6cm, minimum height=5.2cm},
    dblock/.style={draw=gray!70, dashed, rounded corners=6pt, thick}
}

\fill[gray!8, rounded corners=6pt] (4.3,-0.8) rectangle (7.5,6.5);
\draw[dblock] (4.3,-0.8) rectangle (7.5,6.5);
\node[text=gray!70!black] at (7.5,6.9) {\large $\mathcal{D}^1$};

\draw[block] (0,3.8) rectangle (1.4,6.4);
\node at (0.7,5.1) {$\mathcal{P}$};

\draw[block] (0,0.2) rectangle (1.4,2.8);
\node at (0.7,1.5) {$\mathcal{A}$};

\draw[qarrow] (1.4,5.6) -- (8.0,5.6);
\node at (2.3,5.85) {$B_1'$};

\draw[qarrow] (1.4,4.5) -- (6.2,4.5);
\node at (2.3,4.75) {$B_1''$};

\draw[qarrow] (1.4,1.5) -- (2.2,1.5);
\node at (1.8,1.75) {$B_2$};

\draw[piblock] (2.2,0.2) rectangle (3.6,2.8);
\node at (2.9,1.5) {$\Pi_{B_2}$};

\draw[qarrow] (3.6,2.0) -- (6.2,2.0);
\node at (4.1,2.25) {$A_2'$};

\draw[qarrow] (3.6,1.0) -- (6.2,1.0);
\node at (4.1,1.25) {$A_2''$};

\draw[classarrow] (2.9,0.2) -- (2.9,-0.5) -- (8.0,-0.5);
\node at (6.7,-0.2) {$Z = \ket{1}\!\bra{1}$};

\fill[black] (5.8,4.7) -- (6.3,4.7) -- (6.2,4.3) -- (5.9,4.3) -- cycle;
\fill[black] (5.8,1.2) -- (6.3,1.2) -- (6.2,0.8) -- (5.9,0.8) -- cycle;
\fill[black] (5.8,2.2) -- (6.3,2.2) -- (6.2,1.8) -- (5.9,1.8) -- cycle;

\end{tikzpicture}
\caption{Scheme of decoder $\mathcal{D}^1$ corresponding to the case that an erasure is detected.}
\label{fig:decoder_erasure}
\end{figure}

\subsection{Proof of equivalence}

Using the notation introduced above, the claim of the proposition can be rewritten, using~\eqref{eq:output_state_protocol}, as\begin{equation}
\xi^{(0)} = \rho, \qquad \xi^{(1)} = (\mathcal{X}^{\overline{p}} \circ \overline{\Delta})(\rho) \qquad \forall \rho \in \mathcal{D}(A_1'),
\end{equation}where $\overline{\Delta}(\cdot) \coloneqq \sum_{i=0}^{K-1} \ket{i}\!\bra{i} \,(\cdot)\, \ket{i}\!\bra{i}$ is the completely dephasing channel and $\mathcal{X}^{\overline{p}}(\cdot) \coloneqq \sum_{\ell=0}^{K-1} p_\ell \, X^\ell (\cdot) X^{-\ell}$ is a random displacement channel. Here, $X^\ell$ denotes a Heisenberg--Weyl displacement operator, and $\overline{p} \equiv \left\{p_\ell\right\}_\ell$ is the probability distribution determined by the PPT channel via~\eqref{eq:choi_state_generic_PPT_channel}.
Indeed, this is the same as requiring that in case of no-erasure the effective channel behaves as the identity, while in case of an erasure it behaves as the  noisy Pauli channel $\mathcal{X}^{\overline{p}} \circ \overline{\Delta}$.

\paragraph{Case 1: No Erasure:}

When no erasure occurs, the erasure channel acts as $\operatorname{id}_{A_2}$, and the conditional state is
$\omega^{(0)}_{B_1 A_2} = (\mathcal{P} \otimes\operatorname{id}_{A_2})(\zeta_{A_1 A_2})$. We need to show that in this case the decoder recovers the input, i.e.\ $\xi^{(0)} = \rho$.

We first compute the channel output using the Choi state of $\mathcal{P}$. The encoder output~\eqref{eq:output_state_encoder} can be expanded as:
\begin{equation}
    \zeta_{A_1 A_2} =\frac{1}{d_s} \sum_{i,j = 0}^{K-1}\sum_{a,b = 0}^{d_s-1} \rho_{ij}\,
    \ket{i}\!\bra{j}_{A_1'} \otimes \ket{i}\!\bra{j}_{A_2'} \otimes
     \ket{a}\!\bra{b}_{A_1''} \otimes
    \ket{a}\!\bra{b}_{A_2''}.
\end{equation}
Since the channel $\mathcal{P}$ acts on $A_1 = A_1'A_1''$ while $A_2$ passes through, we have by linearity:
\begin{equation}
\label{eq:output_omega_0_explicit}
    \omega^{(0)}_{B_1 A_2} = \frac{1}{d_s} \sum_{i,j = 0}^{K-1}\sum_{a,b = 0}^{d_s-1} \rho_{ij}\,\mathcal{P}_{A_1 \to B_1}\!\left(\ket{i}\!\bra{j}_{A_1'}  \otimes\ket{a}\!\bra{b}_{A_1''}\right) \otimes \ket{i}\!\bra{j}_{A_2'}\otimes\ket{a}\!\bra{b}_{A_2''}.
\end{equation}
We compute the channel output using~\eqref{eq:choi_jamiolkovski_isomorphism}, substituting in the Choi decomposition in~\eqref{eq:gamma_ell_explicit}:
\begin{equation}
    \Phi^{\mathcal{P}}_{A_1 B_1} = \frac{1}{K}\sum_{\ell=0}^{K-1}p_\ell\sum_{i', j'=0}^{K-1}\ket{i'}\!\bra{j'}_{A_1'} \otimes X^\ell \ket{i'}\!\bra{j'}_{B_1'}X^\ell \otimes (U^{i',\ell}\,\sigma^\ell\,(U^{j',\ell})^\dagger)_{A_1''B_1''},
\end{equation}
and we find that
\begin{equation}
\label{eq:action_P_standard_basis}
\begin{aligned}
    &\mathcal{P}\!\left(\ket{i}\!\bra{j}_{A_1'} \otimes\ket{a}\!\bra{b}_{A_1''}\right) \\
    &\quad= K d_s \Tr_{A_1}[\ket{i}\!\bra{j}_{A_1'} \otimes\ket{a}\!\bra{b}_{A_1''} \otimes \mathbbm{1}_{B_1}) \otimes \Phi^{\mathcal{P}}_{A_1B_1}]\\[4pt]
    &\quad= Kd_s \frac{1}{K} \sum_{\ell=0}^{K-1} p_\ell\sum_{i',j'=0}^{K-1}\operatorname{Tr}_{A_1'}\!\left[\ket{j}\!\bra{i}_{A_1'}\,\ket{i'}\!\bra{j'}_{A_1'}\right]\, X^{\ell}\ket{i'}\!\bra{j'}_{B_1'}X^{-\ell}\\
    &\qquad\qquad\qquad\qquad\quad \otimes\,\operatorname{Tr}_{A_1''}\!\left[\bigl(\ket{b}\!\bra{a}_{A_1''} \otimes \mathbbm{1}_{B_1''}\bigr)\,(U^{i',\ell}\,\sigma^\ell\,(U^{j',\ell})^\dagger)_{A_1''B_1''}\right]\\
    &\quad=d_s \sum_{\ell=0}^{K-1}
    p_\ell\,X^{\ell}\ket{i}\!\bra{j}_{B_1'}X^{-\ell} \otimes \operatorname{Tr}_{A_1''}\!\left[\bigl(\ket{b}\!\bra{a}_{A_1''} \otimes \mathbbm{1}_{B_1''}\bigr)\,
(U^{i,\ell}\,\sigma^\ell\,(U^{j,\ell})^\dagger)_{A_1''B_1''}\right]\\
    &\quad=d_s \sum_{\ell=0}^{K-1} p_\ell\,\ket{i\oplus \ell}\!\bra{j\oplus \ell}_{B_1'} \otimes \bigl(\bra{a}_{A_1''} \otimes \mathbbm{1}_{B_1''}\bigr)\,
  (U^{i,\ell}\,\sigma^\ell\,(U^{j,\ell})^\dagger)_{A_1''B_1''}\,
    \bigl(\ket{b}_{A_1''} \otimes \mathbbm{1}_{B_1''}\bigr),
\end{aligned}
\end{equation}
Substituting~\eqref{eq:action_P_standard_basis} into the full output~\eqref{eq:output_omega_0_explicit}, we get:
\begin{equation}
\begin{aligned}
\label{eq:omega0_full}
    \omega^{(0)}_{B_1 A_2} &=  \sum_{\ell=0}^{K-1} p_\ell \sum_{i,j = 0}^{K-1}\rho_{ij}\,
    X^{\ell}\ket{i}\!\bra{j}_{B_1'}X^{-\ell} \otimes \ket{i}\!\bra{j}_{A_2'} \otimes\\
    &\qquad \qquad \qquad \left(\sum_{a,b = 0}^{d_s-1} \left(\bigl(\bra{a}_{A_1''} \otimes \mathbbm{1}_{B_1''}\bigr)\,
  (U^{i,\ell}\,\sigma^\ell\,(U^{j,\ell})^\dagger)_{A_1''B_1''}\,
    \bigl(\ket{b}_{A_1''} \otimes \mathbbm{1}_{B_1''}\bigr) \otimes\ket{a}\!\bra{b}_{A_2''}\right)\right)\\
    &=  \sum_{\ell=0}^{K-1} p_\ell \sum_{i,j = 0}^{K-1}\rho_{ij}\,
    X^{\ell}\ket{i}\!\bra{j}_{B_1'}X^{-\ell} \otimes \ket{i}\!\bra{j}_{A_2'} \otimes  (U^{i,\ell}\,\sigma^\ell\,(U^{j,\ell})^\dagger)_{A_2''B_1''},
\end{aligned}
\end{equation}
where the last equality follows by explicitly writing the shield system operators
\begin{align}
X^{i,j}_{A_1'' B_1''}(\ell) & \equiv (U^{i,\ell}\,\sigma^\ell\,(U^{j,\ell})^\dagger)_{A_1''B_1''}, \\
Y^{i,j}_{B_1''A_2''}(\ell)& \equiv\sum_{a,b=0}^{d_s-1}\left((\bra{a}_{A_1''}\otimes \mathbbm{1}_{B_1''})X^{i,j}(\ell)_{A_1''B_1''}(\ket{b}_{A_1''}\otimes \mathbbm{1}_{B_1''})\right)\otimes\ket{a}\!\bra{b}_{A_2''},
\end{align}
in the basis $\{\ket{a}_{A_1''}\otimes \ket{a'}_{B_1''}\}_{a,a' \in [d_s]}$, and observing that, up to the relabeling $A_1'' \simeq A_2''$, we have $Y^{i,j}_{B_1''A_2''}(\ell) = \mathcal{F}(X^{i,j}_{A_1'' B_1''}(\ell))$ where $\mathcal{F}$ denotes the flip channel, which swaps systems $A_1''$ and $B_1''$, for all $i,j, \ell \in [K]$.

Moreover, the state~\eqref{eq:omega0_full} has a \textit{parity block-diagonal} structure on the key systems $B_1'A_2'$. To see this, observe that each $\ell$-term in~\eqref{eq:omega0_full} is supported on key pairs of the form $\ket{i \oplus \ell,\, i}\!\bra{j \oplus \ell,\, j}$, so we get:
\begin{equation}
\label{eq:omega0_direct_sum}
\begin{aligned}
        \omega^{(0)}_{B_1A_2} &\eqqcolon \bigoplus_{\ell=0}^{K-1}\;\omega_{B_1A_2}^{(0)}(\ell),
\end{aligned}
\end{equation}
where \begin{equation}
    \omega^{(0)}_{B_1A_2}(\ell) \equiv \sum_{i,j=0}^{K-1} \rho_{ij}\, p_\ell\,
    \ket{i \oplus \ell}\!\bra{j \oplus \ell}_{B_1'}\otimes \ket{i}\!\bra{j}_{A_2''} \otimes 
    (U^{i,\ell}\,\sigma^\ell\,(U^{j,\ell})^\dagger)_{A_2''B_1''}.
\end{equation}
The action of the untwisting unitary channel~\eqref{eq:untwisting_unitary} on~\eqref{eq:omega0_direct_sum} preserves the block-diagonal structure; indeed, for every $\ell \in [K]$, the corresponding set of unitaries on the shield system are simply $\left(U^{i, \ell}_{A_2'B_1'}\right)^\dag = \left(U^{i, \ell}_{A_2'B_1'}\right)^{-1}$, for all $i \in [K]$, and
\begin{equation}
    \mathcal{V}_{B_1 A_2}^\dag (\omega^{(0)}_{B_1A_2}) = \bigoplus_{\ell=0}^{K-1} p_\ell \ket{i \oplus \ell}\!\bra{j \oplus \ell}_{B_1'}\otimes \ket{i}\!\bra{j}_{A_2''} \otimes \sigma^\ell_{A_2''B_1''}.
\end{equation}
Tracing over the shield gives:
\begin{equation}
    \Tr_{B_1'' A_2''}[\mathcal{V}^\dag (\omega^{(0)}_{B_1A_2})] = \bigoplus_{\ell=0}^{K-1} p_\ell \ket{i \oplus \ell}\!\bra{j \oplus \ell}_{B_1'}\otimes \ket{i}\!\bra{j}_{A_2''}.
\end{equation}
Since the state is still parity block-diagonal, the parity measurement~\eqref{eq:parity_measurement} projects onto sector $\ell$ with probability $p_\ell$, and the correction $X_{B_1'}^{-\ell}$ shifts $\ket{i \oplus \ell} \mapsto \ket{i}$, leading to the state \begin{equation}
\label{eq:state_after_correction}
     \sum_{\ell = 0}^{K-1} p_\ell \sum_{i,j = 0}^{K-1} \rho_{i,j}\, \ket{i}\!\bra{j}_{B_1'} \otimes
    \ket{i}\!\bra{j}_{A_2'} = \sum_{i,j = 0}^{K-1} \rho_{i,j}\, \ket{i}\!\bra{j}_{B_1'} \otimes
    \ket{i}\!\bra{j}_{A_2'},
\end{equation}
where in the last equality we used $\sum_{\ell=0}^{K-1} p_\ell = 1$. Finally, by applying the inverse CNOT~\eqref{eq:gen_inv_CNOT} ($(U^{\operatorname{CNOT}})^{-1}\ket{i}_{B_1'}\ket{i}_{A_2'} = \ket{i}_{B_1'}\ket{0}_{A_2'}$) followed by tracing over $A_2'$, we obtain the final state
\begin{equation}
    \xi^{(0)}_{B_1'} = \mathcal{D}^0_{B_1 A_2 \to B_1'} (\omega^{(0)}_{B_1 A_2}) =\sum_{i,j=0}^{K-1} \rho_{i,j}\, \ket{i}\!\bra{j}_{B_1'} = \rho_{B_1'}.
\end{equation}
This proves that, conditioned on no erasure, the effective channel acts as the identity.

\paragraph{Case 2: Erasure:}

When an erasure takes place, the system $A_2 = A_2'A_2''$ is lost to the environment and the quantum correlations between $A_1$ and $A_2$ are  destroyed. Specifically, tracing over  $A_2 = A_2'A_2''$ from the encoder output~\eqref{eq:output_state_encoder}  gives the reduced state on $A_1$:
\begin{align}
    \zeta_{A_1} &\equiv \Tr_{A_2}\!\left[\zeta_{A_1A_2}\right] 
    = \sum_{i,j}\rho_{ij}\, \ket{i}\!\bra{j}_{A_1'} \cdot 
    \underbrace{\delta_{ij}}_{\Tr[\ket{i}\!\bra{j}_{A_2'}]} \cdot 
    \underbrace{\pi_{A_1''}}_{\Tr_{A_2''}[\Phi^{d_s}_{A_1'' A_2''}]} 
    \nonumber\\[4pt]
    &= \overline{\Delta}_{A_1'}(\rho_{A_1'}) \otimes \pi_{A_1''},
    \label{eq:zeta_A1_erasure}
\end{align}
where $\pi_{A_1''} = \mathbbm{1}_{A_1''}/d_s$. Thus, the loss of $A_2$  simultaneously dephases the key system $A_1'$ and replaces the shield system $A_1''$ with the maximally mixed state.

The decoder $\mathcal{D}^1$ traces out $B_1''$ from the channel output $\omega^{(1)}_{B_1} = \mathcal{P}(\zeta_{A_1})$, yielding a state on the key system $B_1'$ alone. To evaluate this, we observe that the composition of applying $\mathcal{P}$ to the input $\overline{\Delta}(\rho) \otimes \pi$ and tracing $B_1''$ is equivalent to a key attack on the Choi state in~\eqref{eq:key-attacked}. More precisely, using~\eqref{eq:choi_jamiolkovski_isomorphism}:
\begin{equation}
\label{eq:partial_trace_completely_dephasing}
\begin{aligned}
    \xi^{(1)}_{B_1'} &\coloneqq \operatorname{Tr}_{B_1''}\!\left[\mathcal{P}(\overline{\Delta}(\rho) \otimes \pi_{A_1''})\right]\\
    &= K d_s \operatorname{Tr}_{A_1 B_1''}\!\left[
    \bigl(\overline{\Delta}(\rho)^T_{A_1'} \otimes \pi_{A_1''} \otimes 
    \mathbbm{1}_{B_1}\bigr)\, \Phi^{\mathcal{P}}_{A_1 B_1}\right]\\
    &= K \sum_{i=0}^{K-1} \rho_{i,i}\, 
    \operatorname{Tr}_{A_1'' B_1''}\!\left[
    \bigl(\ket{i}\!\bra{i}_{A_1'} \otimes 
    \mathbbm{1}_{A_1'' B_1}\bigr)\, 
    \Phi^{\mathcal{P}}_{A_1 B_1}\right],
\end{aligned}
\end{equation}
where we used $\overline{\Delta}(\rho)^T = \overline{\Delta}(\rho)$ and we wrote explicitly the action of the 
completely dephasing channel. Using the Choi state decomposition~\eqref{eq:choi_state_generic_PPT_channel} and 
the explicit form~\eqref{eq:gamma_ell_explicit}, the dephasing forces 
$i = j$ in each $\gamma^\ell$ from the partial trace in~\eqref{eq:partial_trace_completely_dephasing}, and the shield trace gives unity ($\operatorname{Tr}[U^{i,\ell}\sigma^\ell(U^{i,\ell})^\dagger] = 1$), yielding:
\begin{align}
    \xi^{(1)} &= K \sum_{\ell=0}^{K-1}  \sum_{i=0}^{K-1} p_\ell \frac{1}{K} \rho_{ii}  
    \, 
    X^{\ell}\ket{i}\!\bra{i}_{B_1'}X^{-\ell} \nonumber\\
    &= \sum_{\ell=0}^{K-1} p_\ell\, X^{\ell}\, 
    \overline{\Delta}(\rho)\, X^{-\ell}
    = (\mathcal{X}^{\overline{p}} \circ \overline{\Delta})(\rho).
    \label{eq:erasure_effective_channel}
\end{align}

This proves that, conditioned on erasure, the effective channel acts as $\mathcal{X}^{\overline{p}} \circ \overline{\Delta}$, where $\overline{p} \equiv \{p_\ell\}_{\ell=0}^{K-1}$ is the probability distribution from the Choi decomposition~\eqref{eq:choi_state_generic_PPT_channel}.

The two cases are combined, weighted by the erasure probability $q$, into the flagged channel:
\begin{equation}
    \widetilde{\mathcal{N}}(\rho) = (1-q)\,\rho \otimes\ket{0}\!\bra{0}_Z +  q\, (\mathcal{X}^{\overline{p}} \circ \overline{\Delta})(\rho) \otimes \ket{1}\!\bra{1}_Z,
\end{equation}
thus concluding the proof.
\end{proof}

\begin{proof}[Proof of Proposition~\ref{prop:equivalent_channel}]
By Lemma~\ref{lem:choi-as-mixture}, the Choi state of the Horodecki channel admits the decomposition $\Phi^{\mathcal{P}}_{A_1B_1} = (1-p)\gamma^0 + p\gamma^1$ with $p = 1/(1+\sqrt{2})$. Since $K = d_s = 2$, the hypotheses of Theorem~\ref{thm:effective_channel_general} are satisfied with $q = 1/2$ and $\overline{p} = (1-p, p)$, giving:
\begin{equation}
    \widetilde{\mathcal{N}}(\rho) = \frac{1}{2}\,\rho \otimes \ket{0}\!\bra{0}_Z + \frac{1}{2}\,(\mathcal{X}^p \circ \overline{\Delta})(\rho) \otimes \ket{1}\!\bra{1}_Z,
\end{equation}
where $\mathcal{X}^p(\rho) = (1-p)\rho + p\,\sigma_X\rho\,\sigma_X$.

It remains to specify the four untwisting unitaries, which are simply given by the adjoint of the twisting unitaries from Lemma~\ref{lem:choi-as-mixture}. There are two parity sectors, each with two key pairs:

\medskip\noindent\textit{Even parity} ($\ell = 0$, key pairs $(0,0)$ and $(1,1)$):
\begin{equation}
   (U^{00})^\dagger = \mathbbm{1}_4, \qquad   (U^{10})^\dagger  = \begin{pmatrix}
    1 & 0 & 0 & 0\\
    0 & 0 & 1 & 0\\
    0 & 1 & 0 & 0\\
    0 & 0 & 0 & -1
    \end{pmatrix}.
\end{equation}
\medskip\noindent\textit{Odd parity} ($\ell = 1$, key pairs $(1,0)$ and $(0,1)$):
\begin{equation}
    (U^{01})^\dagger = \mathbbm{1}_4, \qquad  
     (U^{11})^\dagger  = \begin{pmatrix}
    \frac{1}{\sqrt{2}} & 0 & 0 & \frac{1}{\sqrt{2}}\\[2pt]
    0 & 0 & 1 & 0\\
    0 & 1 & 0 & 0\\[2pt]
    \frac{1}{\sqrt{2}} & 0 & 0 & -\frac{1}{\sqrt{2}}
    \end{pmatrix}.
\end{equation}
Notice that for this specific channel, the twisting and untwisting unitaries coincide, since they are all unitary involutions. In this case, the mixture weights $p_0 = 1-p$ and $p_1 = p$ from Lemma~\ref{lem:choi-as-mixture} become operationally equal to the probabilities of the even- and odd-parity outcomes, respectively, in the parity measurement~\eqref{eq:parity_measurement}, and the correction procedure for $\ell = 1$ reduces to the Pauli bit-flip $\sigma_X$. Moreover, the inverse CNOT coincides with the CNOT itself in~\eqref{eq:gen_CNOT} because $ \sigma_X^{-1} = \sigma_X$.
\end{proof}

In the following, we will focus on the lowest dimensional channels with $d_s = K = 2$, as discussed in the main text; in this case, the protocol consists of very simple operations (state preparation, unitary gates, measurements, etc.) on qubit systems; thus, in principle, it is easy to implement on a present-day quantum computer.

\section{Second-Order Analysis of Superactivation}

Having established the effective channel $\widetilde{\mathcal{N}}$ in Theorem~\ref{thm:effective_channel_general}, we can now compute its coherent information and coherent information variance analytically, leveraging the simple structure of~\eqref{eq:effective_channel_general}.

\subsection{Coherent Information Moments of the Effective Channel}

Consider sending the maximally entangled state $\Phi^K_{RA_1'}$ through the effective channel $\widetilde{\mathcal{N}}_{A_1' \to B_1'Z}$ in~\eqref{eq:effective_channel_general}. Following the notation of Section~\ref{sec:effective_protocol}, the corresponding states at various stages of the protocol are\begin{align}
    \zeta_{RA_1A_2} &\equiv \Phi_{RA_1A_2}^{\widetilde{\mathcal{E}}} = (\mathrm{id}_R \otimes \widetilde{\mathcal{E}}_{A_1' \to A_1A_2})(\Phi^K_{RA_1'}),\\
    \label{eq:choi_encoder+channels}
    \omega_{R B_1 B_2} &\equiv \Phi_{RB_1 B_2}^{(\mathcal{P}\otimes \mathcal{A})\circ \widetilde{\mathcal{E}}} = (\id_R \otimes (\mathcal{P}\otimes \mathcal{A})_{A_1A_2 \to B_1 B_2}) (\zeta_{RA_1A_2}),\\
     \omega_{R B_1 A_2 Z} &\equiv \Phi_{RB_1 A_2 Z}^{\Pi \circ (\mathcal{P}\otimes \mathcal{A})\circ \widetilde{\mathcal{E}}} = (\id_{RB_1} \otimes (\Pi)_{B_2 \to A_2Z}) (\omega_{R B_1 A_2 Z}) \\
     \label{eq:choi_encoder+channels+erasure_measurement}
     &= (1-q) \omega^{(0)}_{R B_1 A_2} \otimes \ket{0}\!\bra{0}_Z + q \omega^{(1)}_{RB_1} \otimes \ket{1}\!\bra{1}_Z,\\
     \label{eq:choi_effective_channel}
     \xi_{RB_1'Z} &\equiv  \Phi_{RB_1'Z}^{\widetilde{\mathcal{N}}}  = 
    (\mathrm{id}_R \otimes \widetilde{\mathcal{N}}_{A_1' \to B_1'Z})
    (\Phi^K_{RA_1'})\\
    &=(1-q)(\id_R \otimes \mathcal{D}^0)(\omega^{(0)}_{R B_1 A_2}) \otimes \ket{0}\!\bra{0}_Z + q (\id_R \otimes \mathcal{D}^1)\omega^{(1)}_{RB_1} \otimes \ket{1}\!\bra{1}_Z,
\end{align}
where:\begin{align}
    \omega^{(0)}_{R B_1A_2} \equiv \Phi_{RB_1 A_2}^{(\mathcal{P}\otimes \mathrm{id})\circ \widetilde{\mathcal{E}}}= (\mathrm{id}_R \otimes (\mathcal{P}_{A_1 \to B_1}\otimes \mathrm{id}_{A_2}))(\zeta_{RA_1A_2}),\\
    \omega^{(1)}_{R B_1} \equiv \Phi_{RB_1 A_2}^{(\mathcal{P}\otimes \mathrm{Tr})\circ \widetilde{\mathcal{E}}}= (\mathrm{id}_R \otimes (\mathcal{P}_{A_1 \to B_1}\otimes \mathrm{Tr}_{A_2}))(\zeta_{RA_1A_2}).
\end{align}
Using the explicit encoder channel $\widetilde{\mathcal{E}}$ as given in Theorem~\ref{thm:effective_channel_general}, these states can be determined explicitly. Specifically:\begin{equation}
    \label{eq:choi_encoder_protocol}
    \zeta_{RA_1A_2} = \Phi^{\mathrm{GHZ},K}_{RA_1'A_2'}\otimes \Phi^{d_s}_{A_1'' A_2''},
\end{equation}
where we defined:\begin{equation}
\label{eq:GHZ_state}
    \Phi_{RA_1'A_2'}^{\text{GHZ}, K} \coloneqq \frac{1}{K}\sum_{i,j=0}^{K-1}%
|i\rangle\!\langle j|_{R}\otimes|i\rangle\!\langle j|_{A_1'}%
\otimes|i\rangle\!\langle j|_{A_2'}.
\end{equation}
Similarly:
\begin{align}
\label{eq:omega_0_choi_state}
\omega^{(0)}_{R B_1A_2} &= \mathcal{U}^{\mathrm{CNOT}}_{RA_2'} (\Phi^{\mathcal{P}}_{R A_2''B_1' B_1''}\otimes \ket{0}\!\bra{0}_{A_2'})\\
&=V_{A_2' B_1' A_2''B_1''}\left(\sum_{\ell=0}^{K-1}p_{\ell}\Phi_{RB_1'A_2'}^{\text{GHZ},K^{\ell}}\otimes
\sigma_{A_2''B_1''}^{\ell}\right)  (V_{A_2' B_1' A_2''B_1''
})^{\dag},
\end{align}
where\begin{align}
\Phi_{RB_1'A_2'}^{\text{GHZ},K^{\ell}} &\coloneqq X_{A_1'\to B_1'}^{\ell}\Phi_{RA_1'A_2'}^{\text{GHZ}}X_{A_1'\to B_1'}^{-\ell}.
\end{align}
Also, in~\eqref{eq:omega_0_choi_state} we used the definition of $\Phi^{\mathcal{P}}$ and the fact that $\Phi^{K\cdot d_s}_{RA_2'' A_1'A_1''} \coloneq \Phi^{K}_{RA_1'}\otimes \Phi^{d_s}_{A_1''A_2''}$. On the other hand:
\begin{equation}
    \omega^{(1)}_{R B_1} = \sum_{\ell = 0}^{K-1} p_\ell\overline{\Phi^K
}_{RB_1'}^{\ell}\otimes\pi_{B_1''},
\end{equation}
where we used similar steps as in Lemma~\ref{lem:key-attacked} to compute the partial trace over $B_1'A_2'A_2''$.  

Finally, the Choi state of the effective channel $\widetilde{\mathcal{N}}$ can be computed explicitly as
\begin{align}
\label{eq:output_MES_effective}
    \xi_{RB_1'Z} &= (1-q)\, \Phi^K_{RB_1'} \otimes \ket{0}\!\bra{0}_Z 
    + q \sum_{\ell=0}^{K-1} p_\ell\, 
    (\overline{\Phi^K})^{\ell}_{RB_1'} \otimes \ket{1}\!\bra{1}_Z,
\end{align}
where we used $(\mathrm{id}_R \otimes \overline{\Delta})(\Phi^K) =  (\overline{\Phi^K})$ and $(\mathrm{id}_R \otimes \mathcal{X}^{\overline{p}})(\overline{\Phi^K}) = \sum_\ell p_\ell (\overline{\Phi^K})^\ell$.

\begin{remark}
    In~\eqref{eq:omega_0_choi_state}, the systems $RA_2''$ play the role of reference systems in the definition the Choi state $\Phi^{\mathcal{P}}$, in accordance with the conventions of Figure~\ref{fig:encoder_protocol}. Pictorially, one can think that Alice and Bob first share the Choi state $\Phi^{\mathcal{P}}$; specifically, Alice possesses systems $R$ and $A_2''$, Bob possesses systems $B_1'$ and $B_1''$, and the goal is to generate entanglement (of dimension $K$) using the bound entangled state $\Phi^{\mathcal{P}}$, a local encoding $\widetilde{\mathcal{E}}$ and the zero-quantum capacity channel $\mathcal{A}^q$. This setting, already considered in \cite{Wilde_2010}, gives an alternative interpretation of $\widetilde{\mathcal{E}}$ in the context of \textit{entanglement generation}: to generate entanglement between systems $R$ and $B_1'$, Alice appends the state $|0\rangle\!\langle0|_{A_2'}$ and performs a CNOT locally to obtain the state $\omega^{(0)}$. Then, she sends $A_2'A_2''$ over the erasure channel $\mathcal{A}^q$. If an erasure is detected, Bob then traces out system $B_1''$ and replaces it with the maximally mixed state, obtaining $\omega^{(1)}$. This interpretation applies to the special case where the input state is $\Phi^K_{RA_1'}$, where the encoder gives rise to the state in~\eqref{eq:choi_encoder_protocol}.
\end{remark}

The following theorem gives an alternative interpretation to the decoding protocol $\widetilde{\mathcal{D}}$ of Figure~\ref{fig:decoder_complete}. Recall that by the data-processing inequality for the coherent information (see, e.g., \cite[Theorem 11.9.3]{Wilde_book}), we have:\begin{equation}
\label{eq:DPI_coherent_info}
     I(R \rangle B_1'Z)_\xi \leq I(R\rangle B_1B_2)_{\omega},
\end{equation}
where $\omega_{RB_1B_2}$ is defined in~\eqref{eq:choi_encoder+channels} and $\xi_{RB_1'Z} = (\id_R \otimes \widetilde{\mathcal{D}}_{B_1B_2 \to B_1'Z})(\omega_{RB_1B_2})$ in~\eqref{eq:choi_effective_channel}. We now show that these two states have the \textit{same} coherent information, that is, the decoder $\widetilde{\mathcal{D}}$ is precisely a quantum channel that saturates~\eqref{eq:DPI_coherent_info}. We do so by explicitly computing all coherent information moments of $\omega_{RB_1B_2}$ and $\xi_{RB_1'Z}$ in closed form, making use of the classical--quantum structure of the underlying states and the preliminary results of Section~\ref{sec:second-order-analysis-background}.
\begin{theorem}
\label{thm:coh-info-moments}
For all $k \in \mathbb{N}$, the coherent information $k$th moments of the states $\omega_{RB_1B_2}$ and $\xi_{RB_1'Z}$, defined respectively in~\eqref{eq:choi_encoder+channels} and in~\eqref{eq:choi_effective_channel}, are equal and given by
\begin{equation}
    M_k(R\rangle B_1B_2)_{\omega} =  M_k(R \rangle B_1'Z)_\xi = (1-q)\,[\log K]^k + q\, M_k(\vec{p}),
\end{equation}
where $M_k(\vec{p}) \coloneqq \sum_{\ell=0}^{K-1} p_\ell\, 
[\log p_\ell]^k$.
\end{theorem}

\begin{proof}
First, let us show that the states in~\eqref{eq:choi_encoder+channels} and~\eqref{eq:choi_encoder+channels+erasure_measurement} have the same coherent information moments. This generalizes a standard argument
and follows from the fact that $\omega_{RB_1A_2 Z} = (\id_{RB_1}\otimes \Pi_{B_2 \to A_2Z})(\omega_{RB_1B_2})$, where the erasure measurement $\Pi_{B_2 \to A_2Z}$ can be seen as the application of an isometric channel:\begin{equation}
    U^Z_{B_2 \to B_2Z} = \Pi_{A_2} \otimes \ket{0}_Z + \ket{e}\!\bra{e}_{B_2}\otimes \ket{1}_Z,
\end{equation}followed by tracing out the erasure flag. Then:\begin{align}
\label{eq:M_k_CQ_state_output}
    M_k(R \rangle B_1B_2)_{\omega} &= M_k(R \rangle B_1B_2Z)_{\omega} \\
\label{eq:M_k_CQ_state_output_decomposed}
     &=\left( 1-q\right) M_{k}(R\rangle B_1A_2)_{\omega^{(0)}} + q M_{k}(R\rangle B_1 B_2)_{\omega^{(1)}\otimes \ket{e}\!\bra{e}}\\
        &=\left( 1-q\right) M_{k}(R\rangle B_1A_2)_{\omega^{(0)}} + q M_{k}(R\rangle B_1)_{\omega^{(1)}},
\end{align}
where the equality in~\eqref{eq:M_k_CQ_state_output} follows from Lemma~\ref{lem:mom-iso-inv}, the equality in~\eqref{eq:M_k_CQ_state_output_decomposed} follows from Lemma~\ref{lem:cq-moments}, and the last equality follows from Lemma~\ref{lem:kth-mom-prod-state}. 

Let us consider the two terms separately. The first term is given by:
\begin{align}
\label{eq:proof-M_k_omega_start}
M_{k}(R\rangle B_1A_2)_{\omega^{(0)}}
&  =M_{k}\!\left(  \omega^{(0)}_{RB_1A_2}\middle\Vert \mathbbm{1}_{R}\otimes\omega^{(0)}_{B_1A_2}\right)  \\
&  =M_{k}\!\left(
\begin{array}
[c]{c}%
V_{A_2^{\prime}B_1^{\prime}A_2^{\prime\prime}B_1^{\prime\prime}}\left(  \sum_{\ell}p_{\ell}\Phi_{RB_1^{\prime}A_2^{\prime}}^{\text{GHZ},K^{\ell}}\otimes\sigma_{A_2^{\prime\prime}B_1^{\prime\prime}}^{\ell}\right)
(V_{A_2^{\prime}B_1^{\prime}A_2^{\prime\prime}B_1^{\prime\prime}})^{\dag}\Vert\\
\label{eq:proof-M_k_omega_twisting}
\mathbbm{1}_{R}\otimes V_{A_2^{\prime}B_1^{\prime}A_2^{\prime\prime}B_1^{\prime\prime}}\left(  \sum_{\ell}p_{\ell
}\overline{\Phi}_{B_1^{\prime}A_2^{\prime}}^{\ell}\otimes\sigma_{A_2^{\prime\prime}B_1^{\prime\prime}}^{\ell
}\right)  (V_{A_2^{\prime}B_1^{\prime}A_2^{\prime\prime}B_1^{\prime\prime}})^{\dag}%
\end{array}
\right)  \\
\label{eq:proof-M_k_omega_untwisting}
&  =M_{k}\!\left(  \sum_{\ell}p_{\ell}\Phi_{RB_1^{\prime}A_2^{\prime}}^{\text{GHZ},K^{\ell}
}\otimes\sigma_{A_2^{\prime\prime}B_1^{\prime\prime}}^{\ell}\middle\Vert \mathbbm{1}_{R}%
\otimes\sum_{\ell}p_{\ell}\overline{\Phi}_{B_1^{\prime}A_2^{\prime}}^{\ell}\otimes\sigma_{A_2^{\prime\prime
}B_1^{\prime\prime}}^{\ell}\right)  \\
\label{eq:proof--M_k_omega_post-measurement}
&  =\sum_{\ell}p_{\ell}M_{k}\!\left(  \Phi_{RB_1^{\prime}A_2^{\prime}}^{\text{GHZ},K^{\ell}
}\otimes\sigma_{A_2^{\prime\prime}B_1^{\prime\prime}}^{\ell}\middle\Vert \mathbbm{1}_{R}%
\otimes\overline{\Phi}_{B_1^{\prime}A_2^{\prime}}^{\ell}\otimes\sigma_{A_2^{\prime\prime}B_1^{\prime\prime}}^{\ell
}\right)  \\
\label{eq:proof--M_k_omega_trace_shield}
&  =\sum_{\ell}p_{\ell}M_{k}\!\left(  \Phi_{RB_1^{\prime}A_2^{\prime}}^{\text{GHZ},K^{\ell}
}\middle\Vert \mathbbm{1}_{R}\otimes\overline{\Phi}_{B_1^{\prime}A_2^{\prime}}^{\ell}\right)  \\
\label{eq:proof--M_k_omega_X_correction}
&  =\sum_{\ell}p_{\ell}M_{k}\!\left(  \Phi_{RB_1^{\prime}A_2^{\prime}}^{\text{GHZ},K%
}\middle\Vert \mathbbm{1}_{R}\otimes\overline{\Phi}_{B_1^{\prime}A_2^{\prime}}\right)  \\
&  =M_{k}\!\left(  \Phi_{RB_1^{\prime}A_2^{\prime}}^{\text{GHZ},K}\middle\Vert \mathbbm{1}_{R%
}\otimes\overline{\Phi}_{B_1^{\prime}A_2^{\prime}}\right)  \label{eq:proof--M_k_omega_X_correction_after}.
\end{align}
In the above, we applied Lemmas~\ref{lem:mom-iso-inv}\ and
\ref{lem:kth-mom-prod-state}, as well as~\eqref{eq:direct-sim-kth-moment}.
Defining the projection%
\begin{equation}
\Pi^{\perp}\coloneqq \mathbbm{1}_{R}\otimes\overline{\Gamma}_{B_1^{\prime}A_2^{\prime}}-\Phi_{RB_1^{\prime}A_2^{\prime}}^{\text{GHZ},K},
\end{equation}
now consider that%
\begin{align}
\label{eq:proof--M_k_omega_X_undo_encoder}
  M_{k}\!\left(  \Phi_{RB_1^{\prime}A_2^{\prime}}^{\text{GHZ},K}\Vert \mathbbm{1}_{R}%
\otimes\overline{\Phi}_{B_1^{\prime}A_2^{\prime}}\right)  
&  =\operatorname{Tr}\!\left[  \Phi_{RB_1^{\prime}A_2^{\prime}}^{\text{GHZ},K}\left(  \log
_{2}\Phi_{RB_1^{\prime}A_2^{\prime}}^{\text{GHZ},K}-\log_{2}\mathbbm{1}_{R}\otimes\overline
{\Phi}_{B_1^{\prime}A_2^{\prime}}\right)  ^{k}\right]  \\
&  =\operatorname{Tr}\!\left[  \Phi_{RB_1^{\prime}A_2^{\prime}}^{\text{GHZ},K}\left(  \log
_{2}\Phi_{RB_1^{\prime}A_2^{\prime}}^{\text{GHZ},K}+\left(  \log_{2}K\right)  \mathbbm{1}_{R}\otimes\overline{\Gamma}_{B_1^{\prime}A_2^{\prime}}\right)  ^{k}\right]  \\
&  =\operatorname{Tr}\!\left[  \Phi_{RB_1^{\prime}A_2^{\prime}}^{\text{GHZ},K}\left(  \log
_{2}\Phi_{RB_1^{\prime}A_2^{\prime}}^{\text{GHZ},K}+\left(  \log_{2}K\right)  \left(
\Phi_{RB_1^{\prime}A_2^{\prime}}^{\text{GHZ},K}+\Pi^{\perp}\right)  \right)  ^{k}\right]  \\
&  =\operatorname{Tr}\!\left[  \Phi_{RB_1^{\prime}A_2^{\prime}}^{\text{GHZ},K}\left(  \left(
\log_{2}K\right)  \Phi_{RB_1^{\prime}A_2^{\prime}}^{\text{GHZ},K}+\left(  \log_{2}K\right)
\Pi^{\perp}\right)  ^{k}\right]  \\
&  =\operatorname{Tr}\!\left[  \Phi_{RB_1^{\prime}A_2^{\prime}}^{\text{GHZ},K}\left(  \left(
\log_{2}K\right)  ^{k}\Phi_{RB_1^{\prime}A_2^{\prime}}^{\text{GHZ},K}+\left(  \log
_{2}K\right)  ^{k}\Pi^{\perp}\right)  \right]  \\
&  =\left[  \log_{2}K\right]  ^{k}.
\label{eq:proof--M_k_omega_X_undo_encoder_end}
\end{align}
The second term is given by:
\begin{align}
\label{eq:erasure_term_M_k}
M_k(R \rangle B_1)_{\omega^{(1)}} &=M_k(R \rangle B_1')_{\sum_\ell p_\ell \overline{\Phi^K}^\ell}\\
&= M_{k}\!\left(  \sum_{\ell}p_{\ell}\overline{\Phi}_{R B_1'}^{\ell}
\middle\Vert \mathbbm{1}_{R}\otimes\operatorname{Tr}_{R}\!\left[
\sum_{\ell'}p_{\ell'}\overline{\Phi}_{R B_1'}^{\ell'}\right]  \right)  \\
&= M_{k}\!\left(  \sum_{\ell}p_{\ell}\overline{\Phi}_{R B_1'}^{\ell}
\middle\Vert \frac{\mathbbm{1}_{R}\otimes \mathbbm{1}_{B_1'}}{K}\right)  \\
&= \operatorname{Tr}\!\left[  \sum_{\ell}p_{\ell}\overline{\Phi}_{R B_1'}^{\ell}
\left(  \log_{2}\sum_{\ell'}p_{\ell'}\overline{\Phi}_{R B_1'}^{\ell'}
-\log_{2}\frac{\mathbbm{1}_{R}\otimes \mathbbm{1}_{B_1'}}{K}\right)^{k}\right]  \\
&= \sum_{\ell}p_{\ell}\operatorname{Tr}\!\left[  \overline{\Phi}_{R B_1'}^{\ell}
\left(  \log_{2}\sum_{\ell'}p_{\ell'}K\,\overline{\Phi}_{R B_1'}^{\ell'}\right)^{k}\right]  \\
&= \sum_{\ell}p_{\ell}\operatorname{Tr}\!\left[  \overline{\Phi}_{R B_1'}^{\ell}
\left(  \log_{2}\sum_{\ell'}p_{\ell'}\overline{\Gamma}_{R B_1'}^{\ell'}\right)^{k}\right]  \\
&= \sum_{\ell}p_{\ell}\operatorname{Tr}\!\left[  \overline{\Phi}_{R B_1'}^{\ell}
\left(  \sum_{\ell'}\bigl[\log_{2}(p_{\ell'})\bigr]^{k}\overline{\Gamma}_{R B_1'}^{\ell'}\right)\right]  \\
&= \sum_{\ell,\ell'}p_{\ell}\bigl[\log_{2}(p_{\ell'})\bigr]^{k}
\operatorname{Tr}\!\left[  \overline{\Phi}_{R B_1'}^{\ell}\,
\overline{\Gamma}_{R B_1'}^{\ell'}\right]  \\
&= \sum_{\ell}p_{\ell}\bigl[\log_{2}(p_{\ell})\bigr]^{k} \\
&= M_k(\vec p)\,,
\end{align}
where the first equality follows by Lemma~\ref{lem:kth-mom-prod-state}, and  we used the orthogonality of the projectors $\{\overline{\Gamma^K}^{\ell}\}_\ell \coloneqq \{K\,\overline{\Phi^K}^{\ell}\}_\ell$, and the property $\operatorname{Tr}_{R}[\sum_\ell p_\ell \overline{\Phi^K}^\ell] = 
\pi_K$ for all $\ell = 0, \ldots, K-1$. This concludes the first part of the proof. 

For the second part, since the state $\xi_{RB_1'Z}$ is classical--quantum with respect to $Z$,  by~\eqref{eq:direct-sim-kth-moment}:
\begin{equation}
    M_k(R \rangle B_1'Z)_\xi = (1-q)\, M_k(R \rangle B_1')_{\Phi^K} 
    + q\, M_k(R \rangle B_1')_{\sum_\ell p_\ell \overline{\Phi}^\ell}.
\end{equation}
For the first term, using the projector onto the MES $\Gamma^K$ and the orthogonal complement $\Pi^\perp \coloneqq \mathbbm{1}_{K^2} - \Gamma^K$, we have:
\begin{align}
    M_k(R \rangle B_1')_{\Phi^K} 
    &= \operatorname{Tr}\!\left[\Phi^K \left(\log \Phi^K - 
    \log \frac{\mathbbm{1}_{K^2}}{K}\right)^k
    \right] \nonumber\\
    &= \operatorname{Tr}\!\left[\Phi^K \left(\log \Phi^K + 
    (\log K)\mathbbm{1}_{K^2}\right)^k\right]
    \nonumber\\
    &= \operatorname{Tr}\!\left[\Phi^K \left(\log \Phi^K + 
    (\log K)(\Gamma^K + \Pi^\perp)\right)^k\right] \nonumber\\
    &= \operatorname{Tr}\!\left[\Phi^K \left((\log K)\Gamma^K + 
    (\log K)\Pi^\perp\right)^k\right] = [\log K]^k,
\end{align}
The second term is the same as in~\eqref{eq:erasure_term_M_k}, so the proof is concluded.
\end{proof}  
Notice that the algebraic steps from~\eqref{eq:proof-M_k_omega_start} to~\eqref{eq:proof--M_k_omega_X_undo_encoder_end} are ``physically'' implemented by the correcting protocol $\mathcal{D}^0$ in Figure~\ref{fig:decoder_noerasure}, which allows to recover the identity channel. Specifically, the untwisting unitary channel corresponds to the step from~\eqref{eq:proof-M_k_omega_twisting} to~\eqref{eq:proof-M_k_omega_untwisting}, the generalized parity measurement and trace over the shield systems corresponds to steps from~\eqref{eq:proof-M_k_omega_untwisting} to~\eqref{eq:proof--M_k_omega_trace_shield}, the $\mathcal{X}$-correction corresponds to~\eqref{eq:proof--M_k_omega_X_correction}-\eqref{eq:proof--M_k_omega_X_correction_after}, and the last two steps, which undo the encoding operation $\widetilde{\mathcal{E}}$, correspond to the steps from~\eqref{eq:proof--M_k_omega_X_undo_encoder} to~\eqref{eq:proof--M_k_omega_X_undo_encoder_end}.

\begin{corollary}
\label{cor:coh-info-variance}
The coherent information and coherent information variance of the states $\xi_{RB_1'Z}$ and $\omega_{RB_1B_2}$ are:
\begin{align}
    I(R\rangle B_1B_2)_{\omega} &= I(R \rangle B_1'Z)_\xi = (1-q)\log K - q\, H(\vec{p}), 
    \label{eq:coh_info_effective}\\
    V(R\rangle B_1B_2
)_{\omega} &=  V(R \rangle B_1'Z)_\xi=  q\, V(\vec{p}) + q(1-q)\,
    \bigl(\log K + H(\vec{p})\bigr)^2, 
    \label{eq:coh_var_effective}
\end{align}
where $H(\vec{p}) \coloneqq -\sum_\ell p_\ell \log p_\ell$ and $V(\vec{p}) \coloneqq \sum_\ell p_\ell (-\log p_\ell - H(\vec{p}))^2$ are the Shannon entropy and Shannon entropy variance of the 
distribution $\{p_\ell\}_\ell$.
\end{corollary}
\begin{proof}
Setting $k = 1$ in Theorem~\ref{thm:coh-info-moments} and using 
$M_1(\vec{p}) = -H(\vec{p})$ gives~\eqref{eq:coh_info_effective}. 
Setting $k = 2$:
\begin{align}
    V(R\rangle B_1B_2
)_{\omega}& =V(R \rangle B_1'Z)_\xi \\
&= M_2(R \rangle B_1'Z)_\xi - 
    [I(R \rangle B_1'Z)_\xi]^2 \nonumber\\
    &= (1-q)[\log K]^2 + q\, M_2(\vec{p}) - 
    [(1-q)\log K - q\, H(\vec{p})]^2.
\end{align}
Expanding the square and collecting terms:
\begin{align}
     V(R \rangle B_1'Z)_\xi&= q\, M_2(\vec{p}) + (1-q)[\log K]^2(1 - (1-q)) + 
    2q(1-q)\, H(\vec{p})\log K - q^2[H(\vec{p})]^2 \nonumber\\
    &= q\bigl(M_2(\vec{p}) - [H(\vec{p})]^2\bigr) + 
    q(1-q)\bigl([\log K]^2 + 2H(\vec{p})\log K + [H(\vec{p})]^2\bigr) 
    \nonumber\\
    &= q\, V(\vec{p}) + q(1-q)\bigl(\log K + H(\vec{p})\bigr)^2,
\end{align}
where we used $V(\vec{p}) = M_2(\vec{p}) - [H(\vec{p})]^2$.
\end{proof}

\begin{remark}
Comparing~\eqref{eq:coh_info_effective} with~\eqref{eq:DPI_coherent_info} gives an interpretation of the decoding channel $\mathcal{D}^{p}$ is the channel which saturates the data-processing inequality for the coherent information, that is it satisfies:\begin{equation}
\label{eq:saturation_DPI}
    D( \xi_{RZB_1'} \Vert \mathbbm{1}_R \otimes \xi_{ZB_1'} ) = D(\omega_{RB_1B_2} \Vert \mathbbm{1}_R \otimes \omega_{B_1B_2}),
\end{equation}
where $\xi_{RZB_1'} = (\id_R \otimes \mathcal{D}_{B_1B_2 \to ZB_1'}^p)(\omega_{RB_1B_2})$ and $\xi_{ZB_1'}\equiv \Tr_R[\xi_{RZB_1'}].$ This physically means that the local processing $\widetilde{\mathcal{D}}$ on the state $\omega_{RB_1B_2}$ is one that perfectly preserves the quantum correlations with the reference system $R$, as measured by the coherent information, and thus it is optimal for entanglement transmission over $(\mathcal{P}\otimes \mathcal{A})\circ \widetilde{\mathcal{E}}$.
\end{remark}

\subsection{First Demonstration of Non-Asymptotic Superactivation}

We now focus on the case $K = d_s = 2$ and $q = 1/2$ with the Horodecki channel of Lemma~\ref{lem:choi-as-mixture}. The distribution is $\vec{p} = (1-p,\, p)$ with $p = 1/(1+\sqrt{2})$, and the effective channel is $\widetilde{\mathcal{N}}$ from Proposition~\ref{prop:equivalent_channel}. It is well known from the analysis of \cite{Smith2008} that the state $\omega_{RB_1B_2}$ gives rise to a strictly positive coherent information:\begin{equation}
\label{eq:asymptotic_coherent_info_smith_yard}
    I_c(R\rangle B_1B_2)_{\omega} \approx 0.0107 >0.
\end{equation}
Since we are interested in a lower bound on the $n$-shot quantum capacity, the coherent information and the coherent information variance of $\widetilde{\mathcal{N}}$ can directly be evaluated at the MES $\Phi^2_{RA_1'}$, so to obtain:\begin{equation}
    \iota  \equiv I_c(R\rangle B_1'Z)_\xi =\frac{1}{2} (1  - h(p)),\quad 
    \nu  \equiv V(R \rangle B_1'Z)_\xi = \frac{1}{2}v(p) + \frac{1}{4}(1+h(p))^2,
\end{equation}
with the binary entropy $h(p)$ and binary entropy variance~$v(p)$ defined as
\begin{align}
h(p)& \coloneqq  -p\log p - (1-p) \log (1-p),\\
v(p) & \coloneqq p (-\log p -h(p))^2 
 + (1-p)(-\log (1-p) - h(p))^2.
\end{align}
For $p = 1/(1+\sqrt{2})$, we recover the result of~\eqref{eq:asymptotic_coherent_info_smith_yard}, proving that the channel $\widetilde{\mathcal{N}}$ can transmit quantum information at a non-trivial rate:\begin{equation}
    Q(\widetilde{\mathcal{N}}) \geq I_c(\widetilde{\mathcal{N}}) \geq \iota \approx 0.0107 >0.
\end{equation}

We can compute numerically the corresponding third order moment $\tau \equiv T^3(R\rangle B_1'Z)_\xi$, by evaluating~\eqref{eq:third-order-coh-absolute-moment} for the same input MES $\Phi^2_{RA_1'}, $ thus obtaining, by the data-processing inequality ($Q^{n, \varepsilon}(\mathcal{P}\otimes \mathcal{A}) \geq Q^{n, \varepsilon}(\widetilde{\mathcal{N}})$), the second-order lower bound for the joint channel:
\begin{equation}
    Q^{n, \varepsilon}(\mathcal{P}\otimes \mathcal{A}) \geq \iota + \sqrt{\frac{\nu}{n}} \Phi^{-1}\!\left(\varepsilon - \frac{1}{\sqrt{n}}\frac{C\, \tau}{\nu^{3/2}}\right),
    \label{eq:3nd-order-lower-bnd}
\end{equation}
which is the explicit formula using the Berry--Esseen Theorem on~\eqref{eq:2nd-order-lower-bnd}. The \textit{normal approximation} is instead an approximated lower bound where the higher-order term in \eqref{eq:2nd-order-lower-bnd} is neglected.

For our channel of interest, one has:\begin{equation}
   \iota \approx 0.0107, \qquad \nu \approx 1.009, \qquad \tau \approx 1.061,
\end{equation}
and substituting $C = 0.4748$ \cite{korolev2012improvement}, the results are shown in Figure \ref{fig:rate_bounds}. Notice that the third-order correction of~\eqref{eq:3nd-order-lower-bnd} is $\frac{C,\tau}{\nu^{3/2}\sqrt{n}} \approx \frac{0.4969}{\sqrt{n}}$, so the Berry--Esseen curve resembles the normal approximation, with a modest shift that delays the crossing, as shown in Figure \ref{fig:normal-approx} in the main text. 

To give an idea, for $\varepsilon = 0.25$,  the first value of $n$ yielding a rate larger than both uniform upper bounds in~\eqref{eq:outer_bound_PPT} and~\eqref{eq:outer_bound_erasure} is $n = 4218$ for the normal approximation and $n = 4504$ for the Berry--Esseen corrected bound (a difference of $286$ channel uses).

\begin{figure*}[t]
    \centering
    \includegraphics[width=\textwidth]{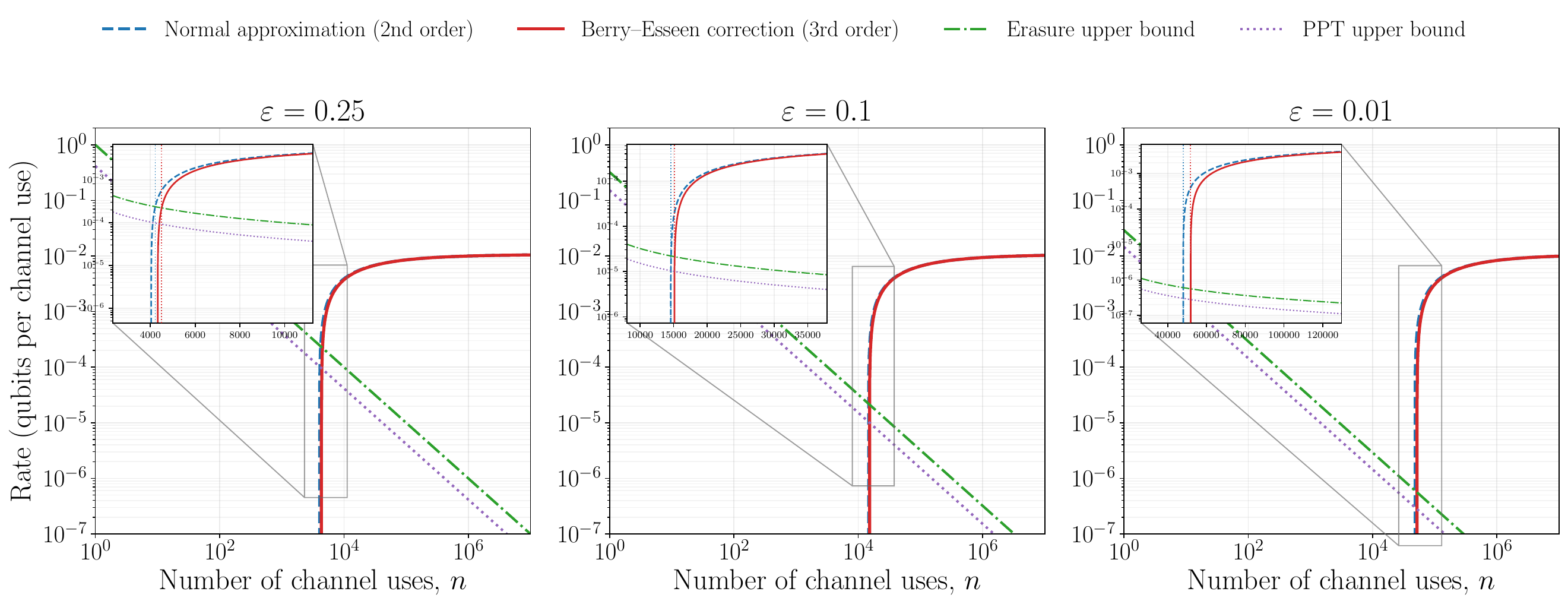}
    \caption{%
        Rate bounds for the superactivation channel $\mathcal{P} \otimes \mathcal{A}$ as a function of the number of channel uses $n$, for three error probabilities $\varepsilon$. The blue dashed curve shows the normal approximation (second-order) lower bound and the red solid curve shows the Berry--Esseen corrected (third-order) lower bound, both from Eq.\eqref{eq:lowerBound_analytic_n_shot} ($\iota \approx 0.0107$, $\nu \approx 1.009$, $\tau \approx 1.061$). The green dash-dotted and purple dotted curves show the upper bounds for the erasure and PPT component channels, respectively, from Eqs.\eqref{eq:outer_bound_erasure} and~\eqref{eq:outer_bound_PPT}. Both axes use logarithmic scale. The insets zoom into the crossing region ($n \sim 10^3$--$10^4$) where each lower bound first exceeds both upper bounds, certifying superactivation in the finite blocklength regime; the separation between the two lower bound crossings illustrates the effect of the third-order correction. The slow convergence at small $n$ motivates the direct optimization methods presented in Section~\ref{sec:symmetric_seesaw}.
    }
    \label{fig:rate_bounds}
\end{figure*}

\section{From Rates to Error Analysis}

The large number of channel uses needed to observe superactivation and the lack of explicit encoders and decoders from second-order analysis motivates us to switch the focus from the rate $Q^{n, \varepsilon }(\mathcal{N})$ to the achievable error $\varepsilon _Q^n(\mathcal{N},d)$, defining non-asymptotic superactivation as follows.
\begin{definition}
    Given $n \in \mathbb{N}$ and $\varepsilon  \in (0,1)$, we say that the channels $\mathcal{P}$ and $\mathcal{A}$ exhibit $n$-shot ($\varepsilon $-error) superactivation if:
    \begin{equation}
    \label{eq:n-shot-superact_SM}
\begin{aligned}
\max\!\left\{F_c(\mathcal{P}^{\otimes m},d),\,
F_c(\mathcal{A}^{\otimes m},d)\right\}
&\le 1-\varepsilon \quad \forall\, m\in\mathbb{N},\\
F_c((\mathcal{P}\!\otimes\!\mathcal{A})^{\otimes n},d)
&>1-\varepsilon .
\end{aligned}
\end{equation}
\end{definition}
In words, there exists a coding protocol for the joint channel achieving a  fidelity that is unattainable by the individual channels, regardless of the number of times they are used. Since the channel fidelity cannot be computed efficiently, we rely again on upper bounds on $F_c$ for an arbitrary number of uses of the single channels and lower bounds on $F_c$ for $n$ uses of the joint channel. 

As proven in the main text, we have the following uniform upper bounds for the two channels, holding for all $m\in \mathbb{N}$:
\begin{align}
\label{eq:PPT_fidelity}
    F_c(\mathcal{P}^{\otimes m},d) &\leq \frac{1}{d},\\
    \label{eq:ADG_fidelity}
   F_c(\mathcal{A}^{\otimes m},d) &\leq \frac{1}{2} + \frac{1}{2d}.
\end{align}
These bounds follow from the facts that the channels $\mathcal{P}^{\otimes m}$ and $\mathcal{A}^{\otimes m}$ are PPT and two-extendible, respectively, and that the local decoding operation preserves the PPT and the two-extendibility properties, respectively, of their output states. 

To approach the task of finding tight lower bounds on channel fidelity for $(\mathcal{P}\otimes \mathcal{A})^{\otimes n}$, we can use two approaches: $n$-shot error bounds and numerical techniques. In both cases, we make use of the effective channel of Section~\ref{sec:effective_protocol}, because one copy of the joint channel has dimensions $d_A = 16$ and $d_B = 20$. It is well-known that such a pre- and post-processing can only decrease the channel fidelity, but, for the sake of completeness, we state this as the following lemma.

\begin{lemma}
\label{lemma:DPI_channel_fidelity}
    For every pre-processing channel $\widetilde{\mathcal{E}} \in \mathrm{CPTP}(A' \to A)$ and post-processing channel $\widetilde{\mathcal{D}} \in \mathrm{CPTP}(B \to B')$, the channel fidelity of $\mathcal{N}_{A \to B}$ is non-increasing:
    \begin{equation}
        F_c(\widetilde{\mathcal{D}} \circ \mathcal{N} \circ \widetilde{\mathcal{E}}, d) \leq F(\mathcal{N}, d).
    \end{equation}
\end{lemma}

\begin{proof}
    Fix $d \in \mathbb{N}$. Then
    \begin{align*}
        F(\widetilde{\mathcal{D}} \circ \mathcal{N} \circ \widetilde{\mathcal{E}}, d) &= \max_{\substack{\mathcal{E} \in \mathrm{CPTP}(A'' \to A') \\ \mathcal{D} \in \mathrm{CPTP}(B' \to B'')}} F\!\left(\Phi^d_{RB''},\, (\mathrm{id}_R \otimes (\mathcal{D} \circ\widetilde{\mathcal{D}} \circ \mathcal{N} \circ \widetilde{\mathcal{E}}\circ \mathcal{E}))(\Phi^d_{RA''})\right) \\
        &\leq \max_{\substack{\tilde{\mathcal{E}} \in \mathrm{CPTP}(A'' \to A) \\ \tilde{\mathcal{D}} \in \mathrm{CPTP}(B \to B'')}} F\!\left(\Phi^d_{RB''},\, (\mathrm{id}_R \otimes (\tilde{\mathcal{D}} \circ \mathcal{M} \circ \tilde{\mathcal{E}}))(\Phi^d_{RB''})\right) \\
        &= F(\mathcal{N}, d),
    \end{align*}
    where the inequality holds because the compositions $\tilde{\mathcal{E}} = \widetilde{\mathcal{E}} \circ \mathcal{E}$ and $\tilde{\mathcal{D}} = \mathcal{D} \circ \widetilde{\mathcal{D}}$ are particular instances of the encoding and decoding channels  over which $F(\mathcal{N}, d)$ is optimized.
\end{proof} 
In our case of interest, this implies the following for the minimum achievable error:\begin{equation}
        \varepsilon_{Q}^n(\widetilde{\mathcal{D}} \circ (\mathcal{P}\otimes \mathcal{A}) \circ \widetilde{\mathcal{E}}, d) \geq   \varepsilon_{Q}^n(\widetilde{\mathcal{N}}, d).
    \end{equation}
for all $n\in \mathbb{N}$ and $d \in \mathbb{N}$.


\subsection{Details of Error Exponent Analysis}

Here we briefly discuss entropic error bounds, following the recent work \cite{berta2026tight}. 

The $\alpha$-\textit{Petz--R\'enyi relative entropy} is defined for $\alpha \in [0,2]$ as \cite{petz1985quasi}:\begin{equation}
\label{petzRenyi}
    \bar{D}_\alpha(\rho \|\sigma) \coloneqq \begin{cases}
        \frac{1}{\alpha -1} \log\left(\Tr[\rho^\alpha \sigma^{1-\alpha}]\right) \ \  \ \ &\text{supp}(\rho) \subseteq \text{supp}(\sigma), \\
        +\infty &\text{otherwise}.
    \end{cases}
\end{equation}
The corresponding $\alpha$-Petz--R\'enyi conditional entropy of a bipartite state is defined via the variational characterization:\begin{equation}
     H_\alpha(A|B)_{\rho} \coloneqq -\inf_{\sigma_B} D_\alpha(\rho_{AB} \|\mathbbm{1}_A \otimes \sigma_B),
\end{equation}
and its $\alpha$-R\'enyi coherent information is:\begin{equation}
\label{eq:petz_coh_renyi}
    I_{\alpha}(A\rangle B)_{\rho} \coloneqq- H_\alpha(A|B)_{\rho}  = \inf_{\sigma_{B}}D_{\alpha}(\rho_{AB}\|I_{A}\otimes\sigma_{B}).
\end{equation}

The result of \cite[Theorem 14]{berta2026tight}, gives an $n$-shot upper bound on the minimum error probability for entanglement transmission (as measured by the channel fidelity) of the form
\begin{equation}
\label{eq:error_exponent_bound_2}
    \varepsilon^n_{Q}(\mathcal{N}, d)
    \leq 2\sqrt{\frac{s^s (1-s)^{1-s}}{s}}\,
         2^{-\frac{s}{2}\bigl(I^c_{1/(1+s)}(\mathcal{N}^{\otimes n}) - \log d\bigr)},
\end{equation}
valid for every $s \in (0,1)$ (equivalently, $\alpha = 1/(1+s)\in (1/2, 1)$), where $I^c_{\alpha}(\mathcal{N}) \coloneqq \max_{\psi_{RA}} I_\alpha(R \rangle B)_{\omega}$ is the Petz--R\'enyi coherent information of $\mathcal{N}$ and $\omega_{RB} = (\mathrm{id}_R\otimes \mathcal{N}_{A\to B})(\psi_{RA})$.

Similarly to the second-order asymptotics analysis, since we are interested in an achievable error, we will use superadditivity of the channel Petz--R\'enyi coherent information to simplify the optimization, $I^c_\alpha(\mathcal{N}^{\otimes n})\geq nI^c_\alpha(\mathcal{N})$, and we will focus on the effective channel $\widetilde{\mathcal{N}}$ instead of $\mathcal{P}\otimes \mathcal{A}$, restricting to a MES input state $\psi_{RA} = \Phi_{RA}$ in the optimization. This can be computed analytically as shown below.


First, observe that the Petz--R\'enyi coherent information of a bipartite state $\rho_{AB}$, as defined in~\eqref{eq:petz_coh_renyi}, can be equivalently written as:
\begin{align}
I_{\alpha}(A\rangle B)_{\rho} & =\frac{\alpha}{\alpha-1}\log\Tr\!\left[\left(\Tr_{A}\!\left[\rho_{AB}^{\alpha}\right]\right)^{\frac{1}{\alpha}}\right],
\end{align}
where the equality follows from \cite[Lemma 3]{Sharma_2013}.
\begin{lemma}
\label{lem:cq-state-renyi-coh-info}For a classical--quantum state
of the following form:
\begin{equation}
\rho_{XAB}\coloneqq\sum_{x}p(x)|x\rangle\!\langle x|_{X}\otimes\rho_{AB}^{x},
\end{equation}
the following equality holds:
\begin{equation}
I_{\alpha}(A\rangle BX)_{\rho}=\frac{\alpha}{\alpha-1}\log\sum_{x}p(x)\Tr\!\left[\left(\Tr_{A}\!\left[\left(\rho_{AB}^{x}\right)^{\alpha}\right]\right)^{\frac{1}{\alpha}}\right].
\end{equation}
\end{lemma}

\begin{proof}
Consider that
\begin{align}
I_{\alpha}(A\rangle BX)_{\rho} & =\frac{\alpha}{\alpha-1}\log\Tr\!\left[\left(\Tr_{A}\!\left[\rho_{XAB}^{\alpha}\right]\right)^{\frac{1}{\alpha}}\right]\\
 & =\frac{\alpha}{\alpha-1}\log\Tr\!\left[\left(\Tr_{A}\!\left[\left(\sum_{x}|x\rangle\!\langle x|_{X}\otimes p(x)\rho_{AB}^{x}\right)^{\alpha}\right]\right)^{\frac{1}{\alpha}}\right]\\
 & =\frac{\alpha}{\alpha-1}\log\Tr\!\left[\left(\Tr_{A}\!\left[\sum_{x}|x\rangle\!\langle x|_{X}\otimes\left(p(x)\rho_{AB}^{x}\right)^{\alpha}\right]\right)^{\frac{1}{\alpha}}\right]\\
 & =\frac{\alpha}{\alpha-1}\log\Tr\!\left[\left(\sum_{x}|x\rangle\!\langle x|_{X}\otimes p(x)^{\alpha}\Tr_{A}\!\left[\left(\rho_{AB}^{x}\right)^{\alpha}\right]\right)^{\frac{1}{\alpha}}\right]\\
 & =\frac{\alpha}{\alpha-1}\log\Tr\!\left[\sum_{x}p(x)|x\rangle\!\langle x|_{X}\otimes\left(\Tr_{A}\!\left[\left(\rho_{AB}^{x}\right)^{\alpha}\right]\right)^{\frac{1}{\alpha}}\right]\\
 & =\frac{\alpha}{\alpha-1}\log\sum_{x}p(x)\Tr\!\left[\left(\Tr_{A}\!\left[\left(\rho_{AB}^{x}\right)^{\alpha}\right]\right)^{\frac{1}{\alpha}}\right],
\end{align}
thus concluding the proof.
\end{proof}
\begin{proposition}
\label{prop:petz-renyi-omega}
For all $\alpha > 0$, the Petz--R\'enyi coherent information of the states 
$\omega_{RB_1A_2Z}$ defined in~\eqref{eq:choi_encoder+channels+erasure_measurement} and $\xi_{RB_1'Z}$ defined in~\eqref{eq:output_MES_effective} are equal to:
\begin{equation}
\label{eq:petz_renyi_omega}
    I_\alpha(R \rangle B_1A_2Z)_\omega = I_\alpha(R \rangle B_1'Z)_\xi = \frac{\alpha}{\alpha-1}
    \log\!\left(q\left(\sum_{\ell=0}^{K-1} p_\ell^\alpha
    \right)^{\frac{1}{\alpha}} + (1-q)\, K^{\frac{\alpha-1}{\alpha}}
    \right).
\end{equation}
\end{proposition}
\begin{proof}
By applying Lemma~\ref{lem:cq-state-renyi-coh-info}, consider that
\begin{align}
I_{\alpha}(R\rangle B_1A_2Z)_{\omega} & =\frac{\alpha}{\alpha-1}\log\!\left(\begin{array}{c}
q\operatorname{Tr}\!\left[\left(\operatorname{Tr}_{R}\!\left[\omega^{(1)}_{RB_1'B_1''}\right)\right]^{\frac{1}{\alpha}}\right]\\
+\left(1-q\right)\operatorname{Tr}\!\left[\left(\operatorname{Tr}_{R}\!\left[\left(\omega^{(0)}_{RB_1'A_2'B_1''A_2''}\right)^{\alpha}\right]\right)^{\frac{1}{\alpha}}\right]
\end{array}\right).
\end{align}
Now consider that:
\begin{align}
 \operatorname{Tr}\!\left[\operatorname{Tr}_{R}\!\left[\left(\omega^{(1)}_{RB_1'B_1''}\right)^\alpha\right]^{\frac{1}{\alpha}}\right] &=\operatorname{Tr}\!\left[\left(\operatorname{Tr}_{R}\!\left[\left(\sum_{\ell=0}^{K-1}p_{\ell}\overline{\Phi^K}_{RB_1'}^{\ell}\otimes\pi_{B_1''}\right)^{\alpha}\right]\right)^{\frac{1}{\alpha}}\right] \nonumber\\
 & =\operatorname{Tr}\!\left[\left(\operatorname{Tr}_{R}\!\left[\sum_{\ell=0}^{K-1}p_{\ell}^{\alpha}\left(\overline{\Phi^K}_{RB_1'}^{\ell}\right)^{\alpha}\otimes\pi_{B_1''}^{\alpha}\right]\right)^{\frac{1}{\alpha}}\right]\\
 & =\operatorname{Tr}\!\left[\left(\operatorname{Tr}_{R}\!\left[\sum_{\ell=0}^{K-1}p_{\ell}^{\alpha}\frac{1}{K^{\alpha}}\overline{\Gamma^K}_{RB_1'}^{\ell}\otimes\pi_{B_1''}^{\alpha}\right]\right)^{\frac{1}{\alpha}}\right]\\
 & =\operatorname{Tr}\!\left[\left(\sum_{\ell=0}^{K-1}p_{\ell}^{\alpha}\frac{1}{K^{\alpha}}\mathbbm{1}_{B_1'}\otimes\pi_{B_1''}^{\alpha}\right)^{\frac{1}{\alpha}}\right]\\
 & =\frac{1}{K}\operatorname{Tr}\!\left[\left(\sum_{\ell=0}^{K-1}p_{\ell}^{\alpha}\right)^{\frac{1}{\alpha}}\mathbbm{1}_{B_1'}\otimes\pi_{B_1''}\right]\\
 & =\left(\sum_{\ell=0}^{K-1}p_{\ell}^{\alpha}\right)^{\frac{1}{\alpha}}.
\end{align}
The first equality follows because $\left\{ \overline{\Phi^K}_{RB_1'}^{\ell}\right\}_{\ell}$
is an orthogonal set. The second equality follows because $\overline{\Phi^K}_{RB_1'}^{\ell}=\frac{1}{K}\overline{\Gamma^K}_{RB_1'}^{\ell}$.
Also, consider that
\begin{align}
 & \operatorname{Tr}\!\left[\left(\operatorname{Tr}_{R}\!\left[\left(\omega^{(0)}_{RB_1'A_2'B_1''A_2''}\right)^{\alpha}\right]\right)^{\frac{1}{\alpha}}\right]\nonumber \\
 & =\operatorname{Tr}\!\left[\left(\operatorname{Tr}_{R}\!\left[\left(V_{A_2'B_1'A_2''B_1''}\left(\sum_{\ell=0}^{K-1}p_{\ell}\Phi_{RB_1'A_2'}^{\mathrm{GHZ},K^{\ell}}\otimes\sigma_{A_2''B_1''}^{\ell}\right)\left(V_{A_2'B_1'A_2''B_1''}\right)^{\dag}\right)^{\alpha}\right]\right)^{\frac{1}{\alpha}}\right]\\
 & =\operatorname{Tr}\!\left[\left(\operatorname{Tr}_{R}\!\left[\left(\sum_{\ell=0}^{K-1}p_{\ell}\Phi_{RB_1'A_2'}^{\mathrm{GHZ},K^{\ell}}\otimes\sigma_{A_2''B_1''}^{\ell}\right)^{\alpha}\right]\right)^{\frac{1}{\alpha}}\right]\\
 & =\operatorname{Tr}\!\left[\left(\operatorname{Tr}_{R}\!\left[\sum_{\ell=0}^{K-1}p_{\ell}^{\alpha}\Phi_{RB_1'A_2'}^{\mathrm{GHZ},K^{\ell}}\otimes\left(\sigma_{A_2''B_1''}^{\ell}\right)^{\alpha}\right]\right)^{\frac{1}{\alpha}}\right]\\
 & =\operatorname{Tr}\!\left[\left(\sum_{\ell=0}^{K-1}p_{\ell}^{\alpha}\overline{\Phi^K}_{B_1'A_2'}^{\ell}\otimes\left(\sigma_{A_2''B_1''}^{\ell}\right)^{\alpha}\right)^{\frac{1}{\alpha}}\right]\\
 & =\operatorname{Tr}\!\left[\sum_{\ell=0}^{K-1}p_{\ell}\left(\overline{\Phi^K}_{B_1'A_2'}^{\ell}\right)^{\frac{1}{\alpha}}\otimes\sigma_{A_2''B_1''}^{\ell}\right]\\
 & =\sum_{\ell=0}^{K-1}p_{\ell}\operatorname{Tr}\!\left[\left(\overline{\Phi^K}_{B_1'A_2'}^{\ell}\right)^{\frac{1}{\alpha}}\right]\\
 & =\sum_{\ell=0}^{K-1}p_{\ell}\frac{K}{K^{\frac{1}{\alpha}}}\\
 & =K^{1-\frac{1}{\alpha}}\\
 & =K^{\frac{\alpha-1}{\alpha}}.
\end{align}
The second equality follows because the function $(\cdot)\mapsto\operatorname{Tr}[(\cdot)^{\frac{1}{\alpha}}]$
is unitarily invariant. The third equality follows because $\left\{ \Phi_{RB_1'A_2'}^{\mathrm{GHZ},K^{\ell}}\right\}_{\ell}$
is an orthogonal set and because each $\Phi_{RB_1'A_2'}^{\mathrm{GHZ},K^{\ell}}$
is a pure state, so that $\left(\Phi_{RB_1'A_2'}^{\mathrm{GHZ},K^{\ell}}\right)^{\alpha}=\Phi_{RB_1'A_2'}^{\mathrm{GHZ},K^{\ell}}$.
The fourth equality follows because $\operatorname{Tr}_{R}\!\left[\Phi_{RB_1'A_2'}^{\mathrm{GHZ},K^{\ell}}\right]=\overline{\Phi^K}_{B_1'A_2'}^{\ell}$.
The fifth equality follows because $\left\{ \overline{\Phi^K}_{B_1'A_2'}^{\ell}\right\}_{\ell}$
is an orthogonal set. This implies that
\begin{align}
I_{\alpha}(R\rangle B_1A_2Z)_{\omega} & =\frac{\alpha}{\alpha-1}\log\!\left(q\left(\sum_{\ell=0}^{K-1}p_{\ell}^{\alpha}\right)^{\frac{1}{\alpha}}+\left(1-q\right)K^{\frac{\alpha-1}{\alpha}}\right),
\end{align}
thus concluding the proof for $I_{\alpha}(R\rangle B_1A_2Z)_{\omega}$. For $\xi_{RB_1'Z}$, by Lemma~\ref{lem:cq-state-renyi-coh-info} we have:\begin{equation}
    I_{\alpha}(R\rangle B_1'Z)_{\xi} =\frac{\alpha}{\alpha-1}\log\!q\operatorname{Tr}\!\left[\left(\operatorname{Tr}_{R}\!\left[\left(\xi^{(1)}_{RB_1'}\right)^\alpha\right]^{\frac{1}{\alpha}}\right)\right]+\left(1-q\right)\operatorname{Tr}\!\left[\left(\operatorname{Tr}_{R}\!\left[\left(\xi^{(0)}_{RB_1'}\right)^{\alpha}\right]\right)^{\frac{1}{\alpha}}\right]
\end{equation}
Here, the computation for the erasure branch is identical to Proposition~\ref{prop:petz-renyi-omega}, since $\omega^{(1)}_{RB_1} = \xi^{(1)}_{RB_1'} \otimes \pi_{B_1''}$, while for the no-erasure branch it simplifies to:
\begin{equation}
    \operatorname{Tr}\!\left[\left(\operatorname{Tr}_R\!\left[(\Phi^K_{RB_1'})^\alpha\right]\right)^{\frac{1}{\alpha}}\right]
    = \operatorname{Tr}\!\left[\left(\frac{\mathbbm{1}_{B_1'}}{K}\right)^{\frac{1}{\alpha}}\right] = \frac{K}{K^{1/\alpha}} 
    = K^{\frac{\alpha-1}{\alpha}},
\end{equation}
thus concluding the proof.
\end{proof}

Note that since\begin{equation}
    I_{\alpha}(A \rangle B)_{\rho} \xrightarrow{\alpha \to 1} I(A \rangle B),
\end{equation} by a continuity argument ($\iota_\alpha \to \iota$ as $\alpha \to 1$) this provides an alternative proof of~\eqref{eq:coh_info_effective}.

For the particular case of interest for us, we then find that the R\'enyi coherent information takes the following value, for $q=\frac{1}{2}$, $K=2$, and $p=\frac{1}{1+\sqrt{2}}$:
\begin{equation}
\label{eq:iota_alpha}
\iota_\alpha \equiv \frac{\alpha}{\alpha-1}\log\!\left(\frac{1}{2}\left(\frac{1}{\left(1+\sqrt{2}\right)^{\alpha}}+\left(\frac{\sqrt{2}}{1+\sqrt{2}}\right)^{\alpha}\right)^{\frac{1}{\alpha}}+2^{-\frac{1}{\alpha}}\right).
\end{equation}

The bound in~\eqref{eq:error_exponent_bound_2} thus becomes:
\begin{equation}
\label{eq:error_exponent_bound_superactivation_special_case}
     \varepsilon^n_{Q}(\widetilde{\mathcal{N}}, d)\leq 
     2\sqrt{\frac{s^s (1-s)^{1-s}}{s}}\cdot 2^{
     -\frac{s}{2} \left(n\, \iota_{1/(1+s)} - \log d\right)},
\end{equation}
valid for every $s \in (0, 1)$. By direct inspection of~\eqref{eq:iota_alpha}, we observe that $\iota_\alpha$ is monotonically increasing in $\alpha$, and $\iota_\alpha < 0$ for all $\alpha < \alpha^* \approx 0.9704$, so the exponent in~\eqref{eq:error_exponent_bound_superactivation_special_case} becomes negative only for a tiny interval of values $\alpha \in (\alpha^*, 1)$, vanishing at the extremes. Imposing $\partial_\alpha \bigl(\tfrac{1-\alpha}{\alpha} \iota_\alpha\bigr) = 0$, the maximum of the error exponent is attained at $\alpha_{\max} \approx 0.9849$ and is very small, of order $10^{-5}$.

We can now compare the achievable error from~\eqref{eq:error_exponent_bound_superactivation_special_case} with the uniform upper bounds in~\eqref{eq:PPT_fidelity} and~\eqref{eq:ADG_fidelity} for $d = 2$, where the bound to surpass to show superactivation is equal to $3/4$, as explained in the main text. The results, shown in Figure~\ref{fig:fidelity_bounds}, demonstrate that non-asymptotic superactivation in the sense of~\eqref{eq:n-shot-superact_SM} can indeed be proven via \eqref{eq:error_exponent_bound_superactivation_special_case}, but only at blocklengths of order $n \approx 10^5$. This is due to the combination of a small exponent (driven by the small value of $\iota_\alpha$) and the non-negligible prefactor $\sqrt{s^s (1-s)^{1-s}/s}$ in~\eqref{eq:error_exponent_bound_2}, which trivialize the bound over a large range of $n$.

\begin{figure}[t]
    \centering
    \includegraphics[width=\textwidth]{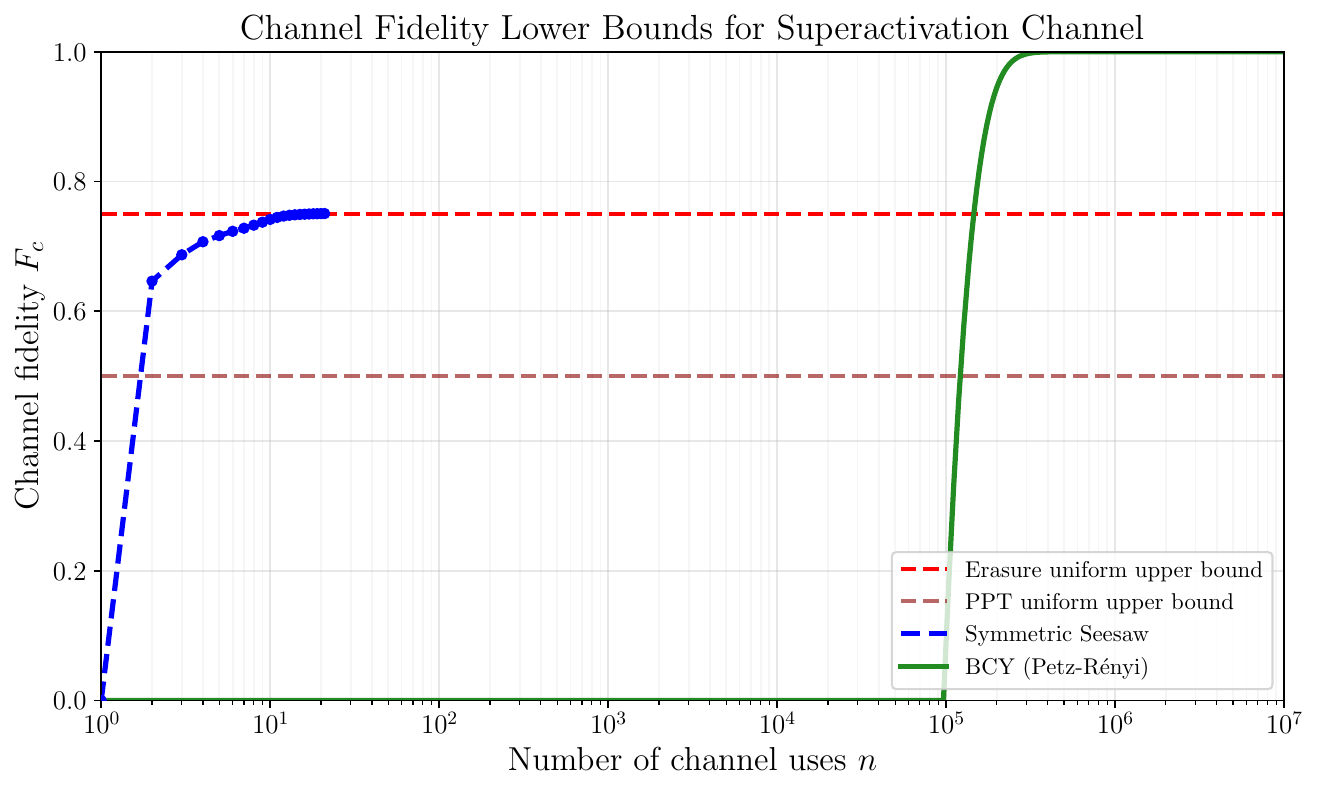}
    \caption{Channel fidelity lower bounds for the superactivation joint
    channel $(\mathcal{P}\otimes \mathcal{A})^{\otimes n}$ as a function
    of the number of channel uses $n$ (log-scale) for a fixed input
    dimension $d = 2$. The green curve shows the Petz--R\'enyi lower
    bound of \cite{berta2026tight} applied to $\widetilde{\mathcal{N}}$
    (equation~\eqref{eq:error_exponent_bound_superactivation_special_case}). Comparison with the uniform upper
    bounds for $\mathcal{P}$ and $\mathcal{A}$ allows one to prove
    non-asymptotic superactivation at blocklengths of order
    $n \approx 10^5$: the erasure threshold $\varepsilon = 1/4$ is
    crossed at $n \approx 1.5 \times 10^5$, while the PPT threshold
    $\varepsilon = 1/2$ is crossed at $n \approx 1.2 \times 10^5$. The
    blue dashed curve corresponds to the achieved fidelity using the
    symmetric seesaw method, also provided in the main text in
    Figure~\ref{fig:main}.}
    \label{fig:fidelity_bounds}
\end{figure}

In the same figure, we also showed the lower bound on $F_c(\tilde{\mathcal{\mathcal{N}}}^{\otimes n},2)$ coming from the symmetric seesaw method. This comparison with analytical bounds motivates the numerical strategy described below, based on SDP optimization.

\subsection{Symmetric Seesaw Method}

This section introduces the numerical method adopted in the analysis. The general method is laid out in detail in \cite{Bergh2026}, of which we give a brief summary here, also focusing more on the very particular problem we want to apply it to. To introduce the method, we first review the seesaw algorithm and the power iteration method originally proposed by~\cite{Reimpell2005}, and then present a symmetric variant of it which reduces the optimization to permutation-invariant (PI) codes, so as to make the optimization scale polynomially with the number of channel uses $n$.  While PI codes need not be optimal, the restriction to PI codes is a computationally motivated ansatz that has proven very effective in practice: Bhalerao and Leditzky~\cite{BhaleraLeditzky2025} recently showed that restricting to permutation invariant states in optimizing the coherent information significantly improves the best known lower bounds on the quantum capacity of several channel families (e.g., the 2-Pauli and BB84 channels), outperforming previous results. In this section, we show that PI codes are similarly effective for the finite-blocklength regime, yielding improved lower bounds on the channel fidelity of the channel \eqref{eq:equivalent_channel} involved in superactivation.


\subsubsection{Seesaw Iteration Method}
The seesaw method, or alternating convex search method, is an iterative numerical procedure for finding lower bounds on bilinear optimization problems like~\eqref{eq:channel_fidelity_SM} by iteratively optimizing over one variable while keeping the other fixed. This simple method gives no guarantee about the optimality of the solution nor estimates of the error, but it often works well in practice and typically converges to a value close to the actual optimum.

The intuition behind the seesaw method can be explained straightforwardly by expanding~\eqref{eq:channel_fidelity_SM}:
\begin{align}
        F_c(\mathcal{N}^{\otimes n}, d) &= \displaystyle \max_{\mathcal{E}_n,\mathcal{D}_n} \Tr[\Phi^d_{RB'} \Phi_{RB'}^{\mathcal{D}_n\circ\mathcal{N}^{\otimes n}\circ\mathcal{E}_n}] \\
        \label{eq:seesaw_first}
        &= \displaystyle \max_{\mathcal{E}_n,\mathcal{D}_n} \Tr[\Phi_{RB^n}^{\mathcal{D}^*_n}\Phi_{RB^n}^{\mathcal{N}^{\otimes n}\circ\mathcal{E}_n}] \\
         \label{eq:seesaw_second}
        &= \displaystyle \max_{\mathcal{E}_n,\mathcal{D}_n} \Tr[\Phi_{RA^n}^{(\mathcal{D}_n\circ \mathcal{N}^{\otimes n})^*} \Phi_{RA^n}^{\mathcal{E}_n}].
    \end{align}
Now, suppose that in~\eqref{eq:seesaw_first} we are provided with an encoder $\mathcal{E}_n$ (e.g., a random seed, or coming from a previous iteration). Then, defining $\mathcal{M}_{A' \to B^n}  \equiv \mathcal{N}^{\otimes n}_{A^n \to B^n}\circ\mathcal{E}_{A' \to A^n}$, its Choi matrix $\Phi^{\mathcal{M}}_{RB^n}$ is known, and the optimization problem becomes:\begin{equation}
    F_D(\mathcal{M}_n) \coloneqq \displaystyle \max_{\mathcal{D}_n} \Tr[\Phi_{RB^n}^{\mathcal{D}^*_n}\Phi_{RB^n}^{\mathcal{M}_n}].
\end{equation}
This is an SDP with input $\Phi^{\mathcal{M}}_{RB^n}$ with respect to the Choi state \footnote{The Choi state of a CPU map need not be normalized, but we maintain the terminology not to overload the notation.} of the adjoint of the decoder map:\begin{equation}
\label{eq:SDP_Maximal_Recovery_Coefficient}
\boxed{
\begin{aligned}
\textbf{Primal:}\quad & \\
\max_{\Phi^{\mathcal{D}_n^*}_{RB^n}} \quad & \Tr\!\big(\Phi_{RB^n}^{\mathcal{M}_n}\, \Phi^{\mathcal{D}_n^*}_{RB^n}\big) \\
\textrm{s.t.} \quad & \Tr_R[\Phi^{\mathcal{D}_n^*}_{RB^n}] = \frac{\mathbbm{1}_{B^n}}{d} \\
& \Phi^{\mathcal{D}_n^*}_{RB^n}\geq 0
\end{aligned}
\qquad\qquad
\begin{aligned}
\textbf{Dual:}\quad & \\
\min_{Y_{B^n}} \quad & \frac{1}{d}\,\Tr(Y_{B^n}) \\
\textrm{s.t.} \quad & \Phi_{RB^n}^{\mathcal{M}_n} \leq \mathbbm{1}_R \otimes Y_{B^n} \\
&
\end{aligned}
}
\end{equation}
This SDP is intimately related with the notion of singlet fraction~\cite{Horodecki1999} (and thus also related to conditional-min entropy~\cite{Koenig2009Operational}). In particular, defining the  \textit{maximal singlet fraction} of a bipartite state $\rho_{AB}$ as:\begin{equation}
    \label{eq:maximal_singlet_fraction}
    F_\Phi(\rho_{AB}) = \max_{\mathcal{R}_{B \to A}\in \operatorname{CPTP}}\bra{\Phi^{d_A}}(\id_A \otimes \mathcal{D}_{B \to A})(\rho_{AB}) \ket{\Phi^{d_A}},
\end{equation}one has:\begin{align}
    \label{eq:recovery_as_singlet_fraction}
       F_{D}(\mathcal{M}_n) =F_{\Phi^d}(\Phi_{RB^n}^{\mathcal{M}_n}),
       \end{align}
where $d = d_{A'} = d_R$ is the input dimension of $\mathcal{M}_n$.

Conversely, suppose that in~\eqref{eq:seesaw_second} we are provided with a decoder $\mathcal{D}_n$ (e.g., a random seed or coming from a previous iteration): then, defining $\mathcal{M}'_{A^n \to B'}  \equiv \mathcal{D}_{B^n \to B'} \circ \mathcal{N}^{\otimes n}_{A^n \to B^n}$, we have that $\Phi^{\mathcal{M}'^*}_{RA^n} = \Phi^{(\mathcal{D}_n \circ \mathcal{N}^{\otimes n})^*}_{RA^n}$, and the optimization problem becomes:\begin{equation}
    F_E(\mathcal{M}'_n) \coloneqq \displaystyle \max_{\mathcal{E}_n} \Tr[\Phi_{RA^n}^{\mathcal{M}_n'^*} \Phi_{RA^n}^{\mathcal{E}_n}].
\end{equation}
This is an SDP with input $\Phi^{\mathcal{M}'^*}_{RA^n}$ with respect to the Choi state of the encoder map:
\begin{equation}
\label{eq:SDP_Maximal_Preparation_Coefficient}
\boxed{
\begin{aligned}
\textbf{Primal:}\quad & \\
\max_{\Phi^{\mathcal{E}_n}_{RA^n}} \quad & \Tr\!\big(\Phi_{RA^n}^{\mathcal{M}'^*_n}\, \Phi^{\mathcal{E}_n}_{RA^n}\big) \\
\textrm{s.t.} \quad & \Tr_{A^n}[\Phi^{\mathcal{E}_n}_{RA^n}] = \frac{\mathbbm{1}_{R}}{d} \\
& \Phi^{\mathcal{E}_n}_{RA^n}\geq 0
\end{aligned}
\qquad\qquad
\begin{aligned}
\textbf{Dual:}\quad & \\
\min_{Y_{R}} \quad & \frac{1}{d}\,\Tr(Y_{R}) \\
\textrm{s.t.} \quad & \Phi_{RA^n}^{\mathcal{M}'^*_n} \leq Y_{R} \otimes \mathbbm{1}_{A^n}  \\
&
\end{aligned}
}
\end{equation}
The quantities in~\eqref{eq:SDP_Maximal_Recovery_Coefficient} and~\eqref{eq:SDP_Maximal_Preparation_Coefficient} correspond respectively to the \textit{maximal fidelity of recovery}, i.e., the optimal entanglement fidelity that can be achieved via a local operation at the decoder's side~\cite{barnum2002reversing}, and the \textit{maximal fidelity of preparation}, i.e.~the optimal entanglement fidelity that can be achieved via a local operation at the encoder's side. These quantities are essentially equivalent, in the sense that for a given channel $\mathcal{W}_{A \to B}$, one has:\begin{equation}
\label{eq:equivalence-F_D-F_E}
    F_D(\mathcal{W}) = \frac{d^2_B}{d^2_A}F_E(\mathcal{W}).
\end{equation}
This follows by using~\eqref{eq:adjoint_choi} on the unnormalized Choi matrices:\begin{align}
    F_E(\mathcal{W}) &= \frac{1}{d_B^2}\max_{\mathcal{E_{B \to A} \in \mathrm{CPTP}}}\Tr[ \Gamma^{\mathcal{W}^*}_{BA}\Gamma^{\mathcal{E}}_{BA}]\\
    &=\frac{1}{d_B^2}\max_{\mathcal{E_{B \to A} \in \mathrm{CPTP}}}\Tr[ \Gamma^{\mathcal{W}}_{AB}\Gamma^{\mathcal{E}^*}_{AB}]\\
    &=\frac{1}{d_B^2}\max_{\mathcal{D_{B \to A} \in \mathrm{CPTP}}}\Tr[ \Gamma^{\mathcal{D}^*}_{AB}\Gamma^{\mathcal{W}}_{AB}] \\
    &= \frac{d^2_A}{d^2_B} F_D(\mathcal{W}).
\end{align}

More details on the meaning and properties of these quantities, which form the core of the seesaw method, are discussed in~\cite{Parentin2025Thesis}. While equivalent (up to a normalization factor, in the sense of~\eqref{eq:equivalence-F_D-F_E}), since the SDPs~\eqref{eq:SDP_Maximal_Recovery_Coefficient} and~\eqref{eq:SDP_Maximal_Preparation_Coefficient} will be solved for varying input channels of different dimensions, we will use a distinct notation and treat the two separately.

We now formally introduce the seesaw algorithm~\cite{Reimpell2005,Fletcher2007,Taghavi_2010,johnson2017qvectoralgorithmdevicetailoredquantum}. 

\begin{definition}
\label{def:seesaw}
Given a quantum channel $\mathcal{N}_{A \to B}$, a number of channel uses $n \in \mathbb{N}$, an input dimension $d \leq \min\{d_A^n, d_B^n\}$, and a convergence threshold $\delta > 0$, in the setting of Figure~\ref{fig:quantumcommunicationquantumchannel}, the seesaw iteration method produces a lower bound $\tilde{F} \leq F_c(\mathcal{N}^{\otimes n}, d)$ together with an encoding--decoding pair $(\tilde{\mathcal{E}}, \tilde{\mathcal{D}})$ that achieves it. The algorithm proceeds as follows.

\begin{enumerate}
    \item \textbf{Initialization.} Choose a random encoder $\mathcal{E}_0 \in \mathrm{CPTP}(A' \to A^n)$ and a random decoder $\mathcal{D}_0 \in \mathrm{CPTP}(B^n \to B')$ as initial seeds. Set $i = 0$.

    \item \textbf{Iteration.} Repeat:
    \begin{enumerate}
        \item \textit{Optimize decoder.} Solve the SDP~\eqref{eq:SDP_Maximal_Recovery_Coefficient} for the channel $\mathcal{M}_i = \mathcal{N}^{\otimes n} \circ \mathcal{E}_{i}$, obtaining the optimal decoder $\mathcal{D}_{i+1} \in \mathrm{CPTP}(B^n \to B)$ and the value $F_{D}(\mathcal{M}_i)$.

        \item \textit{Optimize encoder.} Solve the SDP~\eqref{eq:SDP_Maximal_Preparation_Coefficient} for the channel $\mathcal{M}_i' =\mathcal{D}_{i+1} \circ \mathcal{N}^{\otimes n}$, obtaining the optimal encoder $\mathcal{E}_{i+1}\in \mathrm{CPTP}(A' \to A^n)$ and the value $F_{E}(\mathcal{M}_i')$.
        \item \textit{Convergence check.} If
        \begin{equation}
            F_{E}(\mathcal{M}_i') - F_{D}(\mathcal{M}_i) < \delta,
        \end{equation}
        terminate. Otherwise, set $i \leftarrow i+1$ and repeat.
    \end{enumerate}
    \item \textbf{Output.} Return $\tilde{F} \coloneqq F_e(\mathcal{D}_{i+1} \circ \mathcal{N}^{\otimes n} \circ \mathcal{E}_{i+1})$ and the maps $(\mathcal{E}, \mathcal{D}) \coloneqq (\mathcal{E}_{i+1}, \mathcal{D}_{i+1})$ (in the Choi representation).
\end{enumerate}
\end{definition}
The seesaw method, originally introduced to evaluate the performance of quantum error correcting codes, only guarantees convergence to a local maximum, since it features a monotonic increase in the entanglement fidelity at every half-step. The locally optimal solution $\tilde{F}$ depends on the initial random seeds $(\mathcal{E}_0, \mathcal{D}_0)$, and hence it is inherently stochastic. To obtain better results — ideally close to the global optimum  $F_c( \mathcal{N}^{\otimes n} , d)$— one should repeat the method with multiple random initializations and take the maximum over the resulting fidelities.

\

The entire optimization procedure can be formulated using only the Choi matrices of the involved channels, exploiting the concatenation and adjoint rules for Choi matrices (\cref{eq:concatenation_choi,eq:choi_tensor_product,eq:adjoint_choi}), and essentially consists of repeatedly solving the SDPs $F_D(\mathcal{M}_n)$ and $F_E(\mathcal{M}'_n)$ for varying channel inputs, updated at every step.  

Besides SDP solvers, such as MOSEK~\cite{mosek} and SCS~\cite{SCS}, each half-step~\eqref{eq:SDP_Maximal_Recovery_Coefficient} and~\eqref{eq:SDP_Maximal_Preparation_Coefficient} of the seesaw iteration can be solved via the \textit{channel power iteration method} introduced by Reimpell and Werner~\cite{Reimpell2005}. The method is inspired by the classical power method~\cite{vonMises1929Power} for computing the largest eigenvalue and the corresponding eigenvector of a positive semidefinite matrix $A$, by iteratively applying $v \mapsto Av/\|Av\|$. The channel analogue replaces the linear map $v \mapsto Av$ with a sandwiching operation on the Choi matrix of the encoder (resp. decoder adjoint), followed by a normalization that enforces the trace-preservation (resp. unitality) constraint. 

\begin{definition}
\label{def:power_method}
Given $n \in \mathbb{N}$, a quantum channel $\mathcal{M}_{A' \to B^n} = \mathcal{N}^{\otimes n}_{A^n \to B^n} \circ \mathcal{E}_{A' \to A^n}$, an initial decoder $\mathcal{D}_n \in \mathrm{CPTP}(B^n \to B')$, and a convergence threshold $\delta_p > 0$, the channel power iteration produces an estimate of $F_{D}(\mathcal{M})$ and a decoder $\mathcal{D}$ that achieves it. Set $j = 0$ and $\mathcal{D}_0 = \mathcal{D}$ and repeat:
\begin{enumerate}
    \item \textbf{Apply.} Compute the unnormalized Choi matrix:
    \begin{equation}
    \label{eq:power_step}
        \Gamma_{RB^n}' = \Gamma^{\mathcal{M}}_{RB^n}  \Gamma^{\mathcal{D}^*_j}_{RB^n}\Gamma^{\mathcal{M}}_{RB^n}.
    \end{equation}
    \item \textbf{Normalize.} Compute $S_{B^n} = \Tr_{R}[\Gamma'_{RB^n}]$ and obtain the Choi matrix of the normalized decoder $\mathcal{D}^*_{j+1}$ by applying $S_{B^n}^{-1/2}$ (pseudoinverse) blockwise, i.e.:
    \begin{equation}
    \label{eq:normalization_choi}
        \Gamma^{\mathcal{D}^*_{j+1}}_{RB^n} \coloneqq (\mathbbm{1}_R \otimes S_{B^n}^{-1/2})\Gamma_{RB^n}' (\mathbbm{1}_R \otimes S_{B^n}^{-1/2}).
    \end{equation} 
    This enforces the unitality constraint $\Tr_{R}[\Gamma^{\mathcal{D}^*_{j+1}}_{RB^n}] = \mathbbm{1}_{B^n}$.

    \item \textbf{Convergence check.} If $F_e(\mathcal{D}_{j+1} \circ \mathcal{M}) - F_e(\mathcal{D}_j \circ \mathcal{M}) < \delta_p$, terminate. Otherwise, set $j \leftarrow j+1$ and repeat.
\end{enumerate}
Return $\mathcal{D} \equiv \mathcal{D}_{j+1}$ and $\tilde{F}\approx F_e(\mathcal{D} \circ \mathcal{M})$.
\end{definition}
Under suitable regularity conditions (which are satisfied in our case), the method satisfies monotonicity of the entanglement fidelity:\begin{equation}
    F_e(\mathcal{D}_{j+1} \circ \mathcal{M}) \geq F_e(\mathcal{D}_j \circ \mathcal{M}) \quad \forall j \in \mathbb{N},
\end{equation}
and any globally optimal decoder is a fixed point of the iteration; we refer to Reimpell's dissertation~\cite{reimpell2008quantum} for the convergence analysis. The method can be equally applied to estimate $F_{E}(\mathcal{M}')$ by working with the adjoint channel in the normalization step~\eqref{eq:normalization_choi}. As shown in \cite[Sections~3.2.1 and 3.2.2]{reimpell2008quantum}, the power method provides a significant computational advantage over standard SDP solvers, and it can be efficiently implemented on a GPU architecture, so as to benefit from parallelization and fast sparse linear algebra routines.

However, both methods share the same fundamental bottleneck: by~\eqref{eq:choi_tensor_product}, the Choi matrix $\Gamma^{\mathcal{N}^{\otimes n}}$ grows as $(d_A d_B)^{2n}$. 
Due to this exponential scaling, a preliminary investigation~\cite{Parentin2025Thesis} using the standard seesaw method in~\eqref{def:seesaw} for our channel of interest~\eqref{eq:equivalent_channel} and $n =1, \ldots, 8$, while promising, was not sufficient to establish $n$-shot superactivation. This motivates the following section, which describes a method to increase the number of channel uses, $n$, by exploiting symmetry.

\subsubsection{Permutation-Invariant Codes and Symmetric Seesaw}
\label{sec:symmetric_seesaw}

In order to tackle the exponential scaling with $n$, this section introduces a variant of Definition~\ref{def:seesaw} that we call the \textit{symmetric seesaw method}, which exploits permutation symmetry to reduce the complexity from exponential to polynomial in~$n$, using the representation theory of the symmetric group. This procedure is fairly standard today (see, e.g.,~\cite{Litjens_2016, gijswijt2009, KLERK20101} as general references) and has been applied to various quantum information problems, including bounds on channel capacities~\cite{fawzi2022hierarchy,singh2025}, quantum state and channel discrimination~\cite{Bergh2024,bergh2025discrimination,cheng2025invitation} and efficient approximations to regularized relative entropies~\cite{fang2025efficient}. The recent works~\cite{CTV2023, kossmann2025approximatequantumerrorcorrection} are of particular relevance as they solve a connected problem of finding upper bounds on $F_c(\mathcal{N}, d)$ exploiting similar symmetry considerations. 

\

Given the scheme in Figure~\ref{fig:quantumcommunicationquantumchannel}, we say that an encoder $\mathcal{E}_{A' \to A^n}$ is \textit{symmetric} if $\Gamma_{RA^n}^{\mathcal{E}} \in \mathrm{End}^{\mathfrak{S}_n}(R \otimes A^{\otimes n})$; analogously, a decoder $\mathcal{D}_{B^n \to B'}$ is symmetric if $\Gamma_{B^nR}^{\mathcal{D}} \in \mathrm{End}^{\mathfrak{S}_n}(B^{\otimes n}
\otimes R)$. A pair consisting of a symmetric encoder and decoder is also known as a \textit{permutation-invariant (PI) code}~\cite{Ruskai2000, PollatsekRuskai2004}. The key idea underlying PI codes is that i.i.d.~channels $\mathcal{N}^{\otimes n}$ are permutation covariant, and therefore map permutation-invariant inputs to permutation-invariant outputs.

While the bilinear optimization problem in~\eqref{eq:channel_fidelity_SM} cannot be restricted to PI codes, in this work we decide nonetheless to impose separate permutation invariance as a computational ansatz:
\begin{equation}
\label{eq:symmetric_enc-dec}
    \Gamma^{\mathcal{E}_n}_{RA^n} \in \mathrm{End}^{\mathfrak{S}_n}(R \otimes A^{\otimes n}), \qquad \Gamma^{\mathcal{D}_n^*}_{RB^n} \in \mathrm{End}^{\mathfrak{S}_n}(R \otimes B^{\otimes n}).
\end{equation}
This restriction reduces the optimization space in~\eqref{eq:SDP_Maximal_Recovery_Coefficient} and~\eqref{eq:SDP_Maximal_Preparation_Coefficient} from exponential to polynomial in $n$ using the representation theory of the symmetric group. While potentially excluding some optimal codes, this ansatz will work sufficiently well for our application. 
The key structural result underlying our method is that the seesaw iteration \textit{preserves} the symmetric subspace, in a sense specified by the following lemma.

\begin{lemma}
\label{lemma:seesaw_symmetric}
Let $n \in \mathbb{N}$ be a number of channel uses, and let $d \in \mathbb{N}$ be an input dimension. Given a permutation-covariant channel $\mathcal{N}_n$, i.e.~such that $\Gamma^{\mathcal{N}_n}_{A^nB^n} \in \mathrm{End}^{\mathfrak{S}_n}(A^{\otimes n} \otimes B^{\otimes n})$, if the encoder/decoder seeds $(\mathcal{E}_0,\mathcal{D}_0)$ are symmetric in the sense of~\eqref{eq:symmetric_enc-dec}, then the entire seesaw iteration for $F_c(\mathcal{N}_n, d)$ can be restricted to the symmetric subspace. In particular, the SDPs~\eqref{eq:SDP_Maximal_Recovery_Coefficient} and~\eqref{eq:SDP_Maximal_Recovery_Coefficient} can be restricted to the space $\mathrm{End}^{\mathfrak{S}_n}(R \otimes S^{\otimes n})$, where $S = A$ for $F_E$ and $S = B$ for $F_D$.
\end{lemma}
\begin{proof}
We proceed by induction on the iteration number $i\in \mathbb{N}$. The base case $i = 0$ holds by assumption. Given a generic iteration $i>0$, assume that $(\mathcal{E}_{i-1},\mathcal{D}_{i-1})$ are symmetric in the sense of~\eqref{eq:symmetric_enc-dec}. By permutation covariance of $\mathcal{N}_n$, using~\eqref{eq:concatenation_choi}, the Choi matrix of the concatenated channel $\mathcal{M}_i = \mathcal{N}_n \circ \mathcal{E}_{i-1}$ satisfies:\begin{equation}
\label{eq:concatenated_channek_is_symmetric_proof}
    \Gamma^{\mathcal{M}_i}_{RB^n} \in \mathrm{End}^{\mathfrak{S}_n}(R \otimes
B^{\otimes n}).
\end{equation}
Then the optimal solution of~\eqref{eq:SDP_Maximal_Recovery_Coefficient} can be restricted to lie in $\mathrm{End}^{\mathfrak{S}_n}(R \otimes
B^{\otimes n})$. Indeed, let $\Gamma^{\mathcal{D}^*_i}_{RB^n}$ be feasible for the SDP~\eqref{eq:SDP_Maximal_Recovery_Coefficient} and define its symmetrization 
\begin{equation}
\overline{\Gamma^{\mathcal{D}^*_i}}\coloneqq \frac{1}{n!}\sum_{\pi \in \mathfrak{S}_n}(\mathbbm{1}_R \otimes P_{B^n}(\pi))\Gamma^{\mathcal{D}^*_i}(\mathbbm{1}_R \otimes P_{B^n}(\pi))^\dagger .    
\end{equation}
Then, $\overline{\Gamma^{\mathcal{D}^*_i}} \ge 0$ because it is a convex combination of positive semidefinite operators.
Moreover, using linearity of the trace and the assumption $\mathrm{Tr}_R[\Gamma^{\mathcal{D}^*_i}_{RB^n}]=\mathbbm{1}_{B^n}$, we find that
\begin{equation}
\mathrm{Tr}_R[\overline{\Gamma_{RB^n}^{\mathcal{D}^*_i}}] = \frac{1}{n!}\sum_{\pi} P_{B^n}(\pi)\, \mathbbm{1}_{B^n} P_{B^n}(\pi)^\dagger = \mathbbm{1}_{B^n}.
\end{equation}
Finally, by linearity and cyclicity of the trace:
\begin{align}
\mathrm{Tr}\!\left[ \Gamma^{\mathcal{M}_i}_{RB^n}\, \overline{\Gamma_{RB^n}^{\mathcal{D}^*_i}} \right]
&= \frac{1}{n!}\sum_{\pi} \mathrm{Tr}\!\left[ \Gamma^{\mathcal{M}_i}_{RB^n} (\mathbbm{1}_R \otimes U(\pi)) \Gamma_{RB^n}^{\mathcal{D}^*_i} (\mathbbm{1}_R \otimes U(\pi))^\dagger \right] \\
&= \mathrm{Tr}\!\left[ \overline{\Gamma^{\mathcal{M}_i}_{RB^n}} \Gamma_{RB^n}^{\mathcal{D}^*_i} \right] =\mathrm{Tr}\!\left[ \Gamma^{\mathcal{M}_i}_{RB^n} \Gamma_{RB^n}^{\mathcal{D}^*_i} \right],
\end{align}
where the last equality follows from~\eqref{eq:concatenated_channek_is_symmetric_proof}. The same reasoning can be applied to the encoder step: $\mathcal{D}_{i}$ symmetric and $\mathcal{N}_n$ permutation-covariant imply
$\Gamma^{\mathcal{M}_i'^*}_{RA^n} = \Gamma^{(\mathcal{D}_{i} \circ \mathcal{N}_n)^*}_{RA^n} \in
\mathrm{End}^{\mathfrak{S}_n}(R \otimes A^{\otimes n})$. Then, following the same steps as before for the SDP~\eqref{eq:SDP_Maximal_Preparation_Coefficient}, it follows that $\mathcal{E}_{i}$ can be chosen symmetric. We have shown that at every iteration $i$, if the encoding/decoding pair $(\mathcal{E}_{i-1}, \mathcal{D}_{i-1})$ is symmetric in the sense of~\eqref{eq:symmetric_enc-dec}, then also an optimal $(\mathcal{E}_{i}, \mathcal{D}_{i})$ is. 
\end{proof}

Motivated by Lemma \ref{lemma:seesaw_symmetric}, we call the seesaw method for which the original seeds are restricted to lie in the permutation-invariant subspace~\eqref{eq:symmetric_enc-dec} as the \textit{symmetric seesaw method}.

\begin{definition}
\label{def:symmetric-seesaw}
Given a quantum channel $\mathcal{N}_{A \to B}$, a number of channel uses $n \in \mathbb{N}$, an input dimension $d \leq \min\{d_A^n, d_B^n\}$, and a convergence threshold $\delta > 0$, the \textit{symmetric seesaw method} proceeds exactly as in Definition~\ref{def:seesaw}, with the only additional requirement that the initial seeds $(\mathcal{E}_0, \mathcal{D}_0)$ satisfy~\eqref{eq:symmetric_enc-dec}. It returns a lower bound $\tilde{F}\leq F_c(\mathcal{N}^{\otimes n}, d)$ together with a symmetric encoding / decoding pair that achieves it.
\end{definition}

Due to Lemma~\ref{lemma:seesaw_symmetric}, the two (primal) SDPs~\eqref{eq:SDP_Maximal_Recovery_Coefficient} and~\eqref{eq:SDP_Maximal_Preparation_Coefficient} in the symmetric seesaw method become equivalent to:
\begin{equation}
\begin{aligned}
\label{eq:Symmetric_SDP_Maximal_Recovery_Coefficient}
F_D^S (\mathcal{M}_{A' \to B^n}) \coloneqq \frac{1}{d^2}\max_{\Gamma^{\mathcal{D}_n^*}_{RB^n}} \quad & \Tr\big[\Gamma_{RB^n}^{\mathcal{M}_n}\, \Gamma^{\mathcal{D}_n^*}_{RB^n}\big] \\
\textrm{s.t.} \quad & \Tr_R[\Gamma^{\mathcal{D}_n^*}_{RB^n}] = \mathbbm{1}_{B^n}, \quad \Gamma^{\mathcal{D}_n^*}_{RB^n}\geq 0, \quad \Gamma^{\mathcal{D}_n^*}_{RB^n} \in \mathrm{End}^{\mathfrak{S}_n}(R \otimes B^{\otimes n}),
\end{aligned}
\end{equation}
for $\mathcal{M}_{A' \to B^n} = \mathcal{N}^{\otimes n}_{A^n \to B^n} \circ \mathcal{E}_{A' \to A^n}$, and:
\begin{equation}
\begin{aligned}
\label{eq:Symmetric_SDP_Maximal_Preparation_Coefficient}
F_E^S (\mathcal{M}'_{A^n \to B'}) \coloneqq \frac{1}{d^2}\max_{\Gamma^{\mathcal{E}_n}_{RA^n}} \quad & \Tr\big[\Gamma_{RA^n}^{\mathcal{M}'^*_n}\, \Gamma^{\mathcal{E}_n}_{RA^n}\big] \\
\textrm{s.t.} \quad & \Tr_{A^n}[\Gamma^{\mathcal{E}_n}_{RA^n}] = \mathbbm{1}_R, \quad \Gamma^{\mathcal{E}_n}_{RA^n} \geq 0, \quad  \Gamma^{\mathcal{E}_n}_{RA^n} \in \mathrm{End}^{\mathfrak{S}_n}(R \otimes A^{\otimes n}),
\end{aligned}
\end{equation}
for $\mathcal{M}'_{A^n \to B'} = \mathcal{D}_{A^n \to B'} \circ \mathcal{N}^{\otimes n}_{A^n \to B^n} $. 

The symmetry conditions~\eqref{eq:symmetric_enc-dec} are equivalent to requiring that the one-copy reduced maps $\mathcal{E}^{(1)} \equiv \mathrm{Tr}_{A^{n-1}} \circ \mathcal{E}_n$ and $\mathcal{D}^{*(1)} \equiv \mathrm{Tr}_{B^{n-1}} \circ \mathcal{D}^*_n$ are $n$-extendible. The two SDPs in~\eqref{eq:Symmetric_SDP_Maximal_Recovery_Coefficient} and~\eqref{eq:Symmetric_SDP_Maximal_Preparation_Coefficient} therefore optimize over the same feasible sets as the $n$-extendibility constraints considered, e.g., in~\cite{Kaur2021}. However, the objective functions are evaluated on the $n$-copy extensions, and thus our SDPs are not directly comparable to theirs.

While the dimension of $\text{End}^{\mathfrak{S}_n}(R: S^{n})$ is $m_S\cdot d_R^2$, where\begin{equation}
\label{eq:dimensions-perm-inv}
    m_S \coloneqq \dim \mathrm{End}^{\mathfrak{S}_n}(S^{\otimes n}) =  \binom{n + d_{S}^2 -1}{d_{S}^2 -1}\leq (n+1)^{d_S^2},
\end{equation} 
this does not yet show that the symmetric seesaw method has an efficient (poly($n$)) solution, because we first need to phrase the optimization problems~\eqref{eq:Symmetric_SDP_Maximal_Recovery_Coefficient} and~\eqref{eq:SDP_Maximal_Preparation_Coefficient} and the intermediate operations of the seesaw method in Definition~\ref{def:seesaw} using matrices of size poly($n$). For this purpose, the idea is to construct a suitable basis of $\text{End}^{\mathfrak{S}_n}(R: S^{n})$ and rephrase the SDPs as ones in which we minimize over now $O(\text{poly}(n))$ many basis coefficients. A natural basis of $\text{End}^{\mathfrak{S}_n}(S^{n})$ is called the orbit basis $\{C_r^S\}_{r \in [m_S]}$, which is an orthogonal basis indexed by the orbits of $\mathfrak{S}_n$.
Given an operator $X\in \mathcal{L}(\mathcal{H})$, the coefficients of its tensor product $X^{\otimes n} \in \mathrm{End}^{\mathfrak{S}_n}(\mathcal{H}^{\otimes n})$ with respect to the orbit basis can be computed straightforwardly from $X$ by picking the representative of orbit $r$, call it $(\underline{i_r}, \underline{j_r})$, and computing the product of the different $n$ elements of the original matrix in the given position:
\begin{equation}
\label{eq:coeff_tensor_product}
    c_r = \prod_{k=1}^n X_{({i_r}_k, {j_r}_k)}.
\end{equation}
Thus:
\begin{equation}
\label{eq:tensor_product_orbit}
    X^{\otimes n} = \sum_{r = 1}^{m_{\mathcal{H}}}c_r\, C_r^{\mathcal{H}}.
\end{equation}
Since only \textit{one} representative of each orbit is required, these coefficients are computable in $O(\mathrm{poly}(n))$ time. Therefore, all operators involved in the symmetric seesaw iteration method ($\mathcal{E}_n$, $\mathcal{D}^*_n$, $\mathcal{M}_n$, and $\mathcal{M}_n'^*$) can be written in the form:\begin{equation}
\label{eq:generic_perm_inv_input}
    X_{RS^n} =  \sum_{k,l = 0}^{d_R-1}\sum_{r=0}^{m_S -1} x_{k,l,r} \ket{k}\!\bra{l} \otimes C_r^S,
\end{equation}
where $x_{k,l,r} \in \mathbb{C}$, $S = B$ in~\eqref{eq:Symmetric_SDP_Maximal_Recovery_Coefficient}, and $S = A$ in~\eqref{eq:Symmetric_SDP_Maximal_Preparation_Coefficient}. While the orbit matrices $C_r^S$ have an exponential dimension in $n$, many of their properties can be computed without explicitly constructing them. Specifically, their trace $\Tr[C_r^S]$ and their Hilbert--Schmidt inner product $\Tr[(C_r^S)^\dag C_r^S]$ can be computed efficiently, i.e.~in poly($n$) time, using a combinatorial formula \cite{bergh2025discrimination}. Moreover, the intermediate steps of the symmetric seesaw method, namely, the concatenation and adjoint of symmetric channels, can also be performed efficiently in the orbit basis. The details of the numerical method, alongside with the code implementation, will be given in a separate paper \cite{Bergh2026}.

In order to efficiently solve the SDPs, we should phrase the positive semidefinite constraints (and the partial trace constraint in~\eqref{eq:Symmetric_SDP_Maximal_Recovery_Coefficient}) without constructing the exponentially large matrices.
This is possible by making use of the representation theory of the symmetric group~\cite{JamesKerber1981Representation, Sagan2001Symmetric}, which allows to decompose them into polynomially many blocks of polynomial size. This procedure, today adopted in several SDP applications beyond quantum information \cite{gijswijt2009, Litjens_2016, vallentin2009symmetry,fawzi2022hierarchy}, is based on constructing a $*$-algebra isomorphism, which maps the elements of $\mathrm{End}^{\mathfrak{S}_n}(R \otimes S^{\otimes n})$ to a block-diagonal form and preserves the products and the positive-semidefinite structure:\begin{equation}
\label{eq:full_isomorphism_with_reference}
     \widetilde{\psi}_{RS^n}:
    \mathrm{End}^{\mathfrak{S}_n}(R \otimes S^{\otimes n})
    \to \bigoplus_{\lambda \in \mathrm{Par}(d_S, n)}
    \mathbb{C}^{(d_R m_\lambda) \times (d_R m_\lambda)} ,
\end{equation}
where $\mathrm{Par}(d_S,n)$ is the set of \textit{partitions} of $n$ with at most $d_S$ parts, and $m_\lambda \coloneq |\mathcal{T}_{\lambda,d_S}|$ is the number of \textit{semistandard Young tableaux} of shape $\lambda$ with entries in $[d_S]$.  Crucially, both the number of blocks and the size of each block are polynomial in $n$, due to the following bounds \cite{Litjens_2016}:
\begin{equation}
    |\mathrm{Par}(d,n)| \leq (n+1)^{d_S}, \qquad m_\lambda \leq (n+1)^{d_S(d_S-1)/2}.
\end{equation}
The isomorphism acts on a generic input $X_{RS^n}$ in~\eqref{eq:generic_perm_inv_input} as $\widetilde{\psi}(X) \coloneqq \bigoplus_\lambda\widetilde{[X]}_\lambda$, acting trivially on the reference system $R$, so that each block can be decomposed into $d_R \cdot d_R$ matrices:\begin{equation}
\label{eq:blocks}
    \widetilde{[X]}_\lambda^{(k,l)}= \sum_{r=1}^{m_S} x_{k,l,r}\, \widetilde{[C_r^S]}_\lambda \qquad k, l \in [d_R],
\end{equation}
where $x_{k,l,r}$ are the orbit basis coefficients and $\widetilde{[C_r^S]}_\lambda$ denotes the image of the generic orbit matrix $C_r^S$ under $\widetilde{\psi}_{S^n}$. To efficiently compute the action of $\widetilde{\psi}$ on the orbit matrices, one exhibits a representative set $\{U_\lambda\}_{\lambda \in \mathrm{Par}(d,n)}$ for the action of the symmetric group $\mathfrak{S}_n$ on $S^{\otimes n}$, using the Young basis $\{|u_\tau\rangle\}_{\tau \in \mathcal{T}_{\lambda,d}}$ for the $\lambda$-isotypic component of $\mathcal{H}^{\otimes n}$. 
It is well-known (see e.g.~\cite{gijswijt2009, Litjens_2016}) that this map can be implemented efficiently (see the separate paper~\cite{Bergh2026} for the explicit construction in our case of interest).

In order to phrase the two SDPs~\eqref{eq:Symmetric_SDP_Maximal_Recovery_Coefficient} and~\eqref{eq:Symmetric_SDP_Maximal_Preparation_Coefficient} in this new \textit{block basis}, indexed by blocks of the form~\eqref{eq:blocks}, we recall that the full Schur--Weyl decomposition of $S^{\otimes n}$ is as follows:\begin{equation}
\label{eq:Shur-Weyl-decomposition}
    S^{\otimes n} \cong \bigoplus_{\lambda \in \mathrm{Par}(d_S,n)} V_\lambda \otimes W_\lambda,
\end{equation} where $V_\lambda$ denotes irreducible representations (irreps) and $W_\lambda$ the corresponding multiplicity space; we have $\dim V_\lambda \coloneqq m_\lambda$ and $\dim W_\lambda \coloneqq f_\lambda$, where $f_\lambda$ equals the number of \textit{standard Young tableaux} of shape $\lambda$ with entries in $[d_S]$ and is computable via the hook-length formula~\cite{Frame1954}. Permutations act non-trivially only on the irreps $V_\lambda$, so in the following we will account for the multiplicity $f_\lambda$ when taking the trace of an operator $A_{S^n}$, i.e., using~\eqref{eq:Shur-Weyl-decomposition}, $\Tr_{S^n}[A_{S^n}] = \sum_\lambda f_\lambda \Tr_{V_\lambda}[A_{S^n}]$. 
\

Under $\widetilde{\psi}_{RS^n}$, the objective functions become weighted sums over isotypic components and the partial trace constraints become:\begin{itemize}
    \item In~\eqref{eq:Symmetric_SDP_Maximal_Recovery_Coefficient} (optimization over CPU maps), the constraint $\mathrm{Tr}_R[\Gamma^{\mathcal{D}_n^*}_{RB^n}] = \mathbbm{1}_{B^n}$ decouples block by block:
\begin{equation}
\label{eq:decoder_partial_trace_block}
    \sum_{k=1}^{d_R}\widetilde{[\Gamma^{\mathcal{D}_n^*}]}_\lambda^{(k,k)}= \mathbbm{1}_{m_\lambda}\qquad \forall\, \lambda \in \mathrm{Par}(d_B, n).
\end{equation}
\item In~\eqref{eq:Symmetric_SDP_Maximal_Preparation_Coefficient} (optimization over CPTP maps), the constraint $\mathrm{Tr}_{A^n}[\Gamma^{\mathcal{E}_n}_{RA^n}] = \mathbbm{1}_R$ couples all blocks into a single $d_R \times d_R$ matrix equation:
\begin{equation}
\label{eq:encoder_partial_trace_block}
    \sum_{\lambda \in \mathrm{Par}(d_A,n)} f_\lambda\, \mathrm{Tr}_{V_\lambda}\!\left(
    \widetilde{[\Gamma^{\mathcal{E}_n}]}_\lambda\right)
    = \mathbbm{1}_R.
\end{equation}
\end{itemize} 

Applying the isomorphism to the SDPs in~\eqref{eq:Symmetric_SDP_Maximal_Recovery_Coefficient} and~\eqref{eq:Symmetric_SDP_Maximal_Preparation_Coefficient}, we obtain an effective map $\widetilde{\Psi}$ which takes them to the resulting SDPs:
\begin{equation}
\label{eq:block_SDP_FD}
\boxed{
\begin{aligned}
    \widetilde{\Psi}( F^S_{D}(\mathcal{M}_n)) = \max_{\{\widetilde{[\Gamma^{\mathcal{D}_n^*}]}_\lambda\}: \ \widetilde{[\Gamma^{\mathcal{D}_n^*}]}_\lambda \in \mathbb{C}^{d_R m_\lambda \times d_R m_\lambda } } \quad& \frac{1}{d^2} \sum_{\lambda\in \mathrm{Par}(d_B, n)} f_\lambda\,\mathrm{Tr}_{V_\lambda}\!\left[\widetilde{[\Gamma^{\mathcal{D}_n^*}]}_\lambda \cdot\widetilde{[\Gamma^{\mathcal{M}_n}]}_\lambda\right] \\\text{s.t.} \quad& \widetilde{[\Gamma^{\mathcal{D}_n^*}]}_\lambda\geq 0\quad \forall\, \lambda \in \mathrm{Par}(d_B, n) \\& \sum_{k=1}^{d_R}\widetilde{[\Gamma^{\mathcal{D}_n^*}]}_\lambda^{(k,k)}= \mathbbm{1}_{m_\lambda}\quad \forall\, \lambda \in \mathrm{Par}(d_B, n)
\end{aligned}}
\end{equation}
for symmetric $\mathcal{M}_{n} = \mathcal{N}^{\otimes n} \circ \mathcal{E}_{n}$, and:
\begin{equation}
\boxed{
\label{eq:block_SDP_FE}
\begin{aligned}
 \widetilde{\Psi}{(F^S_{E}(\mathcal{M}'_n))} = \max_{\{
    \widetilde{[\Gamma^{\mathcal{E}_n}]}_\lambda\}: \   \widetilde{[\Gamma^{\mathcal{E}_n}]}_\lambda \in \mathbb{C}^{d_R m_\lambda \times d_R m_\lambda } }\quad
    & \frac{1}{d^2} \sum_{\lambda\in \mathrm{Par}(d_A, n)} f_\lambda\,
    \mathrm{Tr}_{V_\lambda}\!\left[
    \widetilde{[\Gamma^{\mathcal{E}_n}]}_\lambda \cdot
    \widetilde{[\Gamma^{\mathcal{M}'^*_n}]}_\lambda\right] \\
    \text{s.t.} \quad
    & \widetilde{[\Gamma^{\mathcal{E}_n}]}_\lambda \geq 0
    \quad \forall\, \lambda \in \mathrm{Par}(d_A, n) \\
    & \sum_\lambda f_\lambda\,
    \mathrm{Tr}_{V_\lambda}\!\left(
    \widetilde{[\Gamma^{\mathcal{E}_n}]}_\lambda\right) =\mathbbm{1}_R
\end{aligned}}
\end{equation}
for symmetric $\mathcal{M}'_n= \mathcal{D}_n\circ \mathcal{N}^{\otimes n}$.

Since $\widetilde{\psi}_{RS^n}$ is a $*$-isomorphism, the new SDPs have the same feasible set and objective function as the original ones. The variables in these SDPs are the block matrices $\widetilde{[\Gamma^{\mathcal{D}^*}]}_\lambda$ (resp.\ $\widetilde{[\Gamma^{\mathcal{E}}]}_\lambda$), each of size $d_R m_\lambda \times d_R m_\lambda$.

Since in the symmetric seesaw method, the intermediate steps (concatenation and adjoint) are conveniently performed in the orbit basis, the orbit basis coefficients of $\Gamma^{\mathcal{M}_n}_{RB^n}$ and $\Gamma^{\mathcal{M}'^*_n}_{RA^n}$ are then transformed to the block basis via~\eqref{eq:full_isomorphism_with_reference} before being passed to the SDP solver. After the optimal decoder or encoder is found, the block variables are converted back to orbit-basis coefficients via $\widetilde{\psi}^{-1}$ for use in the next seesaw half-step. These interconversions can be implemented in practice with a fixed overhead by precomputing the matrices that implement the change of basis between orbits and blocks prior to the iteration.

The channel power iteration of Definition~\ref{def:power_method} can be adapted to work entirely in the block basis, as the three steps translate directly to block-wise operations. We describe the decoder case~\eqref{eq:block_SDP_FD}; the encoder case~\eqref{eq:block_SDP_FE} is analogous, up to the appropriate change in the normalization constraint.

\begin{definition}
\label{def:symmetric_power_method_D}
Given a symmetric channel $\mathcal{M}_{A' \to B^n}$ with $\Gamma^{\mathcal{M}}_{RB^n} \in \mathrm{End}^{\mathfrak{S}_n}(R \otimes B^n)$, an initial symmetric decoder $\mathcal{D}_0$, and a threshold $\delta_p > 0$, set $j = 0$ and repeat:
\begin{enumerate}
    \item \textbf{Apply.} For each
    $\lambda \in \mathrm{Par}(d_B, n)$:
    \begin{equation}
    \label{eq:power_apply_block}
        \widetilde{[\Gamma']}_\lambda
        = \widetilde{[\Gamma^{\mathcal{M}}]}_\lambda \cdot
        \widetilde{[\Gamma^{\mathcal{D}^*_j}]}_\lambda \cdot
        \widetilde{[\Gamma^{\mathcal{M}}]}_\lambda.
    \end{equation}

    \item \textbf{Normalize.} Compute
    $S_\lambda \coloneqq \sum_{k=1}^{d_R}
    \widetilde{[\Gamma']}_\lambda^{(k,k)}$ and set:
    \begin{equation}
    \label{eq:power_normalize_block}
        \widetilde{[\Gamma^{\mathcal{D}^*_{j+1}}]}_\lambda
        = (\mathbbm{1}_{d_R} \otimes S_\lambda^{-1/2})\,
        \widetilde{[\Gamma']}_\lambda\,
        (\mathbbm{1}_{d_R} \otimes S_\lambda^{-1/2}),
    \end{equation}
    using the pseudoinverse where needed. This
    enforces~\eqref{eq:decoder_partial_trace_block}:
    $\sum_k \widetilde{[\Gamma^{\mathcal{D}^*_{j+1}}]}_\lambda^{(k,k)}
    = \mathbbm{1}_{m_\lambda}$ for each $\lambda$.

    \item \textbf{Convergence.} Compute:
    \begin{equation}
    \label{eq:fidelity_block_D}
        F_e^{(j+1)} \coloneqq F_e(\mathcal{D}_{j+1}\circ \mathcal{M}) = \frac{1}{d^2}
        \sum_\lambda f_\lambda\,
        \mathrm{Tr}\!\left[
        \widetilde{[\Gamma^{\mathcal{M}}]}_\lambda \cdot
        \widetilde{[\Gamma^{\mathcal{D}^*_{j+1}}]}_\lambda\right].
    \end{equation}
    If $F_e^{(j+1)} - F_e^{(j)} < \delta_p$, terminate with
    $F^S_{D}(\mathcal{M}) \approx F_e^{(j+1)}$.
\end{enumerate}
\end{definition}

For the encoder (Schr\"odinger picture), we just substitute $B \to A$ and enforce the normalization constraint~\eqref{eq:encoder_partial_trace_block} instead, by sandwiching with the operator $(T_R^{-1/2} \otimes \mathbbm{1}_{m_\lambda})$ in step (2), where: \begin{equation}
    \label{eq:encoder_global_trace}
        T \coloneqq \sum_{\lambda \in \mathrm{Par}(d_A,n)} f_\lambda\,
        \mathrm{Tr}_{V_\lambda}\!\left(
        \widetilde{[\Gamma']}_\lambda\right)
        \in \mathbb{C}^{d_R \times d_R}.
    \end{equation}

Each block has size at most $d_R \cdot m_\lambda$, so that all operations in~\eqref{eq:power_apply_block} and~\eqref{eq:power_normalize_block} involve polynomial-size dense matrices. The entire symmetric power iteration runs in $\mathrm{poly}(n)$ time per step and is dominated by the block-wise multiplications in~\eqref{eq:power_apply_block}.
The block decomposition of permutation-invariant operators makes the symmetric power iteration method particularly well suited for computational optimization. Indeed, since the blocks are small and independent, 
all loops over $\lambda$ can be implemented in parallel on a GPU (e.g., using the \texttt{CuPy} package). Consequently, the actual computation time is significantly shorter than that of any SDP solver.
\

For these reasons, the symmetric seesaw method of Definition~\ref{def:symmetric-seesaw}, used in conjunction with the power iteration in Definition~\ref{def:symmetric_power_method_D}, is the fastest and most promising numerical tool to demonstrate $n$-shot superactivation in the sense of~\eqref{eq:n-shot-superact_SM}. However, in order to simplify the numerical analysis and reach values of $n$ up to $20$, as shown in Figure~\ref{fig:main}, after establishing the simpler effective channel $\widetilde{\mathcal{N}}$, we must further simplify the problem.

First, we decide to fix the input dimension to $d = d_R = 2$. This corresponds to the task of qubit transmission in~\eqref{fig:quantumcommunicationquantumchannel}, which is arguably the most relevant in practice, and allows to keep the dimensions of the involved matrices small.

Two more serious bottlenecks are the following:
\begin{enumerate}
    \item The effective channel $\widetilde{\mathcal{N}}_{A \to B}$ has input and output dimensions, respectively, $d_A = 2$ and $d_B = 4$. Therefore, the decoder's SDP~\eqref{eq:SDP_Maximal_Recovery_Coefficient}, whose complexity depends exponentially on $d_B$ (see~\eqref{eq:dimensions-perm-inv}), quickly becomes a bottleneck for large $n$. However, the classical-quantum (CQ) structure of the output allows us to use a simple decoding protocol that reduces the SDP~\eqref{eq:block_SDP_FD} to $n+1$ SDPs on qubit channels.
    
    \item The concatenation steps in the symmetric seesaw method require constructing all orbit basis coefficients for the tensor product channel. Even if polynomial in $n$, a quick calculation shows that the amount of orbits $m_{AB} =\binom{n + d_A^2 d_B^2 - 1}{n}$ scales as $O(n^{15})$ for qubit channels, so it quickly becomes a bottleneck as $n$ increases. However, exploiting the sparsity of the involved channels enables us to compute the corresponding partial trace relations only for the few nonzero orbit basis coefficients $c^{\widetilde{\mathcal{N}}}_s$.
\end{enumerate}
Addressing these two issues is the content of the next two subsections, which when combined allow us to implement the numerical solution that ultimately leads to Figure~\ref{fig:main}.

\subsubsection{Exploiting the Classical--Quantum Output Structure}
\label{sec:CQ_structure}

The effective channel $\widetilde{\mathcal{N}}$ of Proposition~\ref{prop:equivalent_channel} has input dimension $d_A = 2$ and output dimension $d_B = d_Z \cdot d_Q = 4$, where $Z$ is the flag classical system and $Q$ is the quantum output system. The number of free parameters in the SDP~\eqref{eq:block_SDP_FD} depends on $d_B$ as $O(n^{d_B^2}) = O(n^{16})$, so it quickly becomes the bottleneck of the symmetric seesaw method for large $n$. The CQ structure of the output allows us to circumvent this by decomposing the decoder optimization into $n+1$ independent SDPs on qubit channels ($d_Q = 2$), reducing the scaling to $O(n^7)$. Indeed, after $n$ channel uses, the classical flags $z \in \{0,1\}^n$ are available to the decoder, who can condition his strategy on them (see Figure~\ref{fig:CQ_decoding_strategy}).

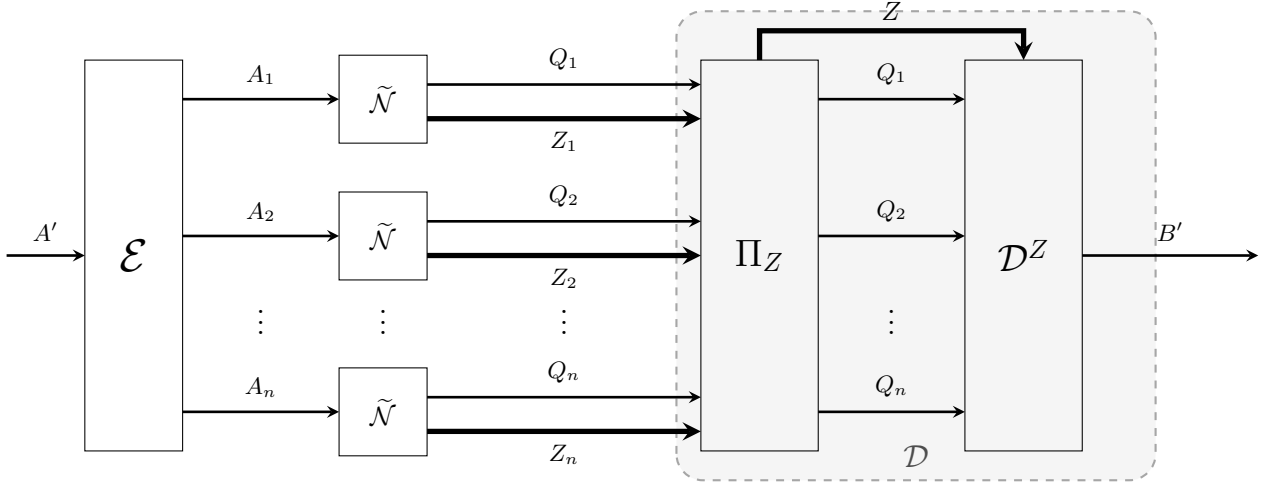
\begin{figure}[!ht]
\centering
\begin{tikzpicture}[>=stealth, scale=1.3]
\tikzset{
    qarrow/.style={->, line width=1pt},
    classarrow/.style={->, line width=2pt},
    block/.style={draw, minimum width=1.2cm, minimum height=4.8cm},
    procblock/.style={draw, minimum width=1.2cm, minimum height=4.8cm},
    channel/.style={draw, minimum width=0.9cm, minimum height=0.9cm},
    chanlabel/.style={font=\normalsize},
    wirelabel/.style={font=\small, above=2pt}
}

\def\xshift{0}

\def\ytop{6.6}      
\def\ymid{5.2}      
\def\ybot{3.4}      
\def\ycenter{5.0}    

\draw[draw=gray!70, thick, dashed, rounded corners=10pt, fill=gray!8]
    ({\xshift+6.05},2.7) rectangle ({\xshift+10.95},7.5);
\node[text=gray!60!black, font=\large] at ({\xshift+8.5},2.95) {$\mathcal{D}$};

\draw[block] (\xshift,3.0) rectangle ({\xshift+1.0},7.0);
\node at ({\xshift+0.5},\ycenter) {\LARGE $\mathcal{E}$};

\draw[channel] ({\xshift+2.6},{\ytop-0.45}) rectangle ({\xshift+3.5},{\ytop+0.45});
\node[chanlabel] at ({\xshift+3.05},\ytop) {$\widetilde{\mathcal{N}}$};

\draw[channel] ({\xshift+2.6},{\ymid-0.45}) rectangle ({\xshift+3.5},{\ymid+0.45});
\node[chanlabel] at ({\xshift+3.05},\ymid) {$\widetilde{\mathcal{N}}$};

\node at ({\xshift+3.05},4.4) {\large $\vdots$};

\draw[channel] ({\xshift+2.6},{\ybot-0.45}) rectangle ({\xshift+3.5},{\ybot+0.45});
\node[chanlabel] at ({\xshift+3.05},\ybot) {$\widetilde{\mathcal{N}}$};

\draw[procblock] ({\xshift+6.3},3.0) rectangle ({\xshift+7.5},7.0);
\node at ({\xshift+6.9},\ycenter) {\Large $\Pi_Z$};

\draw[procblock] ({\xshift+9.0},3.0) rectangle ({\xshift+10.2},7.0);
\node at ({\xshift+9.6},\ycenter) {\Large $\mathcal{D}^Z$};


\draw[qarrow] ({\xshift-0.8},\ycenter) -- (\xshift,\ycenter)
    node[wirelabel, midway] {$A'$};

\draw[qarrow] ({\xshift+1.0},\ytop) -- ({\xshift+2.6},\ytop)
    node[wirelabel, midway] {$A_1$};
\draw[qarrow] ({\xshift+1.0},\ymid) -- ({\xshift+2.6},\ymid)
    node[wirelabel, midway] {$A_2$};
\draw[qarrow] ({\xshift+1.0},\ybot) -- ({\xshift+2.6},\ybot)
    node[wirelabel, midway] {$A_n$};
\node at ({\xshift+1.8},4.4) {\large $\vdots$};

\draw[qarrow] ({\xshift+3.5},{\ytop+0.15}) -- ({\xshift+6.3},{\ytop+0.15})
    node[wirelabel, midway] {$Q_1$};
\draw[qarrow] ({\xshift+3.5},{\ymid+0.15}) -- ({\xshift+6.3},{\ymid+0.15})
    node[wirelabel, midway] {$Q_2$};
\draw[qarrow] ({\xshift+3.5},{\ybot+0.15}) -- ({\xshift+6.3},{\ybot+0.15})
    node[wirelabel, midway] {$Q_n$};

\draw[classarrow] ({\xshift+3.5},{\ytop-0.2}) -- ({\xshift+6.3},{\ytop-0.2})
    node[font=\small, below=1pt, midway] {$Z_1$};
\draw[classarrow] ({\xshift+3.5},{\ymid-0.2}) -- ({\xshift+6.3},{\ymid-0.2})
    node[font=\small, below=1pt, midway] {$Z_2$};
\draw[classarrow] ({\xshift+3.5},{\ybot-0.2}) -- ({\xshift+6.3},{\ybot-0.2})
    node[font=\small, below=1pt, midway] {$Z_n$};
\node at ({\xshift+4.9},4.4) {\large $\vdots$};

\draw[classarrow] ({\xshift+6.9},7.0) -- ({\xshift+6.9},7.3) -- ({\xshift+9.6},7.3) -- ({\xshift+9.6},7.0);
\node[font=\normalsize] at ({\xshift+8.25},7.5) {$Z$};

\draw[qarrow] ({\xshift+7.5},\ytop) -- ({\xshift+9.0},\ytop)
    node[wirelabel, midway] {$Q_1$};
\draw[qarrow] ({\xshift+7.5},\ymid) -- ({\xshift+9.0},\ymid)
    node[wirelabel, midway] {$Q_2$};
\draw[qarrow] ({\xshift+7.5},\ybot) -- ({\xshift+9.0},\ybot)
    node[wirelabel, midway] {$Q_n$};
\node at ({\xshift+8.25},4.4) {\large $\vdots$};

\draw[qarrow] ({\xshift+10.2},\ycenter) -- ({\xshift+12.0},\ycenter)
    node[wirelabel, midway] {$B'$};

\end{tikzpicture}
\caption{CQ decoding strategy. The decoder $\mathcal{D}$ (gray dashed box) first measures the classical flags $Z = (Z_1, \ldots, Z_n)$ via $\Pi_Z$, obtaining the Hamming weight $k = |Z|$, and then applies a $k$-dependent quantum decoder $\mathcal{D}^Z$ to the quantum outputs $(Q_1, \ldots, Q_n)$. This reduces the optimization from a single decoder with $d_B^n = 4^n$-dimensional input to $n+1$ decoders with $d_Q^n = 2^n$-dimensional input.}
\label{fig:CQ_decoding_strategy}
\end{figure}

Since $\mathcal{E}_n$ is symmetric, all flag patterns with the same Hamming weight $|z| = k$ produce equivalent channels up to a permutation of the positions, so the decoder need only be optimized for $n+1$ distinct channels $\mathcal{N}_k \in \mathrm{CPTP}(A \to Q)$, indexed by $k = 0, \ldots, n$:
\begin{equation}
\label{eq:channel_k}
    \mathcal{N}_k \coloneqq \mathcal{P}_{\vec{p}}^{\otimes k} \otimes \mathrm{id}^{\otimes (n-k)},
\end{equation}
where $\mathcal{P}_{\vec{p}} \coloneqq \mathcal{X}^p \circ \overline{\Delta}$ is the qubit Pauli channel with probability distribution
\begin{equation}
\vec{p} \coloneqq  \left( \frac{1}{\sqrt{2}+2}, \frac{1}{2\sqrt{2}+2}, \frac{1}{2\sqrt{2}+2}, \frac{1}{\sqrt{2}+2}\right).
\end{equation}

The channel in~\eqref{eq:channel_k} is not covariant under the full symmetric group $\mathfrak{S}_n$, but rather under the subgroup $G_k = \mathfrak{S}_k \times \mathfrak{S}_{n-k}$, which independently permutes the first $k$ positions (where $\mathcal{P}_{\vec{p}}$ acts) and the last $n-k$ (where $\mathrm{id}$ acts). More formally:
\begin{equation}
\label{eq:covariance_subgroup}
    \mathcal{N}_k \bigl(P_{A^n}(\pi)\,\rho\,P_{A^n}(\pi)^\dagger\bigr) = P_{Q^n}(\pi)\,\mathcal{N}_n(\rho)\,P_{Q^n}(\pi)^\dagger \qquad \forall\, \pi \in G_k.
\end{equation}
Therefore its Choi matrix $\Gamma^{\mathcal{N}_k}_{A^n Q^n}\in \mathrm{End}^{G_k}(A^n \otimes Q^n)$ decomposes in the $G_k$-orbit basis as:
\begin{equation}
\label{eq:choi_N_k_subgroup}
    \Gamma^{\mathcal{N}_k}_{A^n Q^n} = \sum_{s_1=1}^{m_{AQ}^{(k)}} \sum_{s_2=1}^{m_{AQ}^{(n-k)}} c_{s_1}^{\mathcal{P}_{\vec{p}}} c_{s_2}^{\mathrm{id}} \;\; C_{s_1}^{AQ^{(k)}} \otimes C_{s_2}^{AQ^{(n-k)}},
\end{equation}
where $m_{AQ}^{(k)} =\binom{k + d_A^2 d_Q^2 - 1}{k}$, $m_{AQ}^{(n-k)} =\binom{n-k + d_A^2 d_Q^2 - 1}{n-k}$, and $c_{s_1}^{\mathcal{P}_{\vec{p}}}$ and $c_{s_2}^{\mathrm{id}}$ are the orbit coefficients of $\Gamma^{\mathcal{P}_{\vec{p}}^{\otimes k}}$ and $\Gamma^{\mathrm{id}^{\otimes(n-k)}}$, computed using~\eqref{eq:tensor_product_orbit} on $\mathfrak{S}_k$ and $\mathfrak{S}_{n-k}$.
    This procedure allows to significantly reduce the involved dimensions at the decoder's SDP~\eqref{eq:block_SDP_FD}. Indeed, the decoder decomposes into $n+1$ SDPs, and the total dimension is:\begin{equation}
    m^{s}_Q \coloneqq \sum_{k=0}^n m_Q^{(k)} \cdot m_Q^{(n-k)} = \sum_{k=0}^{n}\binom{k+3}{k}\binom{n-k+3}{n-k} = \binom{n+7}{n},
\end{equation}
where in the last identity we used Vandermonde's convolution formula~\cite{Vandermonde1772}. Table~\ref{tab:CQ_dimension_comparison} shows the comparison between the different dimensions involved.
\begin{table}[ht]
\centering
\renewcommand{\arraystretch}{1.25}
\begin{tabular}{r r r r r}
\toprule
$n$ & $m_B = \binom{n+15}{n}$ & $m_Q = \binom{n+3}{n}$
    & $m_Q^s = \binom{n+7}{7}$ & $m_Q^s / m_B$ \\
\midrule
 2 &         136 &  10 &      36 & $26.5\%$ \\
 3 &         816 &  20 &     120 & $14.7\%$ \\
 4 &       3,876 &  35 &     330 & $8.5\%$ \\
 5 &      15,504 &  56 &     792 & $5.1\%$ \\
 6 &      54,264 &  84 &   1,716 & $3.2\%$ \\
 7 &     170,544 & 120 &   3,432 & $2.0\%$ \\
 8 &     490,314 & 165 &   6,435 & $1.3\%$ \\
\bottomrule
\end{tabular}
\caption{Comparison of orbit-basis dimensions for the full $d_B = 4$ output ($m_B$), the qubit CQ output ($m_Q$), and the CQ-decomposed total $m_Q^s = \sum_{k=0}^n m_k \cdot m_{n-k}$ with $m_k = \binom{k+3}{k}$. The ratio $m_Q^s / m_B$ decreases as $O(n^{-8})$, reflecting the asymptotic improvement from $O(n^{15})$ to $O(n^7)$.}
\label{tab:CQ_dimension_comparison}
\end{table}

Assuming the encoder $\mathcal{E}_n$ is symmetric, the concatenated channel $\mathcal{M}_k = \mathcal{N}_k \circ \mathcal{E}_n$ inherits $G_k$-invariance. Indeed, since $G_k \leq \mathfrak{S}_n$:\begin{equation}
    \Gamma^{\mathcal{E}_n}_{RA^n} \in \mathrm{End}^{\mathfrak{S}_n}(R \otimes A^n) \Longrightarrow \Gamma^{\mathcal{E}_n}_{RA^n}\in \mathrm{End}^{G_k}(R \otimes A^n) \quad \forall k =0, \ldots, n,
\end{equation}
and by~\eqref{eq:covariance_subgroup} we also have $\Gamma^{\mathcal{M}_k}_{RQ^n}\in \mathrm{End}^{G_k}(R \otimes Q^n)$. 
For all $k =0, \ldots, n$ the orbit bases coefficients of $\Gamma^{\mathcal{M}_k}_{RQ^n}$ can be efficiently computed straightforwardly from those of $\mathcal{N}_k$ and $\mathcal{E}_n$, and by the same twirling argument as in Lemma~\ref{lemma:seesaw_symmetric}, the optimal adjoint decoder $\mathcal{D}^{*}_k$ for each $\mathcal{M}_k$ can be restricted to satisfy $\Gamma^{\mathcal{D}_k^*}_{RQ^n}\in \mathrm{End}^{G_k}(R \otimes Q^n)$. This restriction is natural: in the decoding strategy $\mathcal{D}^Z$ of Figure~\ref{fig:CQ_decoding_strategy}, one may without loss of generality permute the output systems $(Q_1,\ldots,Q_n)$ such that the first $k$ subsystems correspond to erasure flag $z=1$ and the remaining $n-k$ to $z=0$. We now prove that such a strategy is indeed optimal for a flagged channel like $\widetilde{\mathcal{N}}$.
\begin{lemma}\label{lemma:fidelity_CQ}
Let $\rho_{ABY}$ be a quantum-classical state of the form
\begin{equation}\label{eq:QC_state}
    \rho_{ABY} = \sum_{y \in \mathcal{Y}} p(y)\, \rho_{AB}^{y} \otimes \ket{y}\!\bra{y}_Y.
\end{equation}
Then its maximal singlet fraction~\eqref{eq:maximal_singlet_fraction} satisfies
\begin{equation}
\label{eq:singlet_fraction_CQ}
    F_\Phi(\rho_{ABY}) = \sum_{y \in \mathcal{Y}} p(y)\, F_\Phi(\rho_{AB}^y),
\end{equation}
and this value is attained by a recovery map of the form
\begin{equation}\label{eq:recovery_CQ_state}
    \mathcal{N}_{BY \to \hat A} = \mathcal{R}^Y_{B \to \hat A} \circ \Pi_Y,
\end{equation}
where $\Pi_Y$ is the measurement on $Y$ in the $\{\ket{y}\}_y$ basis and $\{\mathcal{R}^y_{B \to \hat A}\}_y$ is a family of channels.
\end{lemma}

\begin{proof}
Let $\ket{\Phi}_{A\hat{A}}$ denote the MES of Schmidt rank $d_A$. Recall that $F_\Phi$ is defined as
\begin{equation}
    F_\Phi(\sigma) \coloneqq \sup_{\mathcal{N} \in \operatorname{CPTP}} \bra{\Phi}(\mathrm{id}_A \otimes \mathcal{N})(\sigma)\ket{\Phi},
\end{equation}
where the supremum is over CPTP maps from the non-$A$ subsystems of $\sigma$ to $\hat{A}$.

\emph{Lower bound.} For any family of CPTP maps $\{\mathcal{R}^y_{B \to \hat{A}}\}_{y \in \mathcal{Y}}$, the composition $\mathcal{R}^Y_{B \to \hat{A}} \circ \Pi_Y$ --- where $\Pi_Y$ is the projective measurement on $Y$ in the basis $\{\ket{y}\}_y$ and $\mathcal{R}^Y$ applies $\mathcal{R}^y$ conditioned on the classical outcome $y$ --- is itself a valid CPTP map from $BY$ to $\hat{A}$. Since $F_\Phi$ is monotone non-increasing under post-processing, restricting the supremum in the definition of $F_\Phi$ to such channels can only decrease its value:
\begin{align}
    F_\Phi(\rho_{ABY}) 
    &\geq \sup_{\{\mathcal{R}^y\}_y} \bra{\Phi}(\mathrm{id}_A \otimes (\mathcal{R}^Y_{B \to \hat{A}} \circ \Pi_Y))(\rho_{ABY})\ket{\Phi} \\
    &= \sup_{\{\mathcal{R}^y\}_y} \sum_{y \in \mathcal{Y}} p(y)\, \bra{\Phi}(\mathrm{id}_A \otimes \mathcal{R}^y_{B \to \hat{A}})(\rho_{AB}^y)\ket{\Phi}
    \label{eq:step_def_cqq_state} \\
    &= \sum_{y \in \mathcal{Y}} p(y)\, \sup_{\mathcal{R}^y \in \operatorname{CPTP}(B \to \hat{A})} \bra{\Phi}(\mathrm{id}_A \otimes \mathcal{R}^y_{B \to \hat{A}})(\rho_{AB}^y)\ket{\Phi} \\
    &= \sum_{y \in \mathcal{Y}} p(y)\, F_\Phi(\rho_{AB}^y),
\end{align}
where~\eqref{eq:step_def_cqq_state} follows from the definition~\eqref{eq:QC_state} and the action of $\Pi_Y$, and in the subsequent step we exchanged the supremum with the sum, which is justified because the $\{\mathcal{R}^y\}_y$ are independent CPTP maps with no shared constraint.

\emph{Upper bound.} Since $F_\Phi$ is convex in its input state~\cite[Chapter 7]{skrzypczyk2023semidefinite}, we have
\begin{align}
    F_\Phi(\rho_{ABY}) 
    &= F_\Phi\!\left(\sum_{y \in \mathcal{Y}} p(y)\, \rho_{AB}^y \otimes \ket{y}\!\bra{y}_Y\right) \\
    &\leq \sum_{y \in \mathcal{Y}} p(y)\, F_\Phi(\rho_{AB}^y \otimes \ket{y}\!\bra{y}_Y)
    \label{eq:step_convexity} \\
    &= \sum_{y \in \mathcal{Y}} p(y) \sup_{\mathcal{N}_{BY \to \hat{A}}} \bra{\Phi}(\mathrm{id}_A \otimes \mathcal{N}_{BY \to \hat{A}})(\rho_{AB}^y \otimes \ket{y}\!\bra{y}_Y)\ket{\Phi}
    \label{eq:recovery_map_definition} \\
    &= \sum_{y \in \mathcal{Y}} p(y) \sup_{\mathcal{R}^y \in \operatorname{CPTP}(B \to \hat{A})} \bra{\Phi}(\mathrm{id}_A \otimes \mathcal{R}^y_{B \to \hat{A}})(\rho_{AB}^y)\ket{\Phi}
    \label{eq:step_reduction} \\
    &= \sum_{y \in \mathcal{Y}} p(y)\, F_\Phi(\rho_{AB}^y).
\end{align}
Here~\eqref{eq:step_convexity} uses convexity,~\eqref{eq:recovery_map_definition} is the definition of $F_\Phi$, and~\eqref{eq:step_reduction} follows because for every fixed $y \in \mathcal{Y}$, since the $Y$ register is in a fixed pure state $\ket{y}\!\bra{y}_Y$, the action of any joint CPTP map $\mathcal{N}_{BY \to \hat{A}}$ effectively reduces to a CPTP map $\mathcal{R}^y_{B \to \hat{A}}$ acting on system $B$, defined as $\mathcal{R}^y_{B \to \hat{A}}(\cdot) \coloneqq \mathcal{N}_{BY \to \hat{A}}(\cdot \otimes \ket{y}\!\bra{y}_Y)$.

Combining the two bounds yields~\eqref{eq:singlet_fraction_CQ}, and the lower-bound derivation shows that the optimum is attained by a recovery map of the form~\eqref{eq:recovery_CQ_state}.
\end{proof}

\begin{corollary}
\label{corollary:modified_seesaw_recovery_coeff}
    Given the $n$-fold tensor product channel $\widetilde{\mathcal{N}}^{\otimes n}$, with $\widetilde{\mathcal{N}}$ defined in~\eqref{eq:equivalent_channel}, and given a symmetric encoder $\mathcal{E}_n$, the decoding strategy of Figure~\ref{fig:CQ_decoding_strategy} is optimal, i.e.:\begin{equation}
    \label{eq:modified_seesaw_recovery_coeff}
        F_D(\mathcal{M}) = \frac{1}{2^n}\sum_{k = 0}^n \binom{n}{k} F_D(\mathcal{M}_k),
    \end{equation}
    where $\mathcal{M}= \widetilde{\mathcal{N}}^{\otimes n} \circ \mathcal{E}$ and $\mathcal{M}_k = \mathcal{N}_k \circ \mathcal{E}$.
\end{corollary}

\begin{proof}
 It is a direct consequence of Lemma~\ref{lemma:fidelity_CQ} and the link between the fidelity of recovery and the maximal singlet fraction in~\eqref{eq:recovery_as_singlet_fraction}. Indeed, the Choi state of the channel $\mathcal{M}$ is the quantum-classical state $\Phi^{\mathcal{M}}_{RQ^nZ^n}$, so that:
 \begin{equation}
    \label{eq:maximal_recovery_coefficient_CQ}
        F_D(\mathcal{M}) = F_{\Phi}(\Phi_{RQ^nZ^n}^{\mathcal{M}}), \ \ \ \ F_D(\mathcal{M}_k) = F_{\Phi}(\Phi_{RQ^n}^{\mathcal{M}_k}).
    \end{equation}
Using the correspondence $A \coloneqq R, Y \eqqcolon  Z^n$ and $B \eqqcolon  Q^n$, we can substitute in~\eqref{eq:singlet_fraction_CQ}  the states $\Phi_{RQ^n}^{\mathcal{M}_k} \eqqcolon  \rho_{AB}^y$ and $\Phi_{RQ^nZ^n}^{\mathcal{M}} \eqqcolon  \rho_{ABY}$, so that~\eqref{eq:modified_seesaw_recovery_coeff} follows.
\end{proof}
Note that in our case, since the channel $\mathcal{M}$ is symmetric, the probability distribution takes the simple form $\left\{\frac{1}{2^n}\binom{n}{k}\right\}_{k = 0}^{n}$, where the term $1/2^n$ accounts for the 50\%  erasure probability in $\widetilde{\mathcal{N}}$ and the binomial coefficient $\binom{n}{k}$ weights the contributions of the maximal recovery coefficients $F_D(\mathcal{M}_k)$ based on the occurrences of $k$ erasures out of $n$ uses.

Building on the results of the previous section, we can efficiently solve the SDP for the $k$-th decoder, using the fact that the irreducible representations of $G_k = \mathfrak{S}_k \times \mathfrak{S}_{n-k}$ are tensor products $\lambda_k \boxtimes \lambda_{n-k}$, and (\cite[Lemma 2.5]{CTV2023}) the representative set for $R \otimes Q^k \otimes Q^{n-k}$ is $\{\mathbbm{1}_R \otimes U^Q_{\lambda_k} \otimes U^Q_{\lambda_{n-k}}\}$, and for every $k$, the isomorphism $\widetilde{\psi_k}$ maps each operator to a collection of blocks indexed by pairs $(\lambda_k, \lambda_{n-k})$, each of size $d_R \cdot m_{\lambda_k} \cdot m_{\lambda_{n-k}}$. All the formulas from representation theory of the previous section extend straightforwardly with the substitution $\lambda \to (\lambda_k, \lambda_{n-k})$ and the multiplicity $f_\lambda \to f_{\lambda_k} \cdot f_{\lambda_{n-k}}$. 

Thus, for every $k \in \{0, \ldots, n\}$, the decoder optimization in~\eqref{eq:block_SDP_FD} is mapped to:
\begin{equation}
\label{eq:SDP_k_block}
\boxed{
\begin{aligned}
    \widetilde{\Psi_k}(F^S_{D}(\mathcal{M}_k)) = &\max_{\widetilde{[\Gamma^{\mathcal{D}_k^*}]}_{(\lambda_k, \lambda_{n-k})}} \quad  \frac{1}{d^2} \sum_{(\lambda_k, \lambda_{n-k})} f_{\lambda_k}\, f_{\lambda_{n-k}}\, \mathrm{Tr}\!\left[\widetilde{[\Gamma^{\mathcal{D}_k^*}]}_{(\lambda_k, \lambda_{n-k})} \cdot \widetilde{[\Gamma^{\mathcal{M}_k}]}_{(\lambda_k, \lambda_{n-k})}\right] \\
    \text{s.t.} \quad & \widetilde{[\Gamma^{\mathcal{D}_k^*}]}_{(\lambda_k, \lambda_{n-k})} \geq 0, \qquad \forall\, (\lambda_k, \lambda_{n-k}) \in \mathrm{Par}(d_Q, k)\times  \mathrm{Par}(d_Q, n-k) \\
    & \sum_{i=1}^{d_R} \widetilde{[\Gamma^{\mathcal{D}_k^*}]}_{(\lambda_k, \lambda_{n-k})}^{(i,i)} = \mathbbm{1}_{m_{\lambda_k} \cdot m_{\lambda_{n-k}}},  \quad \forall\, (\lambda_k, \lambda_{n-k}) \in \mathrm{Par}(d_Q, k)\times  \mathrm{Par}(d_Q, n-k)
\end{aligned}}
\end{equation}
where the sums run over $\lambda_k \in \mathrm{Par}(d_Q, k)$ and $\lambda_{n-k} \in \mathrm{Par}(d_Q, n-k)$. For $k = 0$ and $k = n$, one of the two factors is trivial, and the SDP reduces to the standard $\mathfrak{S}_n$-symmetric case of~\eqref{eq:block_SDP_FD}.

The block-basis power iteration adapts to the subgroup case by treating each $(\lambda_k, \lambda_{n-k})$ block as a single block of size $d_R \cdot m_{\lambda_k} \cdot m_{\lambda_{n-k}}$. The three steps of Definition~\ref{def:symmetric_power_method_D}, namely, apply, normalize, and convergence check, proceed in the exact same way, with the only difference that now blocks are labeled by pairs of partitions $(\lambda_k, \lambda_{n-k})$.
The $n+1$ decoder optimizations are completely independent and can in principle be solved in parallel, leading to $n+1$ fidelities of recovery, which can be combined using~\eqref{eq:modified_seesaw_recovery_coeff} to get the first estimate of the fidelity at the first step of the symmetric seesaw method.

After solving the $n+1$ decoder SDPs~\eqref{eq:SDP_k_block}, we obtain decoders $\{\mathcal{D}_k\}_{k=0}^n$, each residing in a different $G_k$-invariant space. To perform the encoder optimization, we need to reconstruct the effective channel seen by Alice when the flag pattern is unknown.

For a fixed Hamming weight $k$, the decoder $\mathcal{D}_k$ is optimized for the channel $\mathcal{N}_k = \mathcal{P}_{\vec{p}}^{\otimes k} \otimes \mathrm{id}^{\otimes(n-k)}$, which places all errors in the first $k$ positions. However, from the encoder's perspective, any of the $\binom{n}{k}$ flag patterns with Hamming weight $k$ is equally likely. For each such pattern, the decoder first reorders the systems so that all errors are at the front and then applies $\mathcal{D}_k$, so the effective channel seen by Alice is as follows:
\begin{equation}
\label{eq:effective_channel_alice}
    \overline{\mathcal{M}'} = \frac{1}{2^n}\sum_{k=0}^n \sum_{\pi \in \mathfrak{S}_n / G_k}\mathcal{M}_k' \circ \mathcal{P}_\pi,
\end{equation}
where $\mathcal{P}_\pi(\rho) = P(\pi)\,\rho\,P(\pi)^\dagger$ permutes the input systems and $\mathcal{M}_k' = \mathcal{D}_k \circ \mathcal{N}_k$. Each concatenated channel $\mathcal{M}'_k = \mathcal{D}_k \circ \mathcal{N}_k$ is $G_k$-invariant, so its Choi matrix resides in the space $\mathrm{End}^{G_k}(A^n \otimes R)$, which is \emph{larger} than $\mathrm{End}^{\mathfrak{S}_n}(A^n \otimes R)$. Therefore, the $G_k$-orbit basis is finer than the $\mathfrak{S}_n$-orbit basis, and to express the effective channel $\overline{\mathcal{M}'}$ in the $\mathfrak{S}_n$-orbit basis required by the encoder's SDP~\eqref{eq:block_SDP_FE}, we must \emph{re-symmetrize} each $\mathcal{M}_k$ by averaging over the coset $\mathfrak{S}_n / G_k$. This can be done efficiently, using a simple combinatorial formula.

Since $\Gamma^{\overline{\mathcal{M}'}}_{RA^n}\in \mathrm{End}^{\mathfrak{S}_n}(R \otimes A^{\otimes n})$, the encoder optimization proceeds exactly as in~\eqref{eq:block_SDP_FE}. This completes one iteration in this modified symmetric seesaw method. Combining the results of this section, we can now state the full modified symmetric seesaw algorithm for the flagged channel $\widetilde{\mathcal{N}}$.

\begin{definition}
\label{def:CQ_seesaw}
Given the CQ channel $\widetilde{\mathcal{N}}$ of Proposition~\ref{prop:equivalent_channel}, a number of channel uses $n \in \mathbb{N}$, an input dimension $d$, and a convergence threshold $\delta > 0$, the modified symmetric seesaw method produces a lower bound $\tilde{F} \leq F_c(\widetilde{\mathcal{N}}^{\otimes n}, d)$ together with a symmetric encoder $\mathcal{E}$ and a family of symmetric decoders $\{\mathcal{D}_k\}_{k=0}^n$. The procedure is as follows.
\begin{enumerate}
    \item \textbf{Initialization.} Choose a random symmetric encoder $\mathcal{E}_0$ with $ \Gamma^{\mathcal{E}_0}_{RA^n}
        \in \mathrm{End}^{\mathfrak{S}_n}(R \otimes A^{\otimes n})$, and, for each $k = 0, \ldots, n$, a symmetric decoder $\mathcal{D}_{k,0}$ such that $\Gamma^{\mathcal{D}_{k,0}^*}_{RQ^n} \in \mathrm{End}^{G_k}(R \otimes Q^{\otimes n}).$ Set $i = 0$.

    \item \textbf{Iteration.} Repeat the following steps:

    \begin{enumerate}
        \item \textit{Decoder phase.}
        For every $k = 0, \ldots, n$:
        \begin{enumerate}
            \item Compute the coefficients of $\mathcal{M}_{k,i} = \mathcal{N}_k \circ \mathcal{E}_i$ in the $G_k$-orbit basis.
            \item Optimize the decoder $\mathcal{D}_{k,i}$ by solving the SDP~\eqref{eq:SDP_k_block} or applying the symmetric power method (Definition~\ref{def:symmetric_power_method_D}), obtaining $\mathcal{D}_{k,i+1}$ and $F_{D}(\mathcal{M}_{k,i})$.
        \end{enumerate}

        The $n+1$ decoder optimizations are independent. Define the weighted decoder fidelity
        \begin{equation}
            F^S_{D}(\overline{\mathcal{M}_{i}}) = \frac{1}{2^n} \sum_{k=0}^n \binom{n}{k} F^S_{D}(\mathcal{M}_{k,i}).
        \end{equation}
        \item \textit{Encoder phase.}
        \begin{enumerate}
            \item For each $k$, compute the coefficients of $\mathcal{M}'_{k,i} = \mathcal{D}_{k,i+1} \circ \mathcal{N}_k$ in the $G_k$-orbit basis.
            \item Aggregate the channels $\{\mathcal{M}'_k\}_k$ into the effective channel $\overline{\mathcal{M}_i'}$ in the $\mathfrak{S}_n$-orbit basis via re-symmetrization.
            \item Optimize the encoder over $\mathrm{End}^{\mathfrak{S}_n}(R \otimes A^{\otimes n})$ by solving~\eqref{eq:block_SDP_FE} or applying the symmetric power method, yielding $\mathcal{E}_{i+1}$ and $F^S_{E}(\overline{\mathcal{M}_i'})$.
        \end{enumerate}

        \item \textit{Convergence check.}
        If $F^S_{D}(\overline{\mathcal{M}_{i}}) - F^S_{E}(\overline{\mathcal{M}_i'}) < \delta$, terminate. Otherwise set $i \leftarrow i + 1$ and repeat.
    \end{enumerate}

    \item \textbf{Output.}
    Return $\tilde{F} \coloneqq  F^S_{E}(\overline{\mathcal{M}'})$ together with the code:\begin{equation}
    \label{eq:outcomes_modified_seesaw}
        (\mathcal{E}, \{\mathcal{D}_k\}_{k=0}^n) \equiv (\mathcal{E}_{i+1}, \{\mathcal{D}_{k,i+1}\}_{k=0}^n).
    \end{equation} 
\end{enumerate}
\end{definition}
The asymmetry between the two phases reflects the flagged structure of the channel: for every $n \in \mathbbm{N}$, the decoder specializes to each Hamming weight $k=0, \ldots, n$ (exploiting the $G_k$ symmetry), while the encoder must average over all flag patterns (requiring the $\mathfrak{S}_n$-re-symmetrization).
\subsubsection{Exploiting the Sparsity of Channels}
\label{sec:CLDUI}
The second bottleneck identified above is the serial concatenation step. Even though the number of joint orbits $m_{AQ} = \binom{n + d_A^2 d_Q^2 - 1}{n} = O(n^{15})$ (for $d_A = d_Q = 2$) is polynomial in $n$, it grows rapidly and becomes a bottleneck for large~$n$, both in memory and in time. We now show that the sparsity of the specific channels $\mathcal{N}_k$ involved allows us to restrict all computations to a much smaller set of nonzero orbits.
\begin{lemma}
\label{lem:tensor_product_sparsity}
Let $X \in \mathcal{L}(\mathcal{H})$ have support $\mathcal{A} \coloneqq \{(a,b) : X_{a,b} \neq 0\} \subseteq [d_\mathcal{H}]^2$. Then at most $\binom{n + |\mathcal{A}| - 1}{n} = O(n^{|\mathcal{A}|-1})$ orbits contribute to $X^{\otimes n}$ in the orbit basis~\eqref{eq:tensor_product_orbit}.
\end{lemma}
\begin{proof}
By~\eqref{eq:coeff_tensor_product}, the coefficient of an orbit with representative $(\underline{i}, \underline{j})$ is $\prod_{k=1}^n X_{i_k, j_k}$, which vanishes as soon as any pair $(i_k, j_k) \notin \mathcal{A}$. Nonzero orbits are therefore in bijection with multisets of size $n$ drawn from $\mathcal{A}$, of which there are $\binom{n + |\mathcal{A}| - 1}{n}$.
\end{proof}

The Choi matrices of the two single-copy channels of interest,
\begin{equation}
    \Gamma^{\mathrm{id}}_{AQ} = \begin{pmatrix}
        1 & 0 & 0 & 1 \\ 0 & 0 & 0 & 0 \\ 0 & 0 & 0 & 0 \\ 1 & 0 & 0 & 1
    \end{pmatrix}, \qquad
    \Gamma^{\mathcal{P}_{\vec{p}}}_{AQ} = \begin{pmatrix}
        1{-}p & 0 & 0 & 0 \\ 0 & p & 0 & 0 \\ 0 & 0 & p & 0 \\ 0 & 0 & 0 & 1{-}p
    \end{pmatrix},
\end{equation}
each have exactly four nonzero entries. (Both are instances of conjugate local diagonal unitary invariant (CLDUI) matrices~\cite{Johnston2018, Singh2021Diagonal}: the former is an isotropic state~\cite{Horodecki1999Reduction}, the latter is Bell-diagonal.) Since the orbit basis factorizes across the two tensor factors of $\mathcal{N}_k = \mathcal{P}_{\vec{p}}^{\otimes k} \otimes \mathrm{id}^{\otimes(n-k)}$, Lemma~\ref{lem:tensor_product_sparsity} applies independently to each, bounding the number of nonzero orbits in $\Gamma^{\mathcal{N}_k}$ by
\begin{equation}
    \binom{k+3}{k} \cdot \binom{n-k+3}{n-k} = O(n^6),
\end{equation}
compared to $O(n^{30})$ in the generic case; at the boundary $k \in \{0, n\}$ this reduces further to $O(n^3)$ versus $O(n^{15})$. Concretely, this sparsity is exploited at two points in the algorithm: (i) in the \emph{initialization step}, only nonzero orbits of $\Gamma^{\mathcal{P}_{\vec{p}}^{\otimes k}}$ and $\Gamma^{\mathrm{id}^{\otimes(n-k)}}$ need to be enumerated and stored; (ii) in the \emph{concatenation step}, sums over joint orbits $s$ can be restricted to those with $c_s^{\mathcal{N}_k} \neq 0$. 

Combining this sparsity with the decoding strategy of the previous section yields the algorithm used to produce the results of Figure~\ref{fig:main}.



\end{document}